%% file: paper.tex
\documentclass{llncs}
\usepackage{times}
\usepackage{amsmath}
\usepackage{amssymb}
\usepackage{graphicx}
\usepackage{picinpar}
\title{Simulating Parity Reasoning (extended version)\thanks{The original version of the paper has been accepted to 19th International Conference on Logic for Programming Artificial Intelligence and Reasoning, LPAR-19. The extended version contains proofs and an additional section ``Connection to Treewidth''}}
\author{Tero Laitinen \and Tommi Junttila \and Ilkka Niemel\"a}
\institute{Aalto University\\
Department of Information and Computer Science\\
PO Box 15400, FI-00076 Aalto, Finland\\
\email{{Tero.Laitinen,Tommi.Junttila,Ilkka.Niemela}@aalto.fi}
}

\input{macros}

\begin{document}

\maketitle

\begin{abstract} 
\input{abstract}

%Resolution, which is equivalent to the underlying proof system of modern
%conflict-driven SAT solvers, is shown to simulate equivalence reasoning.
%
%Parity explanations on non-deterministic unit propagation derivations is shown
%to simulate Gauss-Jordan elimination.
%

%
%It is shown that equivalence reasoning can be simulated by
%resolution.  It is proven that Gaussian elimination can be simulated by
%non-deterministic unit propagation and parity explanations, and also by
%plain unit propagation by adding additional parity constraints.
%We also show that unit propagation can polynomially simulate
%Gaussian elimination if the xor-instance has a bounded treewidth.
%We also show that unit propagation cannot polynomially simulate
%Gaussian elimination.
\end{abstract}

%\todo{Abstract submission April 17, submission April 22, full paper submissions are limited to 15 LNCS pages plus references}

\input{introduction}

\input{preliminaries}
\input{eqreasoning}
\input{parityexp}

\input{xupify}

\input{treewidth}

%\input{decomposition}
%\input{gesimulation}
\input{conclusions}

\subsubsection*{Acknowledgments.}
This work has been financially supported by the Academy  of Finland under the Finnish Centre of Excellence in Computational Inference (COIN).
We acknowledge the computational resources provided by Aalto Science-IT project.

%\bibliographystyle{named}
\bibliographystyle{splncs}
\bibliography{paper}

\input{proofs}

\input{proofs2}
\end{document}

%% file: macros.tex
\newcommand{\tmpfname}[1]{f_{#1}}
\newcommand{\tmpf}[2]{\tmpfname{#2}(#1)}

\newcommand{\Models}{\ensuremath{\models}}
\newcommand{\NotModels}{\ensuremath{\not\models}}

\newcommand{\Equal}{\equiv}

\newcommand{\defrule}[2]{
\begin{tabular}{@{}c@{}}
        \ensuremath{#1} \\
    \hline
        \ensuremath{#2} 
    \end{tabular}}

\newcommand{\inferencerule}[3]{\defrule{{#1}\hspace{5mm}{#2}}{#3}}

\newcommand{\simplification}[3]{\ensuremath{#1 \left[ \substitution{#2}{#3}\right] } }
\newcommand{\substitution}[2]{#1/#2}

\newcommand{\SUBST}{{\ensuremath{\textup{\small\textsf{Subst}}}}}

\newcommand{\UP}{\ensuremath{\textup{\textsf{\small{UP}}}}}
\newcommand{\UPderiv}{\ensuremath{\mathrel{\vdash_{\textup{\textsf{\scriptsize{UP}}}}}}}
\newcommand{\NotUPderiv}{\ensuremath{\mathrel{\not\vdash_{\textup{\textsf{\scriptsize{UP}}}}}}}

\newcommand{\parity}[1]{\ensuremath{p_{#1}}}

\newcommand{\Parityp}{\ensuremath{p'}}
\newcommand{\Paritypp}{\ensuremath{p''}}

\newcommand{\igraph}{G}
\newcommand{\vertices}{V}
\newcommand{\edges}{E}
\newcommand{\Vertex}{v}
\newcommand{\Vertexp}{v'}
\newcommand{\Vertexpp}{v''}

\newcommand{\TA}{\tau}
\newcommand{\Cut}{W}

\newcommand{\Vars}{X}
\newcommand{\CutA}{V_{\textup{a}}}
\newcommand{\CutB}{V_{\textup{b}}}
\newcommand{\unitruleP}{\ensuremath{\oplus\mbox{-Unit}^+}}
\newcommand{\unitruleN}{\ensuremath{\oplus\mbox{-Unit}^-}}
\newcommand{\eqvruleP}{\ensuremath{\oplus\mbox{-Eqv}^+}}
\newcommand{\eqvruleN}{\ensuremath{\oplus\mbox{-Eqv}^-}}
\newcommand{\CExpl}[1]{\mathit{Expl}(#1)}

\newcommand{\PExpl}[1]{\mathit{Expl_{\X}}(#1)}

\newcommand{\cnf}[1]{\operatorname{cnf}(#1)}

\newcommand{\set}[1]{\left\{ #1 \right\}}
\newcommand{\Set}[1]{\set{#1}}

\newcommand{\Setdef}[2]{\left\{{#1}\mid{#2}\right\}}
\newcommand{\Tuple}[1]{\ensuremath{\left\langle{#1}\right\rangle}}

\newcommand{\Card}[1]{\ensuremath{{\left|{#1}\right|}}}

\newcommand{\Tool}[1]{\textsf{\small #1}}

\newcommand{\Minisat}{\Tool{minisat}}
\newcommand{\Booleans}{\mathbb{B}}

\newcommand{\T}{\top}
\newcommand{\F}{\bot}
\newcommand{\Implies}{\Rightarrow}

\newcommand{\X}{\oplus}
\newcommand{\XX}{\,{\oplus}\,}
\newcommand{\Def}{\mathrel{:=}}

\newcommand{\VarsOf}[1]{\operatorname{vars}(#1)}

\newcommand{\ClausesOf}[2]{\operatorname{clauses}(#1, #2)}

\newcommand{\LitsOf}[1]{\operatorname{lits}(#1)}
\newcommand{\SetLC}[1]{\sum_{\XC \in #1} \XC}
\newcommand{\LCVarsGone}[1]{\operatorname{vars}^{\times}(#1)}
\newcommand{\orpart}{\phi_\textup{or}}
\newcommand{\xorpart}{\phi_\textup{xor}}
\newcommand{\xorclauses}{\phi_\textup{xor}}

\newcommand{\XC}{D}
\newcommand{\XCp}{D'}

\newcommand{\XCB}{E}
\newcommand{\ECsys}{\ensuremath{\textup{\textsc{EC}}}}

\newcommand{\IL}{\hat{l}}

\newcommand{\AL}{l}

\newcommand{\Labfunc}{L}
\newcommand{\Lab}[1]{\Labfunc(#1)}

\newcommand{\While}{\KW{while}}
\newcommand{\For}{\KW{for}}

\newcommand{\If}{\KW{if}}

\newcommand{\Return}{\KW{return}}

\newcommand{\KW}[1]{\textsf{\footnotesize{}{#1}}}

%% file: abstract.tex
Propositional satisfiability (SAT) solvers, which typically operate using
conjunctive normal form (CNF), have been successfully applied in many domains. 
However, in some application areas such as circuit verification, bounded model
checking, and logical cryptanalysis, instances can have many parity (xor)
constraints which may not be handled efficiently if translated to CNF. 
Thus, extensions to the CNF-driven search with various parity 
reasoning engines ranging from equivalence reasoning to incremental
Gaussian elimination have been proposed.
This paper studies how stronger parity reasoning techniques in the DPLL(XOR)
framework can be simulated by simpler systems: resolution, unit
propagation, and parity explanations. 
Such simulations are interesting, for example, for developing the next
generation SAT solvers capable of handling parity constraints efficiently.

%% file: introduction.tex
\section{Introduction}

Propositional satisfiability (SAT) solver technology has developed rapidly
providing a powerful solution technique in many industrial application domains
(see e.g.~\cite{Handbook:CDCL}).
The efficiency of SAT solvers is partly due to efficient data structures and
algorithms that allow very efficient Boolean constraint propagation and
conflict-driven clause learning in conjunctive normal form (CNF).
Straightforward Tseitin-translation~\cite{tseitin} of a problem instance to CNF
may result in poor performance, especially in the case of parity (xor)
constraints, that can be abundant in applications such as circuit
verification, bounded model checking, and logical cryptanalysis.
Although pure parity constraints (linear arithmetic modulo two) can be
efficiently solved with Gaussian elimination, they can be very difficult for
resolution \cite{Urquhart:JACM1987} and thus for state-of-the-art
conflict-driven clause learning (CDCL) satisfiability solvers as their
underlying proof system is equivalent to
resolution~\cite{PipatsrisawatDarwiche:AI2011}.
Due to this inherent hardness of parity constraints, several approaches to
combining CNF-level and xor-constraint reasoning have been proposed
~\cite{Li:AAAI2000,Li:IPL2000,BaumgartnerMassacci:CL2000,Li:DAM2003,HeuleMaaren:SAT2004,HeuleEtAl:SAT2004,Chen:SAT2009,SoosEtAl:SAT2009,LJN:ECAI2010,Soos,LJN:ICTAI2011,LJN:SAT2012,LJN:CP2012,LJN:ICTAI2012} (see~\cite{Weaver:2012:SAE:2520447} for an alternative state-based approach).
In these approaches, CNF-driven search has been extended with various parity
reasoning techniques, ranging from plain unit propagation via equivalence
reasoning to Gaussian elimination.
Stronger parity reasoning may prune the search space effectively but often at
the expense of high computational overhead, so resorting to simpler but more
efficiently implementable systems, e.g. unit propagation, may lead to better
performance.

In this paper, we study to what extent such simpler systems can simulate
stronger parity reasoning engines in the DPLL(XOR)
    framework~\cite{LJN:ECAI2010}. The DPLL(XOR), similar to the DPLL($T$)
    approach~\cite{NieuwenhuisEtAl:JACM06} to Satisfiability Modulo Theories,
    is a framework to integrate a parity reasoning engine to a CDCL SAT solver.
The aim is to offer generalizable results that provide a foundation for
developing techniques to handle xor-constraints in next generation SAT solvers.
Instead of developing yet another propagation engine and assessing it through
an experimental comparison we believe that useful insights can be acquired by
considering unanswered questions on how some existing propagation engines and
proof systems relate to each other on a more fundamental level.
Several experimental studies have already shown that SAT solvers extended with
different parity reasoning engines can outperform unmodified solvers on some
instance families, so we focus on more general results on the relationships
between resolution, unit propagation, equivalence reasoning, parity
explanations, and Gauss-Jordan elimination, which is a complete parity
reasoning technique.

We show that resolution can simulate equivalence reasoning efficiently, which
raises a question whether significant reductions in solving time can be gained
by integrating specialized equivalence reasoning in a SAT solver since in
theory it does not strengthen the underlying proof system of the SAT solver. 
In practice, though, the performance of the SAT solver is largely governed by variable selection and other heuristics that are likely to be non-optimal, which may justify the pragmatic use of equivalence reasoning.

Although equivalence reasoning alone is not enough to cross the ``exponential gap'' between resolution and Gauss-Jordan elimination, another light-weight parity reasoning technique comes intriguingly close at simulating complete parity reasoning. We show that parity explanations, an efficiently implementable conflict explanation
technique, on nondeterministic unit propagation derivations can simulate
Gauss-Jordan elimination on a restricted yet practically relevant class of
xor-constraint conjunctions. 
Choosing assumptions and unit propagation steps nondeterministically may not be
possible in an actual implementation with greedy propagation strategies.
However, we present further experimental results indicating that the
simulation may still work in an actual implementation to some degree
provided that parity explanations are stored as learned xor-constraints
as described in~\cite{LJN:SAT2012}.

Additional xor-constraints can also be added to the formula in a preprocessing
step in order to enable unit propagation to deduce more implied literals, which
has the benefit of not requiring modifications to the SAT solver.
We present a translation that enables unit propagation to simulate parity
reasoning systems stronger than equivalence reasoning through the use of
additional xor-constraints on auxiliary variables.
The translation takes into account the structure of the original conjunction of
xor-constraints and can produce compact formulas for sparsely connected
instances.
Using the translation to simulate full Gauss-Jordan elimination with plain unit
propagation requires an exponential number of additional xor-constraints in the
worst case, but we show that the translation is polynomial for instance
families of bounded treewidth.
Recently, it has been shown in~\cite{Kullmann:Sep2013} that a conjunction of
xor-constraints does not have a polynomial-size ``arc consistent''
CNF-representation, which implies it is not feasible to simulate Gauss-Jordan
elimination by unit propagation in the general case. 
On many instances, though, better solver performance can be obtained by
simulating a weaker parity reasoning system as it reduces the size of the
translation substantially.
By applying our previous results on detecting whether unit propagation or
equivalence reasoning is enough to deduce all implied literals, the size of the
translation can be optimized further.
The experimental evaluation on a challenging benchmark set suggests that the
translation can lead to significant reduction in the solving time for some
instances.

%
%The presented reduction from full Gauss-Jordan elimination to plain unit
%propagation, which is exponential in the worst case, is shown to be polynomial
%for instance families of bounded treewidth. Finding the exact value of
%treewidth is an NP-complete problem~\cite{Arnborg:1987:CFE:37170.37183}, so we apply the junction tree algorithm~\cite{DBLP:conf/aaai/Pearl82}
%to
%get an upper bound for treewidth for some SAT Competition benchmark instances. 
%

%The translation that produces a parity reasoning ``simulation formula'', a
%conjunction of xor-constraints, allows for an interesting space-time trade-off:
%simulated parity reasoning can be limited to different classes of parity
%reasoning stronger than equivalence reasoning but weaker than full Gauss-Jordan
%%elimination. 
%
%The resulting ``simulation formula'' can contain xor-constraints that are
%redundant in terms of propagation. We present a simplification technique that
%can remove some of these redundant xor-constraints.
%
%The experimental results suggest that for some instances a significant
%reduction in the number of decisions can be achieved and this may be reflected
%in the solving time.
%
%
The proofs of lemmas and theorems are in the appendix.
%The full version of the paper including the proofs of lemmas and theorems is
%available at \url{http://users.ics.aalto.fi/tolaiti2/lpar2013.pdf}.

%% file: preliminaries.tex
\section{Preliminaries}

\newcommand{\Var}{x}
\newcommand{\AnotherVar}{y}
\newcommand{\ThirdVar}{z}

Let $\Booleans = \Set{\F,\T}$ be the set of truth values ``false'' and ``true''.
A literal is a Boolean variable $x$ or its negation $\neg x$
(as usual, $\neg \neg x$ will mean $x$),
and a clause is a disjunction of literals.
If $\phi$ is any kind of formula or equation,
(i) $\VarsOf{\phi}$ is the set of variables occurring in it,
(ii) $\LitsOf{\phi} = \Setdef{x,\neg x}{x\in\VarsOf{\phi}}$ is the set of literals over $\VarsOf{\phi}$,
and
(iii) a truth assignment for $\phi$ is a, possibly partial,
function $\TA : \VarsOf{\phi} \to \Booleans$.
A truth assignment satisfies (i) a variable $x$ if $\TA(x)=\T$,
(ii) a literal $\neg x$ if $\TA(x)=\F$, and
(iii) a clause $(l_1 \lor .. \lor l_k)$ if it satisfies at least one literal $l_i$ in the clause.

\paragraph{Resolution.}
\newcommand{\ResDer}{\pi}
\newcommand{\RP}{\pi_{r}}
\newcommand{\RC}{\hat{C}}
\newcommand{\cnfformula}{\phi}
Given two clauses, $x \lor C$ and ${\neg x} \lor D$
for arbitrary disjunctions of literals $C$ and $D$,
their resolvent is $C \lor D$.
Given a CNF formula $\cnfformula$, % = {C_1 \land ... \land C_n}$,
a resolution derivation on $\cnfformula$ is a
finite sequence $\ResDer = \RC_1 \RC_2 ... \RC_m$ of clauses such that
for all $1 \le i \le m$ it holds that either
(i) $\RC_i$ is a clause in $\cnfformula$,
or
(ii) $\RC_i$ is the resolvent of two clauses, $\RC_j$ and $\RC_k$, in $\ResDer$ with $1 \le j,k < i$.
A clause $C$ is resolution derivable from $\cnfformula$
if there is resolution derivation on $\cnfformula$ including $C$.
The formula $\cnfformula$ is unsatisfiable if and only if
the empty clause is resolution derivable from $\cnfformula$.

\paragraph{Xor-constraints.}
An \emph{xor-constraint} is an equation of the form $x_1 \X ... \X x_k \Equal p$,
where the $x_i$s are Boolean variables
and
$p \in \Booleans$ is the parity.\footnote{The correspondence of xor-constraints to the ``xor-clause'' representation used e.g.~in \cite{LJN:ECAI2010,LJN:ICTAI2011,LJN:SAT2012} is straightforward: $x_1 \X ... \X x_k \Equal \T$ corresponds to the xor-clause $(x_1 \X ... \X x_k)$ and $x_1 \X ... \X x_k \Equal \F$ to $(x_1 \X ... \X x_k \X \T)$.}
We implicitly assume that duplicate variables are always removed from
%the left-hand side of the equation,
the equations,
e.g.~$x_1 \X x_2 \X x_1 \X x_3 \Equal \T$ is always simplified into
$x_2 \X x_3 \Equal \T$.
If the left hand side does not have variables, then it equals to $\F$;
the equation $\F \Equal \T$ is a contradiction and $\F \Equal \F$ a tautology.
We identify the xor-constraint $\Var \Equal \T$ with the literal $\Var$,
$\Var \Equal \F$ with $\neg\Var$,
$\F \Equal \F$ with $\T$,
and
$\T \Equal \F$ with $\F$.
A truth assignment $\TA$ satisfies an xor-constraint $x_1 \X ... \X x_k \Equal p$ if $\TA(x_1) \X ... \X \TA(x_k) = p$.
We use $\simplification{\XC}{x}{Y}$ to denote the xor-constraint
obtained from $\XC$ by substituting the variable $x$ in it with $Y$.
%
%identical to $C$ except that all occurrences of the variable $x$ in $C$
%are substituted with $Y$ once.
%
For instance,
$\simplification{(x_1 \X x_2 \X x_3 \Equal \T)}{x_1}{x_2 \oplus \T} =
{x_2 \oplus \T \oplus x_2 \oplus x_3 \Equal \T} =
{x_3 \Equal \F}$.
The straightforward CNF translation of an xor-constraint $\XC$ is denoted by $\cnf{\XC}$;
for instance,
$\cnf{x_1 \X x_2 \X x_3 \Equal \F} =
 (\neg x_1 \lor \neg x_2 \lor \neg x_3) \land 
 (\neg x_1 \lor x_2 \lor x_3) \land 
 (x_1 \lor \neg x_2 \lor x_3) \land 
 (x_1 \lor x_2 \lor \neg x_3)$.
\newcommand{\LinComb}{+}
\newcommand{\BigLinComb}{\sum}
We define the linear combination of two xor-constraints,
$\XC = (x_1 \X ... \X x_k \Equal p)$
and
$\XCB = (y_1 \X ... \X y_l \Equal q)$,
by
$\XC \LinComb \XCB =
 (x_1 \X ... \X x_k \X y_1 \X ... \X y_l \Equal {p \X q})$.
An xor-constraint $\XCB = (\Var_1 \X ... \X \Var_k \Equal \parity{})$ with $k \ge 1$
is a \emph{prime implicate} of a satisfiable xor-constraint conjunction $\xorclauses$
if
(i) $\xorclauses \Models \XCB$
but
(ii)
$\xorclauses \NotModels \XCB'$
for all xor-constraints $\XCB'$
for which $\VarsOf{\XCB'}$ is a proper subset of $\VarsOf{\XCB}$.

A \emph{cnf-xor formula} is a conjunction $\orpart \land \xorpart$,
where
$\orpart$ is a conjunction of clauses
and
$\xorpart$ is a conjunction of xor-constraints.
A truth assignment satisfies $\orpart \land \xorpart$ if it satisfies
every clause and xor-constraint in it.

%--------------------------------------------------------------------------
%
%--------------------------------------------------------------------------
%
%--------------------------------------------------------------------------
\subsection{DPLL(XOR) and Xor-Reasoning Modules}
%\footnote{Admittedly, CDCL(XOR) would be better name, like CDCL(T) instead of DPLL(T).}

We are interested in solving the satisfiability of cnf-xor formulas of
the form $\orpart \land \xorpart$ defined above.
%the form defined above % above defined form
%$\orpart \land \xorpart$.
%
%We are interested in solving the satisfiability of cnf-xor formulas of the form
%\[\orpart \land \xorpart\]
%where
%$\orpart$ is a conjunction of clauses (i.e.~a CNF formula) and
%$\xorpart$ is a conjunction of xor-constraints sharing variables with $\orpart$.
%
Similarly to the DPLL($T$) approach for Satisfiability Modulo
Theories, see e.g.~\cite{NieuwenhuisEtAl:JACM06,Handbook:SMT},
the DPLL(XOR) approach \cite{LJN:ECAI2010} for solving cnf-xor formulas
consists of
(i) a conflict-driven clause learning (CDCL) SAT solver that takes care of solving the CNF-part $\orpart$,
and
(ii) an \emph{xor-reasoning module} that handles the xor-part $\xorpart$.
The CDCL solver is the master process,
responsible of guessing values for the variables according to some heuristics
(``branching''),
performing propagation in the CNF-part, conflict analysis, restarts etc.
The xor-reasoning module receives variable values,
called xor-assumptions,
from the CDCL solver and
checks
(i) whether the xor-part can still be satisfied under the xor-assumptions,
and
(ii) whether some variable values, called xor-implied literals,
are implied by the xor-part and the xor-assumptions.
These checks can be incomplete,
like in~\cite{LJN:ECAI2010,LJN:ICTAI2011} for the satisfiability and
in~\cite{LJN:ECAI2010,LJN:ICTAI2011,SoosEtAl:SAT2009} for the implication checks,
as long as the satisfiability check is complete when all the variables have values.
%
%The Gauss-Jordan xor-reasoning module presented in Sect.\ \ref{Sect:GaussJordan} %in this paper
%is complete for both of these checks.
%
The very basic interface for an xor-reasoning module can consist of the following methods:
\begin{itemize}
\item
  $\operatorname{init}(\xorpart)$
  initializes the module with $\xorpart$.
  It may return ``unsat'' if it finds $\xorpart$ unsatisfiable,
  or a set of \emph{xor-implied literals},
  i.e.~literals $\IL$ such that $\xorpart \Models \IL$ holds.

\item
  $\operatorname{assume}(\AL)$ is used to communicate a new variable value
  $\AL$ deduced in the CNF solver part to the xor-reasoning module.
  This value, called \emph{xor-assumption} literal $\AL$,
  is added to the list of current xor-assumptions.
  If $[\AL_1,...,\AL_k]$ are the current xor-assumptions,
  the module then tries to 
  (i) deduce whether $\xorpart \land \AL_1 \land ... \land \AL_k$ became unsatisfiable, i.e.~whether an \emph{xor-conflict} was encountered,
  and if this was not the case,
  (ii) find \emph{xor-implied literals},
  i.e.~literals $\IL$ for which $\xorpart \land \AL_1 \land ... \land \AL_k \Models \IL$ holds.
  The xor-conflict or the xor-implied literals are then returned to the CNF solver part so that it can start conflict analysis (in the case of xor-conflict) or
  extend its current partial truth assignment with the xor-implied literals.

  In order to facilitate conflict-driven backjumping and clause learning
  in the CNF solver part,
  the xor-reasoning module has to provide a clausal \emph{explanation} for each
  xor-conflict and xor-implied literal it reports.
  That is,
  \begin{itemize}
  \item if $\xorpart \land \AL_1 \land ... \land \AL_k$ is
  deduced to be unsatisfiable,
  then the module must report a (possibly empty) clause
  $({\neg l'_1} \lor ... \lor {\neg l'_m})$
  such that
  (i) each $l'_i$ is an xor-assumption or an xor-implied literal,
  and
  (ii) $\xorpart \land l'_1 \land ... \land l'_m$ is unsatisfiable
  (i.e.~$\xorpart \Models ({\neg l'_1} \lor ... \lor {\neg l'_m})$);
  and
  \item
  if it was deduced that
  $\xorpart \land \AL_1 \land ... \land \AL_k \Models \IL$ for some $\IL$,
  then the module must report a clause 
  $({\neg l'_1} \lor ... \lor {\neg l'_m} \lor \IL)$ such that
  (i) each $l'_i$ is an xor-assumption or an xor-implied literal reported earlier,
  and
  (ii)
  $\xorpart \land l'_1 \land ... \land l'_m \Models \IL$,
  i.e.~$\xorpart \Models ({\neg l'_1} \lor ... \lor {\neg l'_m} \lor \IL)$. %,  holds.
  \end{itemize}
\item
  $\operatorname{backtrack}()$ retracts the latest xor-assumption
  and
  all the xor-implied literals deduced after it.
\end{itemize}
Naturally, variants of this interface are easily conceivable.
For instance, a larger set of xor-assumptions can be given with
the $\operatorname{assume}$ method at once instead of only one.

For xor-reasoning modules based on equivalence reasoning,
see~\cite{LJN:ECAI2010,LJN:ICTAI2011}.
The Gaussian and Gauss-Jordan elimination processes
in~\cite{SoosEtAl:SAT2009,Soos,HanJiang:CAV2012,LJN:ICTAI2012}
can also be easily seen as xor-reasoning modules.

%% file: eqreasoning.tex
\section{Equivalence Reasoning and Resolution}
%\section{Resolution Polynomially Simulates Equivalence Reasoning}

We know that there exist infinite families of xor-constraint conjunctions $\xorpart$
whose CNF translations $\bigwedge_{\XC \in \xorclauses}\cnf{\XC}$
have no polynomial size resolution proofs \cite{Urquhart:JACM1987}.
On the other hand,
Gaussian elimination~\cite{Soos} can solve the satisfiability of xor-constraint conjunctions
in polynomial time
(and Gauss-Jordan~\cite{HanJiang:CAV2012,LJN:ICTAI2012}
 can detect all xor-implied literals as well).
As these elimination procedures can be computationally heavy,
more light-weight ``equivalence reasoning'' systems have been proposed~\cite{Li:IPL2000,HeuleEtAl:SAT2004,LJN:ECAI2010,LJN:ICTAI2011}.

Here we study how the equivalence reasoning systems \SUBST{}~\cite{LJN:ECAI2010} and \ECsys{}~\cite{LJN:ICTAI2011} relate to resolution.
These systems are equally powerful in detecting unsatisfiability and xor-implied literals (we'll use \SUBST{} due to its notational simplicity);
they are more powerful than unit propagation but weaker than
Gaussian/Gauss-Jordan elimination.
 %
%Laitinen et al \cite{LaitinenJunttilaNiemela:ECAI2010}
%only show that \SUBST{} is weaker than Gaussian elimination
%\todo{in the sense that it cannot detect all inconsistencies and xor-implied li5terals, explicate here?}
%but leave the relationship to resolution open.

%\subsection{Deduction}

\begin{figure}[tb]
  \centering
  \begin{tabular}{c@{\qquad}c@{\qquad}c@{\qquad}c}
    $\inferencerule{x \Equal \T}{\XC}{\simplification{\XC}{x}{\T}}$
    &
    $\inferencerule{x \Equal \F}{\XC}{\simplification{\XC}{x}{\F}}$
    &
    $\inferencerule{x \X y \Equal \F}{\XC}{\simplification{\XC}{x}{y}}$
    &
    $\inferencerule{x \X y \Equal \T}{\XC}{\simplification{\XC}{x}{y \X \T}}$
    \vspace{1mm}
    \\
    \unitruleP
    &
    \unitruleN
    &
    \eqvruleP
    &
    \eqvruleN
  \end{tabular}%
  \caption{Inference rules of \SUBST{}; $x$ and $y$ are variables, $\XC$ is an xor-constraint, and $x$ occurs in $\XC$.}
  \label{Fig:SUBST}%
\end{figure}

%\begin{figwindow}[2,r,%
%\makebox[.49\textwidth][c]{\includegraphics[width=0.47\textwidth]{Figures/ex-subst}},%
%{A \SUBST{}-derivation.}\label{Fig:SubstRes}]
%
The \SUBST{} deduction system consists of the inference rules in Fig.~\ref{Fig:SUBST}.
Given a conjunction $\psi$ of xor-constraints,
a \emph{\SUBST-derivation} on it is a vertex-labeled directed acyclic graph
$\igraph = \Tuple{\vertices, \edges, \Labfunc}$ such that
for each vertex $\Vertex \in \vertices$ it holds that
(i) if $\Vertex$ has no incoming edges, then $\Lab{\Vertex}$ is an xor-constraint in $\psi$,
and
(ii) otherwise $\Vertex$ has two incoming edges, say from $\Vertexp$ and $\Vertexpp$, and $\Lab{\Vertex}$ is obtained from $\Lab{\Vertexp}$ and $\Lab{\Vertexpp}$ by applying one of the inference rules.
As an example, Fig.~\ref{Fig:SubstRes}(a) shows a \SUBST-derivation
on $(x \X y \X z \Equal \T) \land (x \X z \X w \Equal \F) \land (y \X w \X t \Equal \T) \land (x)$,
please ignore the ``Cut $\Cut$'' line for now.
%\end{figwindow}

%\iffalse
\begin{figure}[tb]
  \centering
  \begin{tabular}{c@{\quad}c}
    \includegraphics[width=0.48\textwidth]{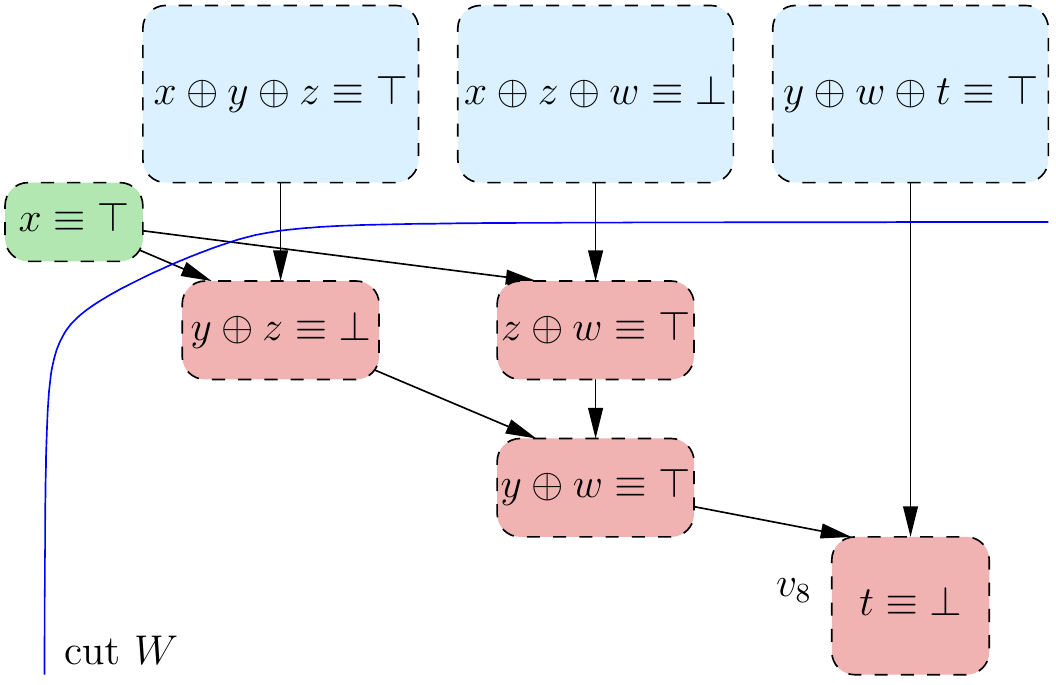}
    &
    \includegraphics[width=0.48\textwidth]{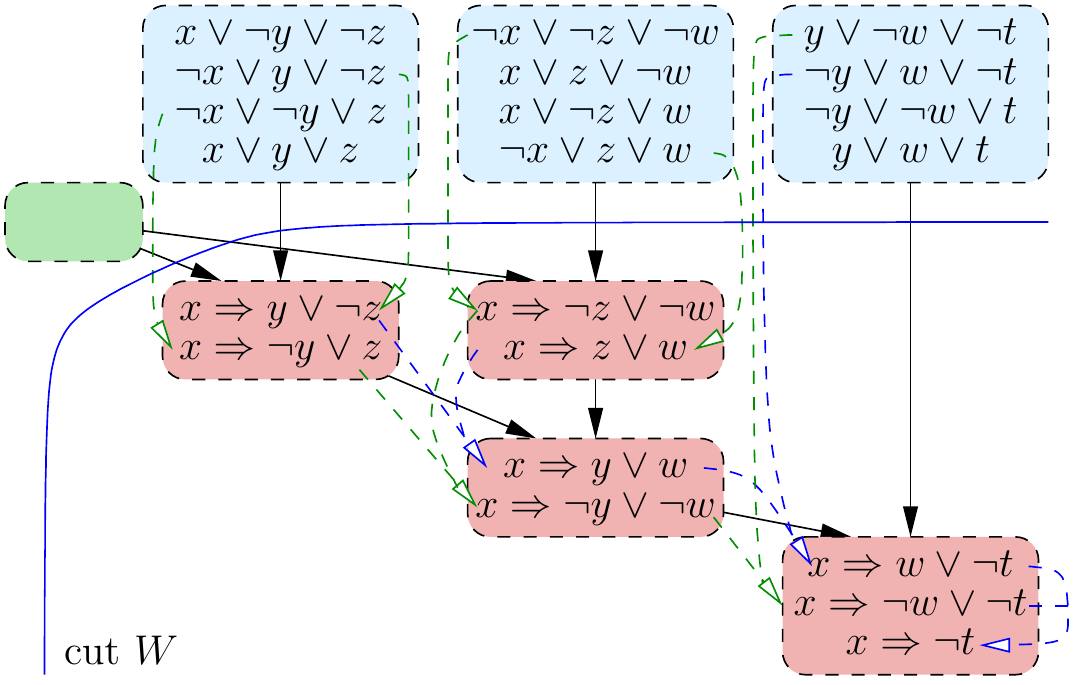}
    \\
    (a) a \SUBST{}-derivation
    &
    (b) the resolution derivation for an implicative %explanation
    \\
    %& an implicative explanation
    & explanation (the dotted arrows)
  \end{tabular}%
  \caption{\SUBST{}-derivations and resolution.}
  \label{Fig:SubstRes}
\end{figure}
%\fi

If we can derive an xor-constraint $\XC$ with \SUBST{}, we can derive (in the CNF
translated instance) a CNF translation of $\XC$ with
resolution relatively compactly:
\begin{theorem}\label{Thm:ResDeriv}
  Assume a \SUBST-derivation $\igraph = \Tuple{\vertices,\edges,\Labfunc}$ on
  a conjunction $\psi$ of xor-constraints.
  There is a resolution derivation $\ResDer$ on
  $\bigwedge_{\XC \in \psi}\cnf{\XC}$
  such that
  (i)
  if $\Vertex \in \vertices$ and $\Lab{\Vertex} \neq \T$,
  then the clauses $\cnf{\Lab{\Vertex}}$ occur in $\ResDer$,
  and
  (ii)
  $\ResDer$ has at most $\Card{\vertices}2^{m-1}$ clauses,
  where $m$ is the number of variables in the largest xor-constraint in $\psi$.
\end{theorem}
\iffalse %PREVIOUS
\begin{theorem}
  \label{Thm:ResDeriv}
  Assume a \SUBST-derivation $\igraph = \Tuple{\vertices,\edges,\Labfunc}$ on
  a conjunction $\psi$ of xor-constraints.
  There is a resolution derivation $\RP$ on
  $\bigwedge_{\XC \in \psi}\cnf{\XC}$
  such that
  (i) $\RP$ has at most $\Card{\vertices}2^{m-1}$ clauses,
  where $m$ is the number of variables in the largest xor-constraint in $\psi$,
  and
  (ii)
  if $\XC$ labels a vertex of $\igraph$,
  then the clauses $\cnf{\XC}$ occur in $\RP$.
\end{theorem}
\fi %PREVIOUS
A similar result is already observed in \cite{Li:IPL2000}
when restricted on binary and ternary xor-constraints.
Recalling that for each xor-constraint $\XC$ the CNF translation
$\cnf{\XC}$ is exponentially large in the number of variables in $\XC$,
we can say that resolution simulates \SUBST{}-derivations ``pseudo-linearly''.
Furthermore, the natural encodings in many application domains (e.g.\ logical cryptanalysis) seem to employ xor-constraints with only few (typically 3) variables only.

\subsection{Implicative Explanations}

In the DPLL(XOR) framework,
the clausal explanations for the xor-implied literals and xor-conflicts
are vital for the CDCL solver
when it performs its conflict analysis and clause learning.
%
%In addition to derivation,
%the xor-reasoning module must also be able to provide clausal explanations
%for the xor-implied literals.
%
We next show that the implicative explanation procedure described
in \cite{LJN:ECAI2010} can also be simulated with resolution,
and discuss the consequence of this result.
%
%This is in contrast to the improved explanation technique introduced in \cite{LJN:SAT2012}: we will show in the next section that it can actually simulate Gaussian elimination and is thus stronger than ...

Like the conflict resolution methods in modern CNF-level CDCL solvers,
the explanation method is based on taking certain cuts in derivations.
Assume a \SUBST-derivation $\igraph = \Tuple{\vertices, \edges, \Labfunc}$
on $\xorclauses \land {\AL_1 \land ... \land \AL_k}$,
where $\xorclauses$ is a conjunction of xor-constraints and
$\AL_1,...,\AL_k$ are some xor-assumption literals.
For a non-input vertex $\Vertex \in \vertices$,
a \emph{cut for $\Vertex$} is a partitioning $(\CutA,\CutB)$ of $\vertices$
such that
(i) $\Vertex \in \CutB$, and
(ii) if $\Vertexp \in \vertices$ is an input vertex and
there is a path from $\Vertexp$ to $\Vertex$, then $\Vertexp \in \CutA$.
As an example, the line ``cut $\Cut$'' shows a cut for
the vertex $\Vertex_8$ in Fig.~\ref{Fig:SubstRes}(a).
The \emph{implicative explanation} of the vertex $\Vertex$ under the cut $\Cut$
is the conjunction $\CExpl{\Vertex,\Cut} = \tmpf{\Vertex}{\Cut}$,
there $\tmpfname{\Cut}$ is recursively defined as: % follows:
\begin{itemize}
\item[E1]
  If $u$ is an input vertex with $\Lab{u} \in \xorclauses$,
  then $\tmpf{u}{\Cut} = \T$.
\item[E2]
  If $u$ is an input vertex with $\Lab{u} \in \Set{\AL_1,...,\AL_k}$,
  then $\tmpf{u}{\Cut} = \Lab{u}$.
\item[E3]
  If $u$ is a non-input vertex in $\CutA$,
  then $\tmpf{u}{\Cut} = \Lab{u}$.
\item[E4]
  If $u$ is a non-input vertex in $\CutB$,
  then
  $\tmpf{u}{\Cut} =
  {\tmpf{u_1}{\Cut} \land \tmpf{u_2}{\Cut}}$,
  where $u_1$ and $u_2$ are the source vertices of the two edges incoming to $u$.
\end{itemize}
If the cut is \emph{cnf-compatible},
meaning that all the vertices in $\CutA$ having an edge to a vertex in $\CutB$ are either (i) xor-constraints in $\xorclauses$ or (ii) unary xor-constraints,
then the explanation $\CExpl{\Vertex,\Cut}$ is a conjunction of literals
and
the clausal explanation of the xor-implied literal $\Lab{\Vertex}$
returned to the CDCL part is $\CExpl{\Vertex,\Cut} \Implies \Lab{\Vertex}$.
As an example,
for the vertex $\Vertex_8$ and cnf-compatible cut $\Cut$ in Fig.~\ref{Fig:SubstRes}(a),
we have $\CExpl{\Vertex_8,\Cut} = (x)$
and
the clausal explanation is thus $x \Implies {\neg t}$, i.e.,
$({\neg x} \lor {\neg t})$.

We now prove that
all such clausal explanations can in fact be derived with resolution from
the CNF translation of the original xor-constraints $\xorclauses$ only,
without the use of xor-assumptions.
To illustrate some parts of the construction,
Fig.~\ref{Fig:SubstRes}(b) shows how the clausal explanation
$x \Implies {\neg t}$ above can be derived.
%
%For example,
%Fig.~\ref{Fig:SubstRes}(b) shows how the clausal explanation
%$x \Implies {\neg t}$ can be derived.
%
\begin{theorem}\label{Thm:ResExp}
  Assume a \SUBST-derivation $\igraph = \Tuple{\vertices,\edges,\Labfunc}$ on
  $\xorpart \land \AL_1 \land \dots \land \AL_k$
  and
  a cnf-compatible cut $\Cut = (\CutA, \CutB)$.
  %for a non-input vertex $\Vertex$ with $\Lab{\Vertex} \neq \T$.
  %
  There is a resolution derivation $\ResDer$ on
  $\bigwedge_{\XC \in \xorpart}\cnf{\XC}$
  such that
  (i)
  for each vertex $\Vertex \in \CutB$ with $\Lab{\Vertex} \neq \T$,
  $\ResDer$ includes all the clauses in $\Setdef{\CExpl{v,\Cut} \Implies C}{C \in \cnf{\Lab{v}}}$,
  and
  (ii)
  $\ResDer$ has at most $\Card{\vertices}2^{m-1}$ clauses,
  where $m$ is the number of variables in the largest xor-constraint in $\xorpart$.
\end{theorem}
\iffalse % PREVIOUS
\begin{theorem}
  \label{Thm:ResExp}
  Assume a \SUBST-derivation $\igraph = \Tuple{\vertices,\edges,\Labfunc}$ on
  $\xorclauses \land \AL_1 \land \ldots \land \AL_k$
  and a cnf-compatible cut $\Cut = (\CutA, \CutB)$.
  There is a resolution derivation $\RP$ on
  $\bigwedge_{\XC \in \xorclauses}\cnf{\XC} \land \AL_1 \land ... \land \AL_k$
  such that
  (i) $\RP$ has most $\Card{\vertices} 2^{m-1}$ clauses,
  where $m$ is the number of variables in the largest xor-constraint in $\xorclauses$,
  and
  (ii)
  for each vertex $\Vertex \in \vertices$ it holds:
  \begin{itemize}
  \item
    if $\Vertex \in \CutA$, then the clauses $\cnf{\Lab{\Vertex}}$ are in $\RP$,
  \item
    if $\Vertex \in \CutB$,
    then all the clauses $\Setdef{\CExpl{v,\Cut} \Implies C}{C \in \cnf{\Lab{\Vertex}}}$ are in $\RP$.
\end{itemize}
\end{theorem}
\fi % PREVIOUS
As modern CDCL solvers can be seen as resolution proof producing engines
\cite{ZhangMalik:DATE2003,BeameKautzSabharwal:JAIR2004},
a DPLL(XOR) solver with \SUBST{} or \ECsys{} as the xor-reasoning module
can thus also be seen as such engine:
the clausal explanations used by the CDCL part can be first
obtained with resolution
and
then treated as normal clauses when producing the resolution proof
corresponding to the execution of the CDCL part.
%
%And,
%recalling that modern CDCL solvers can polynomially simulate resolution \cite{PipatsrisawatDarwiche:AI2011},
%we can conclude that, for instances with fixed width xor-constraints,
%the underlying proof system of a DPLL(XOR) solver using \SUBST{} or \ECsys{}
%as the xor-reasoning module is polynomially equivalent to resolution.
%
And,
recalling that modern CDCL solvers can polynomially simulate resolution \cite{PipatsrisawatDarwiche:AI2011},
we have the following:
\begin{corollary}
For cnf-xor instances with fixed width xor-constraints,
the underlying proof system of a DPLL(XOR) solver using \SUBST{} or \ECsys{}
as the xor-reasoning module is polynomially equivalent to resolution.
\end{corollary}

%% file: parityexp.tex
\section{Parity Explanations and Gauss-Jordan Elimination}
%\section{Non-deterministic Unit Propagation and Parity Explanations Simulate Gaussian Elimination}

A key observation made in \cite{LJN:SAT2012} was that the inference rules
in Fig.~\ref{Fig:SUBST} (and some others, as explained in \cite{LJN:SAT2012})
could not only be read as ``the premises imply the consequence'' but also
as ``the linear combination of premises equals the consequence''.
This led to the introduction of an improved explanation method,
called parity explanations,
which can produce (i) smaller clausal explanations, and
(ii) new xor-constraints that are logical consequences of the original ones.
As shown in \cite{LJN:SAT2012},
even when applied on a very weak deduction system \UP{},
which only uses the unit propagation rules $\unitruleP$ and $\unitruleN$
in Fig.~\ref{Fig:SUBST},
the parity explanation method can quickly detect the unsatisfiability
of some instances whose CNF translations have no polynomial size resolution refutations \cite{Urquhart:JACM1987}.
We now strengthen this result and
prove that parity explanations on \UP-derivations can in fact produce
xor-constraints corresponding to the explanations produced by
Gauss-Jordan elimination,
provided that one can make the xor-assumptions suitably and
each variable in the xor-constraint conjunction occurs at most three times
(Thm.\ \ref{Thm:PexpSim} below).
%
%We now strengthen this result and
%prove that parity explanations on \UP-derivations can in fact produce
%xor-constraints that have at most the same number of variables that
%corresponding explanations produced by Gauss-Jordan elimination,
%provided that one can make the xor-assumptions suitably and
%each variable in the xor-constraint conjunction occurs at most three times
%(Thm.\ \ref{Thm:PexpSim} below).

%$\inferencerule{x \X y \Equal \T}{\XC}{\simplification{\XC}{x}{y \X \T}}$

Formally,
assume a \UP-derivation
$\igraph = \Tuple{\vertices, \edges, \Labfunc}$
for $\xorclauses \land \AL_1 \land ... \land \AL_k$.
For each non-input vertex $\Vertex$ of $\igraph$,
and
each cut $\Cut=(\CutA,\CutB)$ of $\igraph$ for $\Vertex$,
the \emph{parity explanation} of $\Vertex$ under $\Cut$
is $\PExpl{\Vertex,\Cut} = \tmpf{\Vertex}{\Cut}$,
there $\tmpfname{\Cut}$ is recursively defined as earlier
for $\CExpl{\Vertex,\Cut}$ except that the case ``E4'' is replaced by
\begin{itemize}
\item[E4]
  If $u$ is a non-input node in $\CutB$,
  then
  $\tmpf{u}{\Cut} =
  {\tmpf{u_1}{\Cut} \LinComb \tmpf{u_2}{\Cut}}$,
  where $u_1$ and $u_2$ are the source nodes of the two edges incoming to $u$.
\end{itemize}
As shown in \cite{LJN:SAT2012},
$\xorclauses \Models {\PExpl{\Vertex,\Cut} \LinComb \Lab{\Vertex}}$
and
the clausal explanation for $\Lab{\Vertex}$ can be obtained from
$\cnf{\PExpl{\Vertex,\Cut} \LinComb \Lab{\Vertex}}$.
As an example,
the parity explanation $\PExpl{\Vertex_8,\Cut}$ of the vertex $\Vertex_8$
in Fig.~\ref{Fig:SubstRes}(a) is $(\F \Equal \F)$, i.e.~$\T$,
and
indeed $(x \X y \X z \Equal \T)\land(x \X z \X w \Equal \F) \land (y \X w \X t \Equal \T) \Models (\F \Equal \F) \LinComb \Lab{\Vertex_8} = (t \Equal \F)$.
Note that %the variable
$x$ does not occur in the parity explanation or
in the clausal explanation $(\neg t)$ returned.

For instances in which each variable occurs at most three times
we can prove that,
by selecting the xor-assumptions appropriately,
%parity explanations can in fact produce all logical consequence xor-constraints
%or their linear combination summands:
parity explanations can in fact produce all prime implicate xor-constraints:
\begin{theorem}\label{Thm:PexpSim}
  Let $\xorclauses$ be a conjunction of xor-constraints
  such that each variable occurs in at most three xor-constraints.
  
  If $\xorclauses$ is unsatisfiable,
  then there is a
  \UP{}-derivation on $\xorclauses \land y_1 \land ... \land y_m$
  with some $y_1,...,y_m \in \VarsOf{\xorclauses}$,
  a vertex $\Vertex$ with $\Lab{\Vertex} = (\F \Equal \T)$ in it,
  and
  a cut $\Cut$ for $\Vertex$ such that
  $\PExpl{\Vertex, \Cut} = (\F \Equal \F)$
  and thus
  $\PExpl{\Vertex, \Cut} \LinComb \Lab{\Vertex} = (\F \Equal \T)$.

  If $\xorclauses$ is satisfiable and
  $\xorclauses \Models (\Var_1 \X ... \X \Var_k \Equal \parity{})$,
  then there is a
  \UP{}-derivation on $\xorclauses \land (\Var_1\Equal\parity{1}) \land ... \land (\Var_k\Equal\parity{k}) \land y_1 \land ... \land y_m$
  with some $y_1,...,y_m \in \VarsOf{\xorclauses}\setminus\Set{\Var_1,...,\Var_k}$,
  a vertex $\Vertex$ with $\Lab{\Vertex} = (\F \Equal \T)$ in it,
  and
  a cut $\Cut$ for $\Vertex$ such that
  $\PExpl{\Vertex, \Cut} \LinComb \Lab{\Vertex} = (\Var'_1 \X ... \X \Var'_l \Equal \Parityp)$ for some $\Set{\Var'_1,...,\Var'_l} \subseteq \Set{\Var_1,...,\Var_k}$ and $\Parityp \in \Set{\F,\T}$
  such that
  $\xorclauses \Models (\Var'_1 \X ... \X \Var'_l \Equal \Parityp)$.
\end{theorem}
%
%As a consequence,
%parity explanations on \UP-derivations can in theory simulate the
%complete Gauss-Jordan elimination propagation engine~\cite{HanJiang:CAV2012,LJN:ICTAI2012}
%in the DPLL(XOR) framework
%if we allow unlimited restarts in the CDCL part and xor-constraint learning \cite{LJN:SAT2012}:
%we can first learn all the linear combinations that the Gauss-Jordan engine
%would use to detect xor-implied literals and conflicts.
%
%\todo{are shorter parity explanations better?}

Now observe that the clausal explanations provided by
the complete Gauss-Jordan elimination propagation engine of~\cite{LJN:ICTAI2012}
are based on prime implicate xor-constraints (this follows from the fact that reduced row-echelon form matrices are used and the explanations are derived from the rows of such matrices).
As a consequence,
for instances in which each variable occurs at most three times,
parity explanations on \UP-derivations can in theory simulate the
complete Gauss-Jordan elimination propagation engine~\cite{LJN:ICTAI2012}
in the DPLL(XOR) framework
if we allow unlimited restarts in the CDCL part and xor-constraint learning \cite{LJN:SAT2012}:
we can first learn all the linear combinations that the Gauss-Jordan engine
would use to detect xor-implied literals and conflicts.

%-------------------------------------------------
%
%-------------------------------------------------
%
%-------------------------------------------------
\subsection{Experimental Evaluation}
\label{Sec:PexpResults}

To evaluate the practical applicability of parity explanations further and to compare it to the xor-reasoning module using incremental Gauss-Jordan elimination presented in~\cite{LJN:ICTAI2012}, we used our
prototype solver based on \Minisat{} \cite{EenSorensson:2004} (version 2.0
core) extended with four different xor-reasoning modules: 
(i) $ \UP{}$ deduction system with implicative explanations,
(ii) $ \UP{}$ with parity explanations (\Tool{UP+PEXP}),
(iii) $ \UP{}$ with parity explanations and xor-constraint learning (\Tool{UP+PEXP+learn}) as
described in~\cite{LJN:SAT2012}, and
(iv) incremental Gauss-Jordan elimination with biconnected component decomposition (\Tool{UP+Gauss-Jordan}) as described in~\cite{LJN:ICTAI2012}.
We ran the solver configurations on two benchmark sets. %~\footnote{We acknowledge the computational resources provided by Aalto Science-IT project.}.
The first benchmark set consists of instances in ``crafted'' and
``industrial/application'' categories of the SAT Competitions 2005, 2007, and
2009 as well as all the instances in the SAT Competition 2011 (see
    \url{http://www.satcompetition.org/}).
We applied the xor-constraint extraction algorithm described in~\cite{Soos} to
these CNF instances and found a large number of instances with xor-constraints.
To get rid of some ``trivial'' xor-constraints,
we eliminated unary clauses and binary xor-constraints from each instance
by unit propagation and substitution, respectively.
After this easy preprocessing, 474 instances (with some duplicates due to
overlap in the competitions) having xor-constraints remained.
In the second benchmark set we focus on the domain of logical
cryptanalysis by modeling a ``known cipher stream'' attack on stream cipher
Hitag2.
The task is to recover the full key when a small number of cipher stream bits
(33-38 bits, 51 instances / stream length) are given. In the attack, the IV and
a number of cipher stream bits are given.  There are only a few more 
generated cipher stream bits than key bits, so a number of keys probably produce
the same prefix of the cipher stream.

\begin{figure}[tb]
\includegraphics[width=0.33\textwidth]{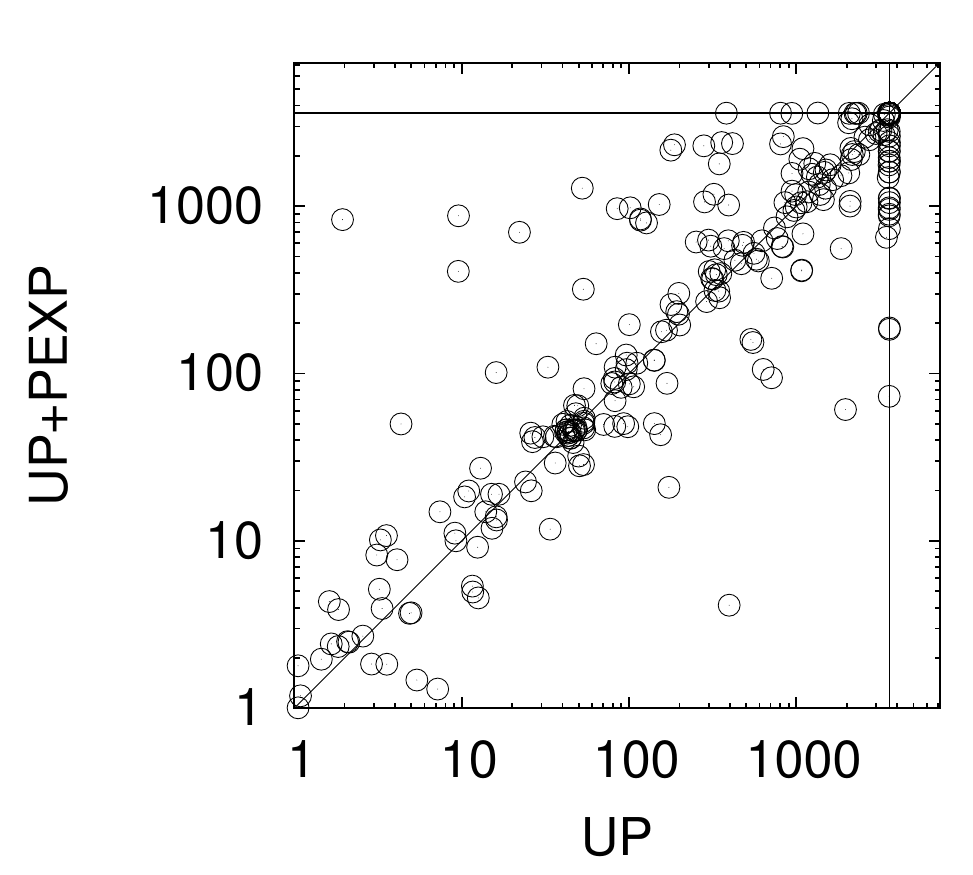}
\includegraphics[width=0.33\textwidth]{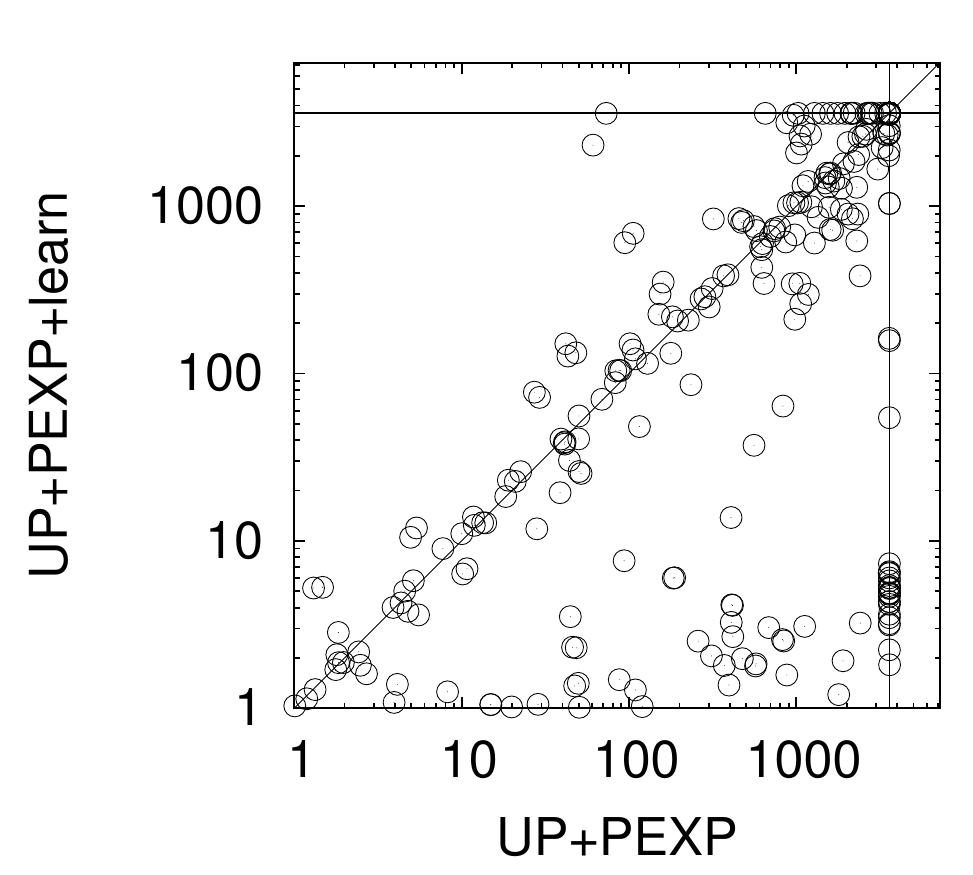}
\includegraphics[width=0.33\textwidth]{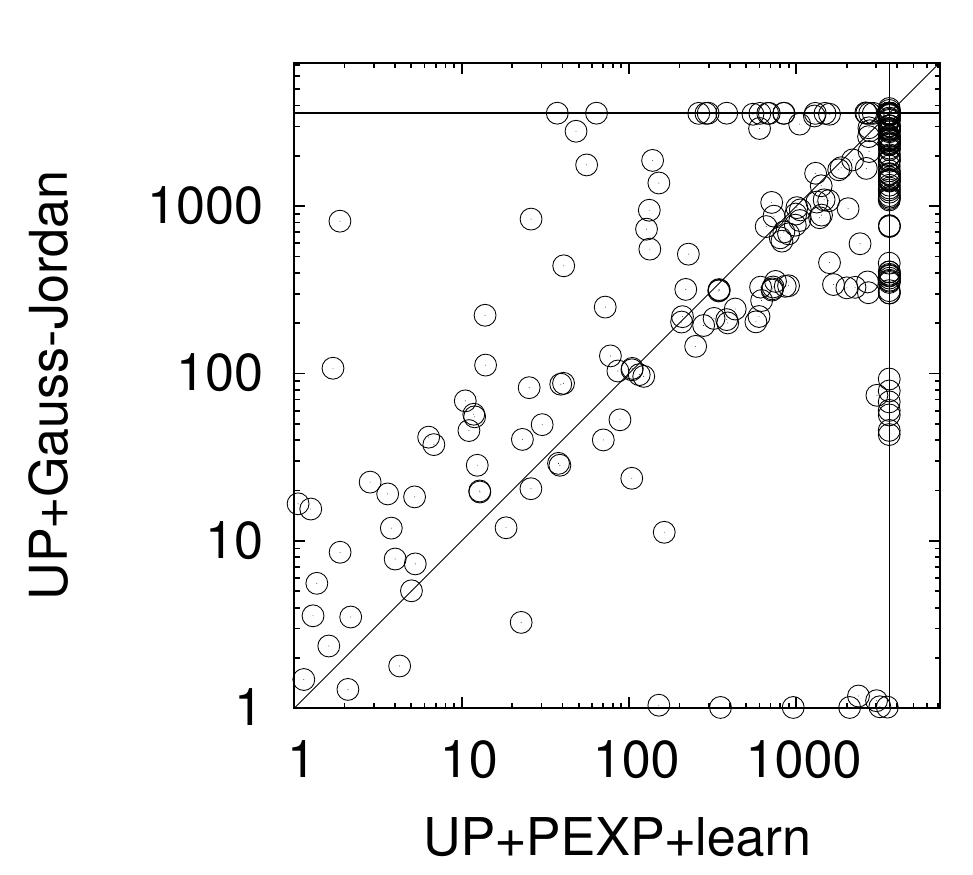}
  \caption{Comparing parity explanations and Gauss-Jordan elimination on SAT 05-11 instances}
  \label{Fig:PexpSatResults}
\end{figure}

\begin{figure}[tb]
\includegraphics[width=0.33\textwidth]{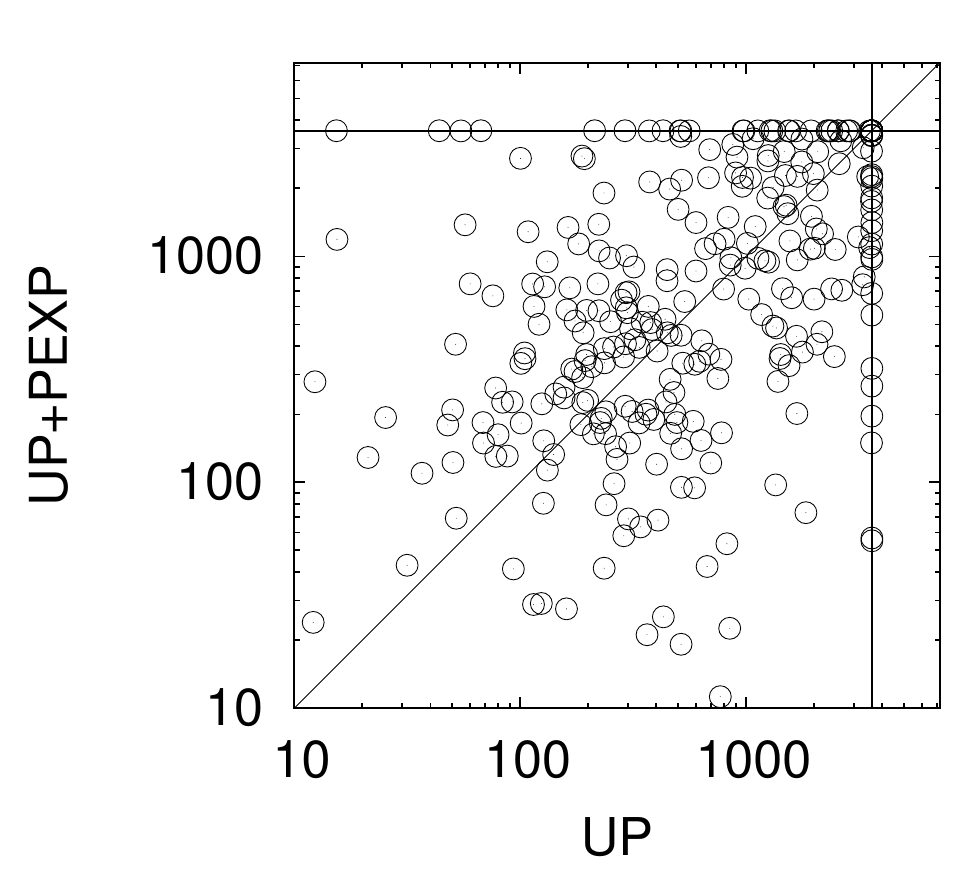}
\includegraphics[width=0.33\textwidth]{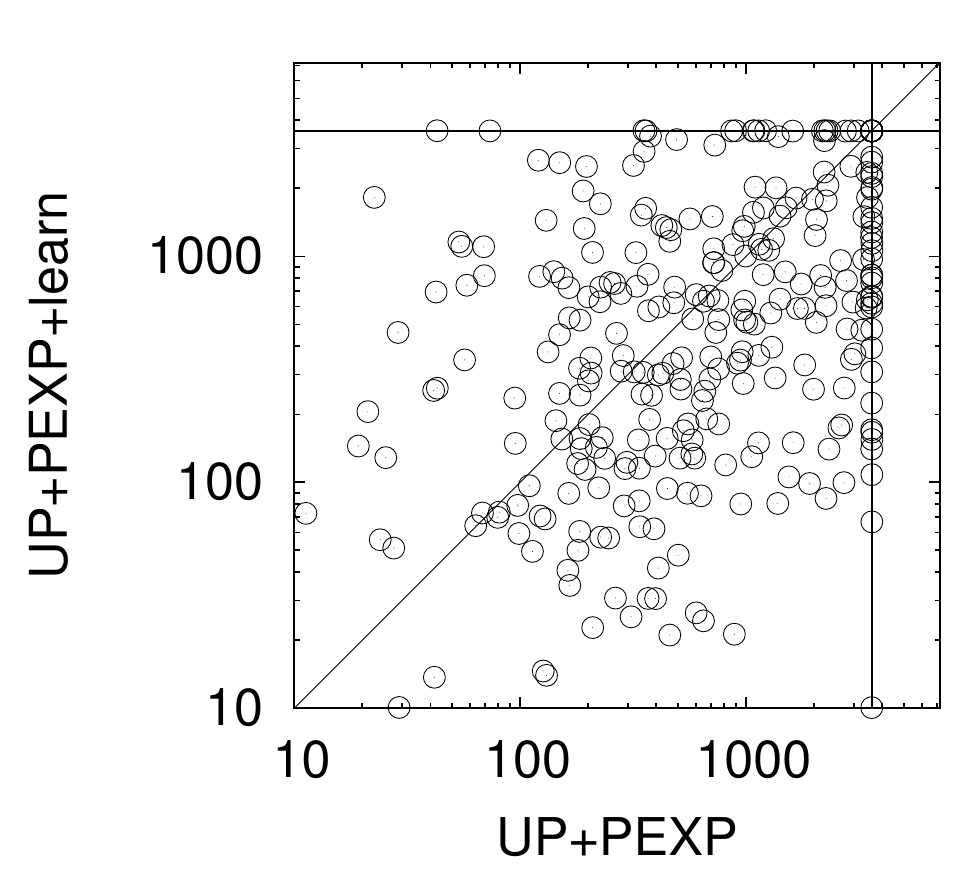}
\includegraphics[width=0.33\textwidth]{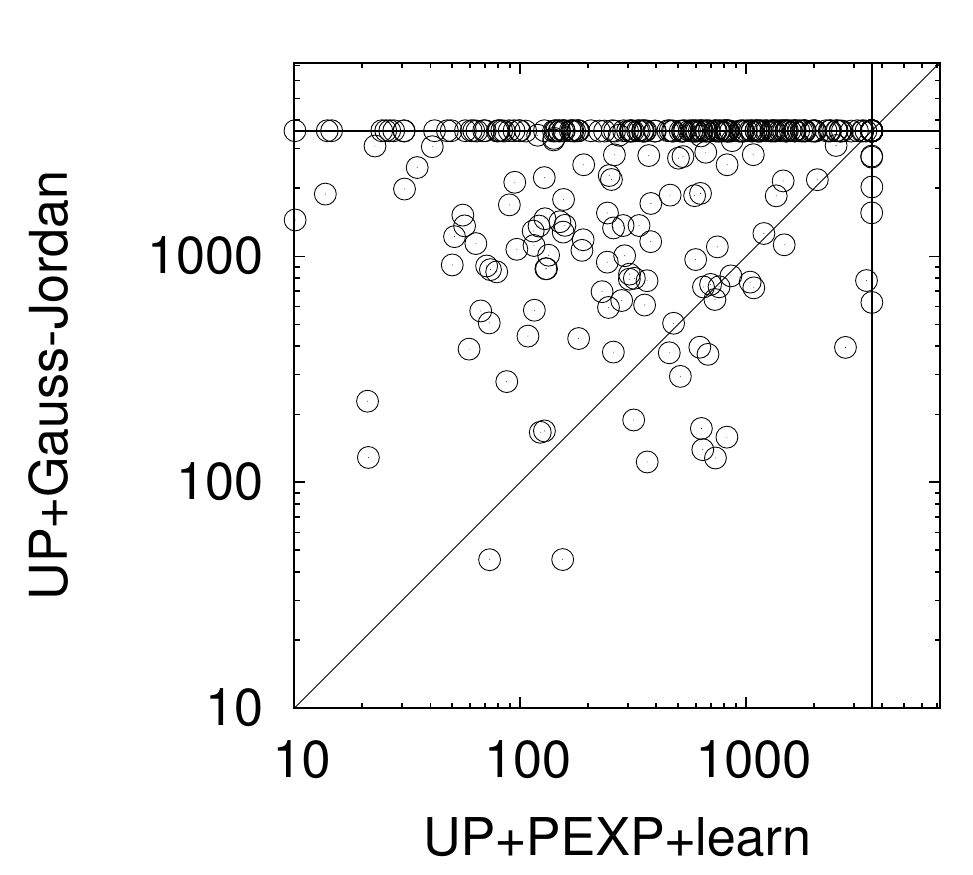}
 
  \caption{Comparing parity explanations and Gauss-Jordan elimination on Hitag2 instances}
  \label{Fig:PexpHitag2Results}
\end{figure}
\begin{figure}[tb]
  \centering
    \begin{tabular}{|r@{\ }|@{\ }c@{\ }|@{\ }c@{\ }|@{\ }c@{\ }|@{\ }c@{\ }|@{\ }c@{\ }|c@{\ }|c@{\ }|c@{\ }|c@{\ }|}
      \hline
      & \multicolumn{5}{c|}{SAT Competition} & Hitag2 & Grain & A5/1 & Trivium \\
      & 2005 & 2007 & 2009 & 2011 & all & & & & \\
      \hline\hline
      instances              & 123         & 100         & 140         & 111         & 474          & 301          & 357          & 640 & 1020  \\
      \hline 
      \Tool{UP}              & 79          & 66          & 82          & 41          & 268          & 264          & \textbf{305} & 605 & 879 \\
      \Tool{UP+PEXP}         & 78          & \textbf{70} & 85          & \textbf{48} & 281          & 257          & 301          & 610 & 867  \\
      \Tool{UP+PEXP+learn}   & 96          & 69          & \textbf{88} & \textbf{48} & \textbf{301} & \textbf{274} & 257          & 635 & \textbf{909}   \\
      \Tool{UP+Gauss-Jordan} & \textbf{97} & 61          & 82          & 39          & 279          & 115          & 84           & \textbf{640} & 880 \\ 
      \hline
    \end{tabular}
  \caption{Number of instances solved within the time limit (3600s)}
  \label{Fig:PexpNumSolved}
\end{figure}

The results for the SAT Competition benchmarks are shown in
Fig.~\ref{Fig:PexpSatResults} and the results for Hitag2 in
Fig.~\ref{Fig:PexpHitag2Results}.
The number of solved instances is shown in Fig.~\ref{Fig:PexpNumSolved}.
For both benchmark sets, parity explanations without learning do not seem to
reduce the number of decisions nor the solving time. 
However, storing parity explanations as learned xor-constraints results in a
significant reduction in the number of decisions and this is also reflected in
the solving time.
Most variables have at most three occurrences (98\% of variables in Hitag2, and
97\% in SAT instances), so in most cases a parity explanation that is
equivalent to the ``Gauss-Jordan explanation'' could be found using
nondeterministic unit propagation.
The SAT competition benchmarks has 64 instances consisting entirely of parity
constraints which were of course solved without branching by Gauss-Jordan
elimination. 
The results of the other instances that require searching on the CNF part
illustrate that when parity explanations are learned, many instances can be
solved much faster than with Gauss-Jordan elimination.
It remains open whether the theoretical power of parity explanations could be
exploited to an even higher degree by employing different propagation
heuristics.

We also evaluated the performance of the four xor-reasoning modules on three other
ciphers, Grain, A5/1, and Trivium, by encoding a similar ``known cipher stream'' attack
as with Hitag2 above. For Grain, the simplest method, plain unit propagation,
works the best. Gauss-Jordan elimination does not reduce the number of
decisions enough to compensate for the computational overhead of complete
parity reasoning. Parity explanations reduce the number of decisions slightly, 
but the small computational overhead is still too much. 
For A5/1, the solver using Gauss-Jordan elimination works the best. The solvers
using parity explanations perform better than plain unit propagation, too, but
not as well as the solver with Gauss-Jordan elimination.
For Trivium, the solver using parity explanations with learning solves the most
instances.
\newcommand{\LC}{+}

%% file: xupify.tex
\section{Simulating Stronger Parity Reasoning with Unit Propagation}

An efficient translation for simulating equivalence reasoning with unit
propagation has been presented in our earlier work~\cite{LJN:CP2012}.
We now present a translation that adds redundant xor-constraints and auxiliary
variables in the problem guaranteeing that unit propagation is enough to always
deduce all xor-implied literals in the resulting xor-constraint conjunction.
The translation thus effectively simulates a complete parity reasoning
engine based on incremental Gauss-Jordan elimination presented
in~\cite{LJN:ICTAI2012,HanJiang:CAV2012}.
The translation can be seen as an arc-consistent encoding of the
xor-reasoning theory (also compare to the eager approach to SMT~\cite{Handbook:SMT}).
The translation is based on ensuring that each relevant linear combination of
original variables has a corresponding ``alias'' variable, and adding
xor-constraints that enable unit propagation to infer values of ``alias''
variables when corresponding linear combinations are implied.
The translation, which is exponential in the worst-case, can be made polynomial
by bounding the length of linear combinations to consider. While unit
propagation may not be able then to deduce all xor-implied literals, the
overall performance can be improved greatly. 

The redundant xor-constraint conjunction, called a {\it GE-simulation formula}
$\psi$, added to $\xorclauses$ by the translation should satisfy the following:
(i) the satisfying truth assignments of $\xorclauses$ are exactly the ones of
$\xorclauses \wedge \psi$ when projected to $\VarsOf{\xorclauses}$, and 
(ii) if $ \xorclauses$ is satisfiable and 
 $\xorclauses \wedge \AL_1 \wedge \dots \wedge \AL_k \Models \IL $, then $\IL$ is \UP{}-derivable from $ (\xorclauses \wedge \psi) \wedge \AL_1 \wedge \dots \wedge \AL_k$, and (iii) if $\xorclauses$ is unsatisfiable, then $(\xorclauses \wedge \psi) \UPderiv (\bot \equiv \top) $.

\newcommand{\PropTable}[1]{\ensuremath{\operatorname{ptable}(#1)}}
\newcommand{\PropTableName}{\ensuremath{\operatorname{ptable}}}
\newcommand{\HasPropTable}[2]{\ensuremath{#1 \subseteq_{\tiny{\operatorname{UP}}} #2}}
\newcommand{\eijtrans}[1]{\ensuremath{\textup{$Eq$}}(#1)}
\newcommand{\opttransname}{\ensuremath{\textup{$Eq$}^\star}}
\newcommand{\eijtransname}{\ensuremath{\textup{$Eq$}}}
\newcommand{\opttrans}[1]{\ensuremath{\textup{$Eq$}^\star}(#1)}
\newcommand{\getrans}[1]{\ensuremath{\textup{$k\mbox{-Ge}$}(#1)}}
\newcommand{\kgetrans}[2]{\ensuremath{\textup{$#1\mbox{-Ge}$}(#2)}}
\newcommand{\getransname}{\ensuremath{\textup{$k\mbox{-Ge}$}}}
\newcommand{\kgetransname}[1]{\ensuremath{\textup{$#1\mbox{-Ge}$}}}

\newcommand{\normalizename}{\ensuremath{\operatorname{3-xor}}}
\newcommand{\normalize}[1]{\normalizename(#1)} 

\begin{figure}[bt]
{\small
\begin{tabbing}
99.x\={mm}\={mm}\={mm}\={mm}\={mm}\=\kill
$\PropTable{Y, \xorclauses, k}$: start with $\xorclauses' = \xorclauses $ \\
1.\>\For{} each $Y' \subseteq Y$ such that $|Y'| \leq k$ and $ Y' \not = \emptyset$\\
2.\>\>\If{} there is no $a \in \VarsOf{\xorclauses'}$ such that $ (a \oplus Y' \equiv \bot) $ is in $\xorclauses'$ \\
3.\>\>\>$\xorclauses' \leftarrow \xorclauses' \wedge (a \oplus Y' \equiv \bot) $ where $a$ is a new ``alias'' variable for $Y'$ \\
4.\>\>\If{} $ (Y' \equiv \parity{}) $ is in $ \xorclauses'$ and $ (a \equiv \parity{}) $ is not in $ \xorclauses'$ where $\parity{} \in \set{\bot,\top}$\\
5.\>\>\>$\xorclauses' \leftarrow \xorclauses' \wedge (a \equiv \parity{}) $ \\
6.\>\For{} each pair of subsets $ Y_1, Y_2 \subseteq Y $ such that $|Y_1| \leq k $, $|Y_2| \leq k$, and $ Y_1 \not = Y_2$ \\
7.\>\>\If{} there is an ``alias'' variable $a_3 \in \VarsOf{\xorclauses'} $ such that $(a_3 \oplus (Y_1 \oplus Y_2) \equiv \bot) $ is in $\xorclauses'$ \\
8.\>\>\>$a_1 \leftarrow \mbox{ the ``alias'' variable } v \mbox{ such that } (v \oplus Y_1 \equiv \bot) \mbox { is in } \xorclauses'$ \\
9.\>\>\>$a_2 \leftarrow \mbox{ the ``alias'' variable } v \mbox{ such that } (v \oplus Y_2 \equiv \bot) \mbox { is in } \xorclauses'$ \\
10.\>\>\>\If{} $ (a_1 \oplus a_2 \oplus a_3 \equiv \bot) $ is not in $ \xorclauses'$\\
11.\>\>\>\>$\xorclauses' \leftarrow \xorclauses' \wedge (a_1 \oplus a_2 \oplus a_3 \equiv \bot) $ \\
12.\> \Return{} $\xorclauses' \setminus \xorclauses$ \\
\end{tabbing}%
}%
\vspace{-2mm}
\caption{The $\PropTableName$ translation}
\label{Fig:PropTable}
\end{figure}%

\newcommand{\vleft}{\ensuremath{V}}
\begin{figure}[bt]
{\small
\begin{tabbing}
99.x\={mm}\={mm}\={mm}\={mm}\={mm}\=\kill
$\getrans{\xorclauses}$:\; start with $ \xorclauses' = \xorclauses $ and $V = \VarsOf{\xorclauses}$ \\
1.\>\While{} ($V \not = \emptyset$):\\
2.\>\>Let $\ClausesOf{x}{\xorclauses'} = \Setdef{\XC}{\XC\mbox{ in } \xorclauses' \mbox{ and } x \in \VarsOf{\XC}}$ \\
3.\>\>Let $x$ be a variable in $V$ minimizing $|\VarsOf{\ClausesOf{x}{\xorclauses'}} \cap V|$ \\
4.\>\>$\xorclauses' \leftarrow \xorclauses' \wedge \PropTable{\VarsOf{\ClausesOf{x}{\xorclauses'}} \cap V, \xorclauses', k}$ \\
5.\>\>Remove $x$ from $V$ \\
6.\> \Return{} $\xorclauses' \backslash \xorclauses$%
\end{tabbing}%
}%
\vspace{-2mm}
\caption{The $\getransname$ translation}
\label{Fig:GETrans}
\end{figure}%

The translation \getransname{}, presented in Fig.~\ref{Fig:GETrans}, where $k$
stands for the maximum length of linear combinations to consider, ``eliminates''
each variable of the xor-constraint conjunction $\xorclauses$ at a time and
adds xor-constraints produced by the subroutine translation $
\PropTableName{}$, presented in Fig.~\ref{Fig:PropTable}.
Although the choice of variable to eliminate does not affect the correctness of
the translation, we employ a greedy heuristic to pick a variable that shares
xor-constraints with the fewest variables because the number of xor-constraints
produced in the subroutine \PropTableName{} is then the smallest.
The translation $\PropTable{Y, \psi, k}$ adds ``alias'' variables and
at most $O(2^{2k}) + |\xorclauses|$
%Simulating stronger parity reasoning with \getransname{} can increase the
%formula size considerably. The worst-case number of
%xor-constraints produced by translation $\PropTable{Y,\xorclauses,k}$ is $
%O(2^{2k}) + |\xorclauses|$.
%
xor-constraints to $\psi$ with the aim to simulate Gauss-Jordan row
operations involving at most $k$ variables in the
xor-constraints of the eliminated variable (the set $Y$) and no other
variables.
Provided that the maximum length of linear combinations to consider, the parameter $k$, is high
enough, the resulting xor-constraint conjunction
$\psi \wedge \PropTable{Y, \psi, k}$ has a {\it \UP{}-propagation table} for the set of variables $Y
\subseteq \VarsOf{\xorclauses}$, %denoted by $\HasPropTable{Y}{\psi}$,
meaning that the following conditions hold for all $Y', Y_1, Y_2 \subseteq Y$:
\begin{itemize}
\item[PT1:] There is an ``alias'' variable for every non-empty subset of $Y$:
  if $Y'$ is a non-empty subset of $Y$,
  then there is a variable $a \in \VarsOf{\psi}$ such that
  $ (a \oplus Y' \equiv \bot) $ is in $\psi$,
  where $(a \oplus Y' \equiv \bot)$ for $ Y' = \Set{y_1', \dots, y_n'} $
  means $ (a \oplus y_1' \oplus \dots \oplus y_n' \equiv \bot) $. 
%(lines 1-3 in Fig.~\ref{Fig:PropTable})
\item[PT2:] There is an xor-constraint for propagating the symmetric difference of any two subsets
of $Y$:
if $Y_1 \subseteq Y$ and $Y_2 \subseteq Y$,
   then there are variables $ a_1,a_2,a_3 \in \VarsOf{\psi}
  $ such that $ (a_1 \oplus Y_1 \equiv \bot),  (a_2 \oplus Y_2
           \equiv \bot ), (a_3 \oplus (Y_1 \oplus Y_2) \equiv \bot), $ and 
   $ (a_1 \oplus a_2 \oplus a_3 \equiv \bot) $
            are in $\psi $.
 
\item[PT3:] Alias variables of original xor-constraints having only variables of $Y$ are assigned: if $ (Y' \equiv \parity{})  $ is an xor-constraint in $\psi$ such that
$Y' \subseteq Y$, then there is a variable $ a \in \VarsOf{\psi}$ 
such that $ (a \oplus Y' \equiv \bot) $ it holds that
$ (a \equiv \parity{}) $ is in $ \psi $
   \end{itemize}

A \UP{}-propagation table for a set of variables $Y$ in $ \psi$ guarantees that
if some alias variables $a_1, \dots, a_n \in \VarsOf{\psi} $ 
binding the variable sets $ Y_1, \dots, Y_n \subseteq Y $ are assigned,
the alias variable $a \in \VarsOf{\psi}$
bound to the linear combination $ (Y_1 \oplus \dots \oplus Y_n) $ is \UP{}-deducible: $ \psi \wedge (a_1 \equiv \parity{1}) \wedge \dots \wedge (a_n \equiv \parity{n}) \UPderiv (a \equiv \parity{1} \oplus \dots \oplus \parity{n}) $.
Provided that sufficiently long linear combinations are considered (the
parameter $k$), \UP{}-propagation tables added by the $ \getransname$
enable unit propagation to always deduce all xor-implied literals, and thus
simulate a complete Gauss-Jordan propagation engine:

\begin{theorem}
\label{Thm:GESimulation}
If $\xorclauses$ is an xor-constraint conjunction, then $\getrans{\xorclauses}$
is a GE-simulation formula for $\xorclauses$ provided that 
$k =|\VarsOf{\xorclauses}|$.
\end{theorem}

\newcommand{\Alias}[1]{a_{#1}}
\newcommand{\XorClausesAt}[1]{\xorclauses^{(#1)}}

\begin{example}\label{Ex:kGE}
Consider the xor-constraint conjunction
$\XorClausesAt{0} =
  (x_1 \XX x_6 \XX x_7 \equiv \T) \land
  (x_2 \XX x_3 \XX x_7 \equiv \T) \land
  (x_2 \XX x_5 \XX x_8 \equiv \F) \land
  (x_3 \XX x_4 \XX x_5 \equiv \T) \land
  (x_4 \XX x_6 \XX x_8 \equiv \F)$
illustrated in Fig.~\ref{Fig:KgeExample}. 
It is clear that $\xorclauses \Models (x_1 \equiv \T) $ and
$ \xorclauses \NotUPderiv (x_1 \equiv \T)$.
\begin{figure}[tb]
\centering
\includegraphics[width=.44\textwidth]{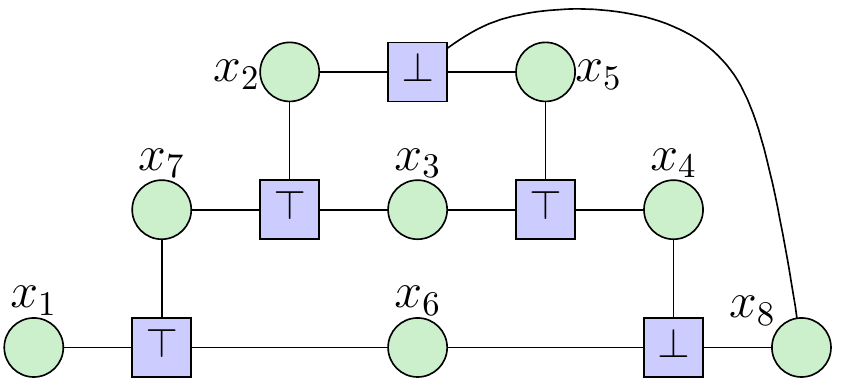}
\caption{Constraint graph of an xor-constraint conjunction}
\label{Fig:KgeExample}
\end{figure}

With the elimination order $(x_1,x_7,x_4,x_5,x_2,x_3,x_6,x_8)$ and $k=4$,
the translation $\getransname$ first extends $\xorclauses$ to $\XorClausesAt{1}$
with the xor-constraints in
$\PropTable{\Set{x_1,x_6,x_7},\xorclauses,k}$.
These include %, among many others,
(i)
the ``alias binding constraints''
$\Alias{1}\XX x_1 \Equal \F$,
$\Alias{6,7}\XX x_6 \XX x_7 \Equal \F$,
$\Alias{1,6,7}\XX x_1 \XX x_6 \XX x_7 \Equal \F$,
(ii)
the ``linear combination constraint''
$\Alias{1}\XX\Alias{6,7}\XX\Alias{1,6,7} \Equal \F$,
and
(iii)
the ``original constraint binder''
$\Alias{1,6,7} \Equal \T$,
where $\Alias{i,...}$ is the alias for the subset $\Set{x_i,...}$ of the original variables.
After unit propagation,
these constraints imply the binary constraint $\Alias{1}\XX\Alias{6,7}\Equal\T$
allowing us to deduce $x_1$
from the parity $\Alias{6,7}$ of $x_6$ and $x_7$.
%if we know the parity $\Alias{6,7}$ of $x_6$ and $x_7$.

\indent
The translation next ``eliminates'' $x_7$ and
adds $\PropTable{\Set{x_2,x_3,x_6,x_7},\XorClausesAt{1},k}$
including
the linear combination constraint
$\Alias{6,7} \XX \Alias{2,3,7} \XX \Alias{2,3,6} \Equal \F$ and
the original constraint binder
$\Alias{2,3,7} \Equal \T$,
propagating the binary constraint
$\Alias{6,7} \XX \Alias{2,3,6} \Equal \T$
allowing us to deduce the parity of $\Set{x_6,x_7}$
from the parity of $\Set{x_2,x_3,x_6}$.

\indent
Eliminating $x_4$ adds
$\PropTable{\Set{x_3, x_4, x_5, x_6, x_8},\XorClausesAt{2},k}$,
including the constraints
$\Alias{3,4,5}\XX\Alias{4,6,8}\XX\Alias{3,5,6,8}\Equal\F$,
$\Alias{3,4,5}\Equal\T$, and
$\Alias{4,6,8}\Equal\T$,
%unit
propagating %that
$\Alias{3,5,6,8}\Equal\T$.

\indent
Eliminating $x_5$ adds
$\PropTable{\Set{x_2,x_3,x_5,x_6,x_8},\XorClausesAt{3},k}$
(observe that $x_6$ is in the set as it occurs in the
constraint $\Alias{3,5,6,8}\XX x_3 \XX x_5 \XX x_6 \XX x_8 \Equal \F$
added in the previous step),
including
$\Alias{2,5,8}\XX\Alias{2,3,6}\XX\Alias{3,5,6,8}\Equal\F$
and
$\Alias{2,5,8}\Equal\F$.

\indent
At this point we could already unit propagate $x_1 \Equal \T$
(from $\Alias{3,5,6,8}\Equal\T$, $\Alias{2,5,8}\Equal\F$, and
 $\Alias{2,5,8}\XX\Alias{2,3,6}\XX\Alias{3,5,6,8}\Equal\F$
we get 
$\Alias{2,3,6}\Equal\T$
and
from this then $\Alias{6,7}\Equal\F$ and finally $\Alias{1}\Equal\T$,
i.e.\ $x_1 \Equal \T$).

\indent
Note that the translation $ \kgetrans{3}{\xorclauses}$ is not a GE-simulation
formula for $\xorclauses$ because \PropTableName{} does not add ``alias''
variables for any 4-subset of original variables and the linear combination
of any two original xor-constraints has at least four variables. 
%\end{figwindow}
\end{example}

%\begin{figwindow}[1,r,%
%\makebox[.55\textwidth][c]{\includegraphics[width=.53\textwidth]{GeneratedPlots/k-ge-vs-ge}},%
%{Xor-constraints in $\PropTable{Y,\xorclauses,k}$}\label{Fig:PtableSize}]
%
%Simulating stronger parity reasoning with \getransname{} can increase the
%formula size considerably. The worst-case number of
%xor-constraints produced by translation $\PropTable{Y,\xorclauses,k}$ is $
%O(2^{2k}) + |\xorclauses|$.
%and thus intractably large when $ k \geq 7 $ as
%shown in Fig.~\ref{Fig:PtableSize}.

%\indent 
The translation $ \PropTableName $ as presented in~\ref{Fig:PropTable}
for illustration purposes adds new ``alias'' variables for all relevant linear
combinations involving at most $k$ original variables. However, in an actual
implementation, the original variables of the xor-constraint conjunction can be
used as ``alias'' variables.
For example, the variable $x_1 $ in the xor-constraint
$(x_1 \oplus x_2 \oplus x_3 \equiv \top) $ can be used as an
``alias'' variable for $ (x_2 \oplus x_3 \equiv \F) $.

%\indent
The translation $ \getransname$ is a generalization of the translation
$\opttransname$, which simulates equivalence reasoning with unit propagation,
    presented in~\cite{LJN:CP2012}. Provided that original variables are
    treated as ``alias'' variables as above and all xor-constraints have at
    most three variables, the translation $ \kgetransname{2}$, that considers
    only (in)equivalences between pairs of variables, enables unit propagation
    to simulate equivalence reasoning.

%\indent
The size of the GE-simulation formula for $\xorclauses$ may be reduced considerably if $
\xorclauses$ is partitioned into disjoint xor-constraint conjunctions $
\xorclauses^1 \wedge \dots \wedge \xorclauses^n $ according to the connected
components of the xor-constraint graph, and then combining the component-wise
GE-simulation formulas $ \kgetrans{k_1}{\xorclauses^1} \wedge \dots \wedge
\kgetrans{k_n}{\xorclauses^n} $.
Efficient structural tests for deciding whether unit propagation or equivalence
reasoning is enough to achieve full propagation in an xor-constraint
conjunction, presented in~\cite{LJN:CP2012}, can indicate appropriate values for some of the parameters $k_1, \dots, k_n$. 
%

%\end{figwindow}

\input{simplify}

\subsection{Experimental evaluation}

\begin{figure}[tb]
\centering
\includegraphics[width=.37\textwidth]{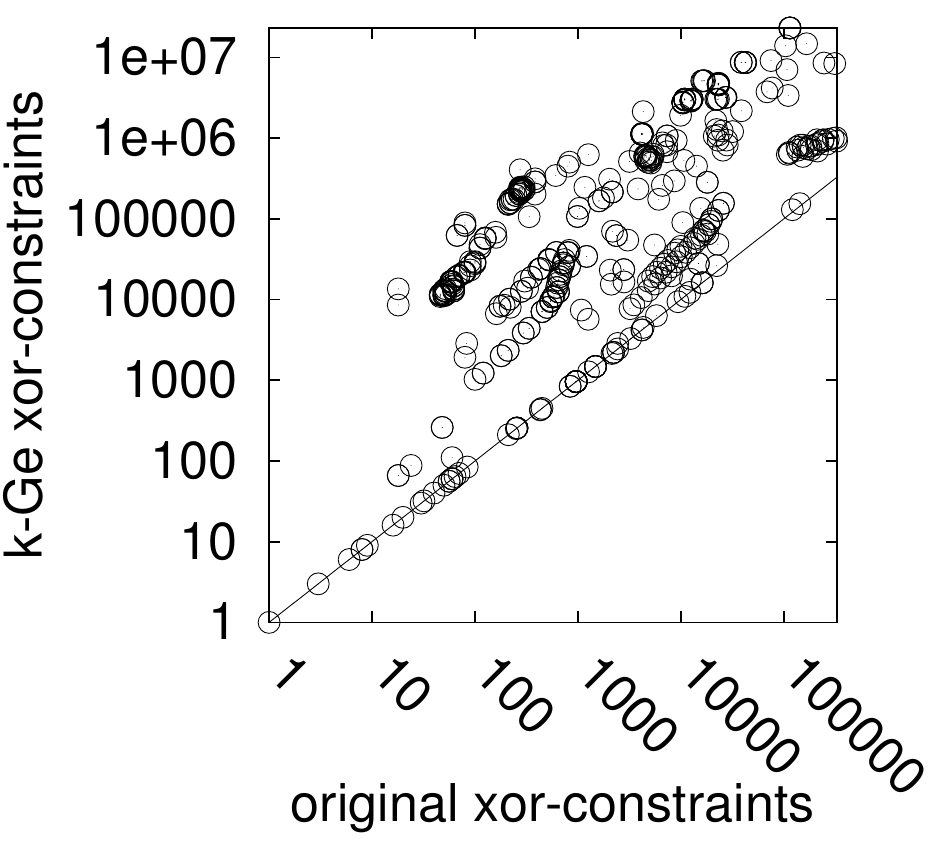}
\caption{Xor-constraints in SAT 05-11 instances}
\label{Fig:KgeSize}
\end{figure}
To evaluate the translation $\getransname$, we studied the benchmark instances
in ``crafted'' and ``industrial/application'' categories of the SAT
Competitions 2005, 2007, 2009, and 2011.
We ran \Tool{cryptominisat 2.9.6}, \Tool{glucose 2.3},
and \Tool{zenn 0.1.0} on the same 474 SAT Competition
cnf-xor instances as in Section~\ref{Sec:PexpResults} with the translations
\kgetransname{k} and \opttransname. %, and without any translation for
%comparison. 
%
It is intractable to simulate full Gauss-Jordan elimination for these
instances, so we adjusted the $k$-value of each call to the subroutine
$\PropTable{Y,\psi,k}$ to limit the number of additional xor-constraints.
The translation was computed for each connected component separately.
We found good performance by (i) stopping when $|Y| > 66$, (ii) setting $k=1$
when it was detected that unit propagation deduces all xor-implied
literals, (iii) setting $k=2$ when $|Y|\in[10,66]$ or when $ |Y| < 10 $
and it was detected that equivalence reasoning deduces all
xor-implied literals, (iv) setting $k=3$ when $|Y|\in[6,9]$, %6 \leq |Y| \leq 9 $
and (v)
setting $k=|Y|$ when $ |Y| \leq 5$.
With these parameters, the worst-case number of xor-constraints added by
the subroutine $ \PropTableName$ is 2145. Figure~\ref{Fig:KgeSize} shows
the increase in formula size by the translation \getransname.
Propagation-preserving xor-simplification was used to simplify the instances
reducing the formula size in 404 instances with the
median reduction being 16\%. 
The translation \opttransname was computed in a similar way.
The results are shown in Fig.~\ref{Fig:KgeResults2}, including
the time spent in computing the translations.
%
%The translation \opttransname to simulate equivalence reasoning presented
%in~\cite{LJN:CP2012} was computed
%until the worst-case number of new xor-constraints per ``eliminated'' 
%variable exceeds 2145.  
%
Using xor-simplification increases the number of solved instances for both translations.
%
%Both solvers solve most SAT'05 instances with the translation \opttransname and xor-simplification, and most SAT'07 instances when no
%translations are added. 
%
%The highest number of SAT'09 and SAT'11 instances are solved with the
%translation \getransname with xor-simplification.
%
The detailed solving time comparison in Fig.~\ref{Fig:KgeResults} shows that
that the translation \getransname{} can incur some overhead, but also allows great speedupds, enabling the three solvers to solve the
highest number of instances for the whole benchmark set. 
%\end{figwindow}

\begin{figure}[tb]
  \centering
\includegraphics[width=0.32\textwidth]{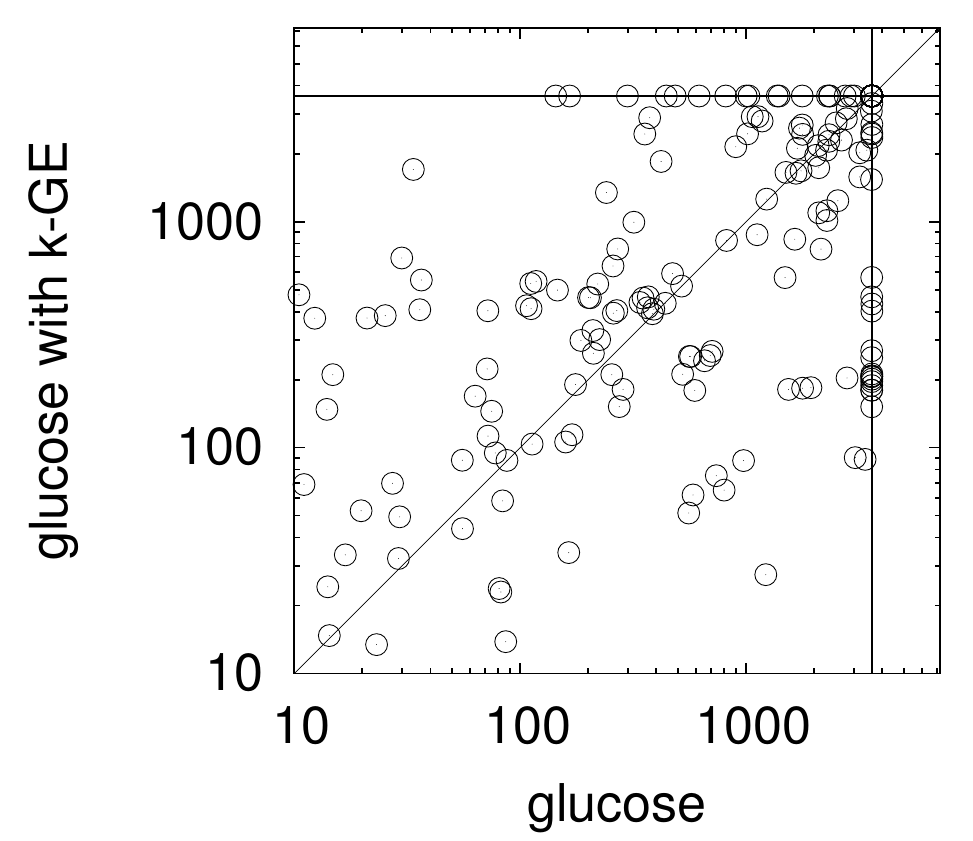} 
\includegraphics[width=0.32\textwidth]{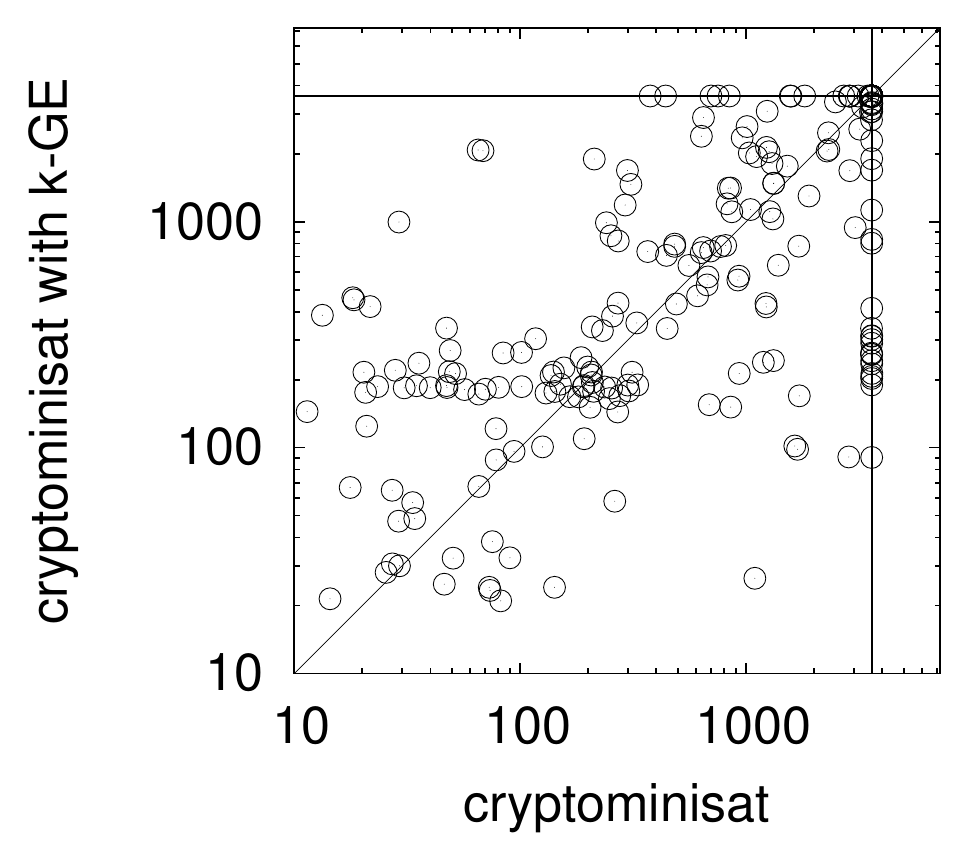} 
\includegraphics[width=0.32\textwidth]{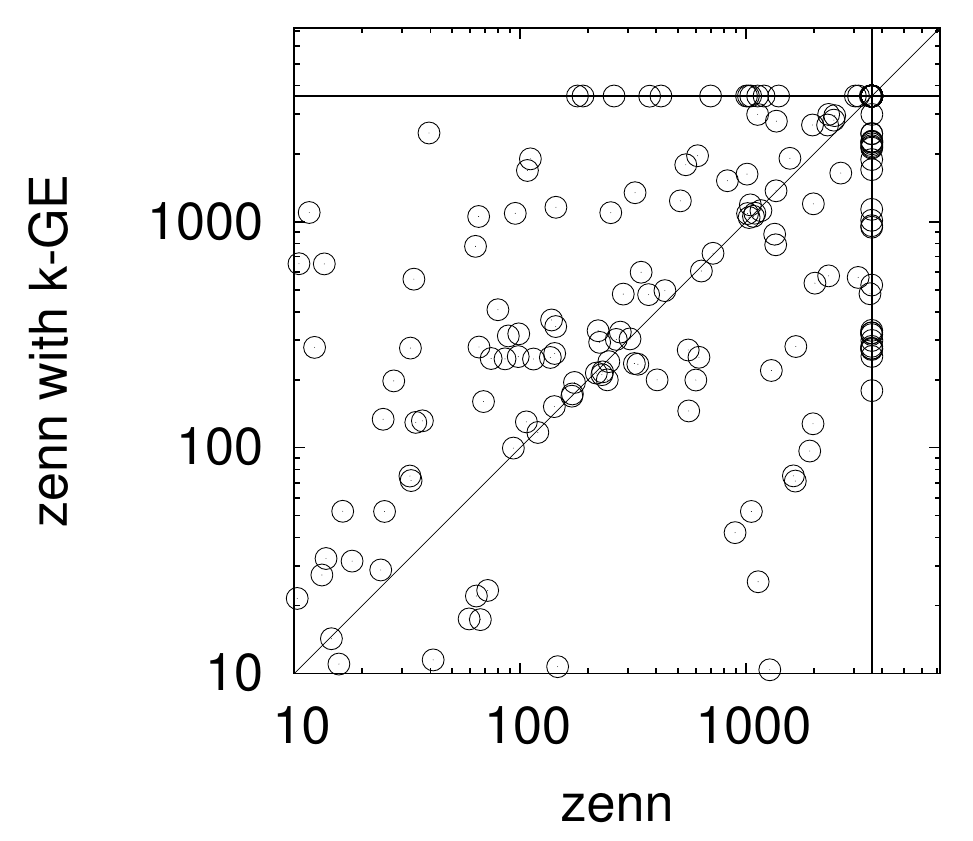} 
\caption{Comparison on solving time between the unmodified instance
           and \getransname{} using \Tool{glucose}, \Tool{cryptominisat}, and \Tool{zenn}. }
  \label{Fig:KgeResults}

\end{figure}
\begin{figure}[tb]
\centering
\begin{tabular}{cc}
    \begin{tabular}{|r@{\ }|@{\ }c@{\ }|@{\ }c@{\ }|@{\ }c@{\ }|@{\ }c@{\ }|c|}
      \hline
      & \multicolumn{5}{c|}{SAT Competition}\\
      & 2005 & 2007 & 2009 & 2011 & all \\
      \hline\hline
      instances                                     & 123         & 100         & 140         & 111 & 474\\
      \hline 
      \Tool{glucose}                                & 63          & \textbf{64} & 88          & 54  & 269       \\
      \Tool{glucose}, \opttransname                 & 64          & 59          & 89          & 52  & 264\\
      \Tool{glucose}, \opttransname, simp       & \textbf{66} & 63          & 90          & 52  & 271 \\
      \Tool{glucose}, \getransname                  & 61          & 50          & 86          & 45  & 242       \\ 
      \Tool{glucose}, \getransname, simp        & 64          & 60          & \textbf{95} & \textbf{58} & \textbf{277} \\ 
      \hline
      \Tool{cryptominisat}                          & 74          & \textbf{70} & 92          & 52   & 288      \\
      \Tool{cryptominisat}, \opttransname           & 73          & 65          & 91          & 49   & 278       \\
      \Tool{cryptominisat}, \opttransname, simp & \textbf{76} & 68          & 91          & 51   & 286       \\
      \Tool{cryptominisat}, \getransname            & 68          & 53          & 83          & 46   & 250       \\ 
      \Tool{cryptominisat}, \getransname, simp  & 71          & 65          & \textbf{94} & \textbf{64} & \textbf{294} \\ 
      \hline
      \Tool{zenn}                                   & 62          & 62          & 91          & 49          & 264    \\
      \Tool{zenn}, \opttransname                    & 62          & 62          & 90          & 49          & 263    \\
      \Tool{zenn}, \opttransname, simp              & \textbf{68} & \textbf{64} & 92          & 48          & 272    \\
      \Tool{zenn}, \getransname                    & 61          & 59          & 89          & 52          & 261    \\
      \Tool{zenn}, \getransname, simp               & 65          & 61          & \textbf{93} & \textbf{54} & \textbf{273} \\
      \hline
    \end{tabular}
%    &
%\includegraphics[width=.34\textwidth]{GeneratedPlots/orig-vs-ge-no-simplify} \\
%   (a) & (b) 
\end{tabular}

  \caption{Number of instances solved within the time limit of 3600s}
%(b) Xor-constraints in SAT 05-11 instances}
  \label{Fig:KgeResults2}
\end{figure}

%% file: simplify.tex
\subsection{Propagation-preserving xor-simplification}

Some of the xor-constraints added by \getransname{} can be redundant
regarding unit propagation.
We now present a simplification method that
%can reduce the size of an xor-constraint conjunction
%while preserving literals that can be implied by unit propagation.
preserves literals that can be implied by unit propagation.
There are two simplification rules,
given a pair of xor-constraint conjunctions $\Tuple{\phi_a, \phi_b}$
(initially $\Tuple{\xorclauses, \emptyset}$):
[S1] an xor-constraint $\XC$ in $\phi_a $ can be moved to $ \phi_b $,
resulting in $ \Tuple{\phi_a \setminus \Set{\XC}, \phi_b \cup \Set{\XC}}$, and
[S2] an xor-constraint $\XC$ in $ \phi_a $ can
    be simplified with an xor-constraint $\XC'$ in $ \phi_b $ to
$ (\XC \LC \XC') $ provided that $\Card{\VarsOf{\XC'} \cap \VarsOf{\XC}} \ge \Card{\VarsOf{\XC'}} - 1$,
resulting in 
$ \Tuple{(\phi_a \setminus \Set{\XC}) \cup \Set{\XC \LC \XC'},
         \phi_b} $.
%
%Applying the simplification rule preserves the literals implied by unit propagation.
%
\begin{theorem}
\label{Thm:XorSimp}
If $ \Tuple{\phi_a', \phi_b'} $ is the result of applying
one of the simplification rules to $ \Tuple{\phi_a, \phi_b} $ 
and 
$\phi_a \wedge \phi_b \wedge \AL_1 \wedge \dots \wedge
\AL_k \UPderiv \IL $, then $\phi_a' \wedge \phi_b' \wedge
\AL_1 \wedge \dots \wedge \AL_k \UPderiv \IL $.  
\end{theorem}

\begin{example}
The conjunction $\kgetrans{3}{(x_1 \XX x_2 \XX x_3 \XX x_4 \Equal \F)}$ contains the alias binding constraints
$\XC_1 \Def (\Alias{1,2,3,4} \XX x_{1} \XX x_2 \XX x_3 \XX x_4 \Equal \F)$,
$\XC_2 \Def (\Alias{1,2} \XX x_1 \XX x_2 \Equal \F) $, 
$\XC_3 \Def (\Alias{3,4} \XX x_3 \XX x_4 \Equal \F)$,
as well as
the linear combination constraint $\XC_4 \Def (\Alias{1,2} \XX a_{3,4} \XX a_{1,2,3,4} \Equal \F)$.
The alias binding constraint $\XC_1$ can in fact be eliminated 
by first applying the rule S1 to the xor-constraints $\XC_2$, $\XC_3$, and $\XC_4$.
Then, by using the rule S2,
the xor-constraint
$\XC_1$ is simplified first with $\XC_2$
to $(\Alias{1,2,3,4} \XX \Alias{1,2} \XX x_3 \XX x_4 \Equal \F)$
and
then with $\XC_3$ to $(\Alias{1,2,3,4} \XX \Alias{1,2} \XX \Alias{3,4} \Equal \F)$,
and finally with $\XC_4$ to $(\F \Equal \F)$.
\end{example}

%% file: treewidth.tex
\section{Connection to Treewidth}

The number of xor-constraints produced by the translation $\getransname$
depends strongly on the instance, as shown in Fig.~\ref{Fig:KgeSize}.
Now we connect the worst-case size of a \PropTableName{}-based GE-simulation
formula to {\it treewidth}, a well-known structural property of (constraint)
graphs used often to characterize the hardness of solving a problem, e.g.
an instance of CSP with bounded treewidth can be solved in polynomial time~\cite{Freuder:1990:CKS:1865499.1865500}.
We develop a new decomposition method that we can apply to a tree decomposition
to produce a polynomial-size GE-simulation formula for instances of bounded
treewidth.
We also present some found upper bounds for treewidth in SAT Competition
instances that illustrate to what extent parity reasoning can be simulated
through unit propagation.

\input{decomposition}

\begin{figure}
\centering
\includegraphics[width=.3\textwidth]{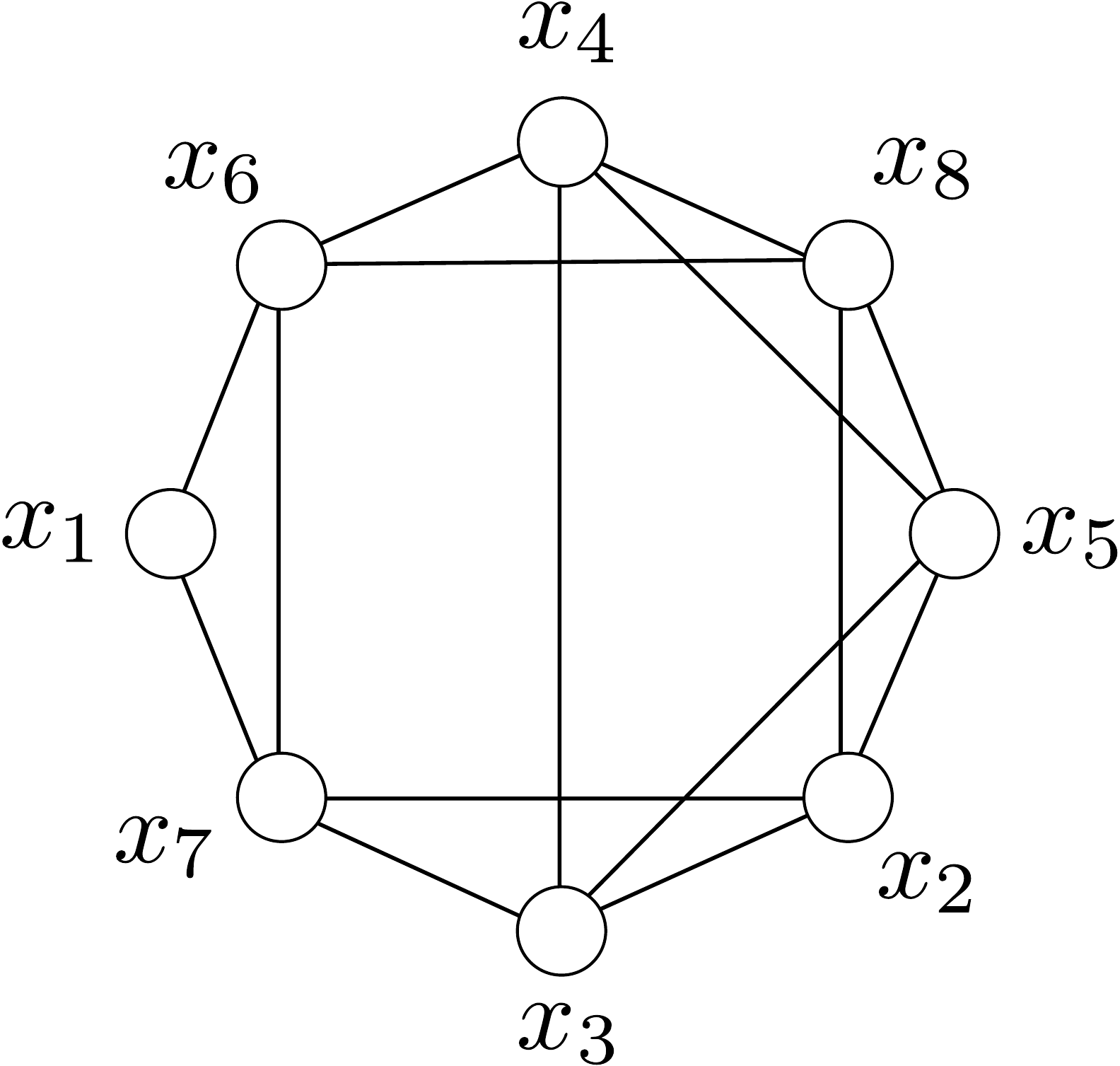}
\caption{Primal graph}
\label{Fig:KgePrimal}
\end{figure}

Now we apply the decomposition method to a tree decomposition to produce a
polynomial-size GE-simulation formula for instances of bounded treewidth.
Formally, a {\it tree decomposition}  
of a graph
$G=\Tuple{V,E}$ is
a pair $\Tuple{X, T}$, where $X = \Set{X_1, \dots, X_n}$ is a family of subsets
of $V$, and $T$ is a tree whose nodes are the subsets $X_i$, satisfying the following
properties: 
(i) $V = X_1 \cup \dots \cup X_n $, 
(ii) if $ \Tuple{v,v'} \in E$, then it holds for at least one $X_i \in X$, that $ \Set{v,v'} \subseteq X_i $, 
and (iii)
if a node $v$ is in two sets $X_i$ and $X_j$, then all nodes in the path between $X_i$ and $X_j$ contain $v$.
\newcommand{\TW}{\operatorname{tw}}
The width of a tree decomposition is the size of its largest set $X_i$ minus
one. The {\it treewidth} $\TW(G) $ of a graph $G$ is the minimum width among all
possible tree decompositions of $G$. 

\begin{figure}
\centering
\includegraphics[width=.3\textwidth]{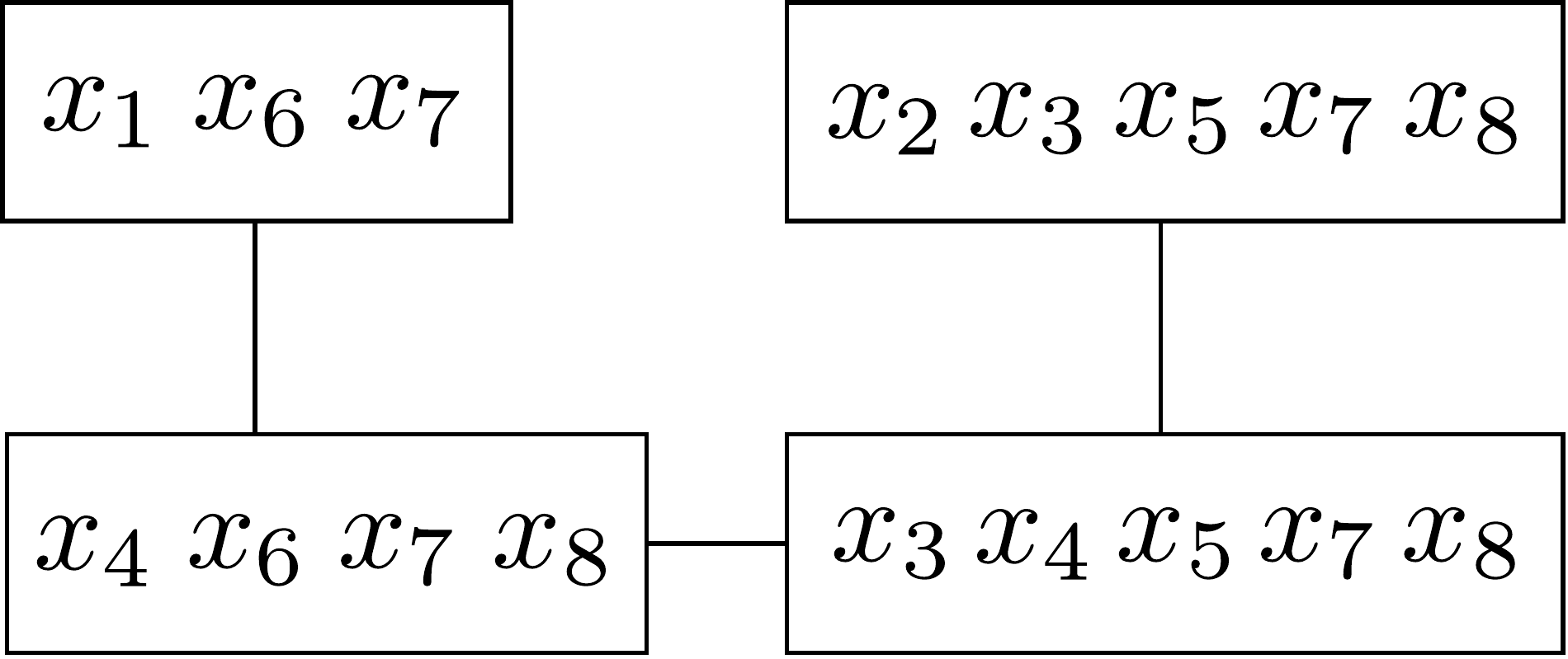}
\caption{A tree decomposition of primal graph in Fig.~\ref{Fig:KgePrimal}}
\label{Fig:KgeDecomp}
\end{figure}
Each pair of adjacent nodes in a tree decomposition defines a cut variable set,
so it suffices to add a \UP{}-propagation table for each node's variable
set. 
The {\it primal graph} for an xor-constraint conjunction $\xorclauses$ is a
graph such that the nodes correspond to the variables of $\xorclauses$ and
there is an edge between two variable nodes if and only if both variables have
an occurrence in the same xor-constraint.
The primal graph of the xor-constraint conjunction shown in Fig.~\ref{Fig:KgeExample} and a
tree decomposition for it are shown in Fig.~\ref{Fig:KgePrimal} and in
Fig.~\ref{Fig:KgeDecomp}.  
If an xor-constraint conjunction has a bounded treewidth, the tree
decomposition can be used to construct a polynomial-size GE-simulation formula:

\begin{theorem}
\label{Thm:TreeDecomposition}
If $\Set{X_1, \dots, X_n} $ is the family of variable sets in the tree
decomposition of the primal graph of an xor-constraint conjunction $
\xorclauses$ and 
$\phi_0, \dots, \phi_n$ is a sequence of xor-constraint conjunctions
such that $ \phi_0 = \xorclauses $ and $ \phi_i = \phi_{i-1} \wedge \PropTable{X_i, \phi_{i-1}, |X_i|} $ for $i \in \Set{1,\dots,n}$, then $ \phi_n \setminus \xorclauses $ is a GE-simulation formula for
$\xorclauses$ with $O(n {2^{2k}}) + |\xorclauses|$ xor-constraints, where
$k = \max(|X_1|, \dots, |X_n|) $.  
\end{theorem}

\begin{figure}
\centering
\includegraphics[width=0.8\textwidth]{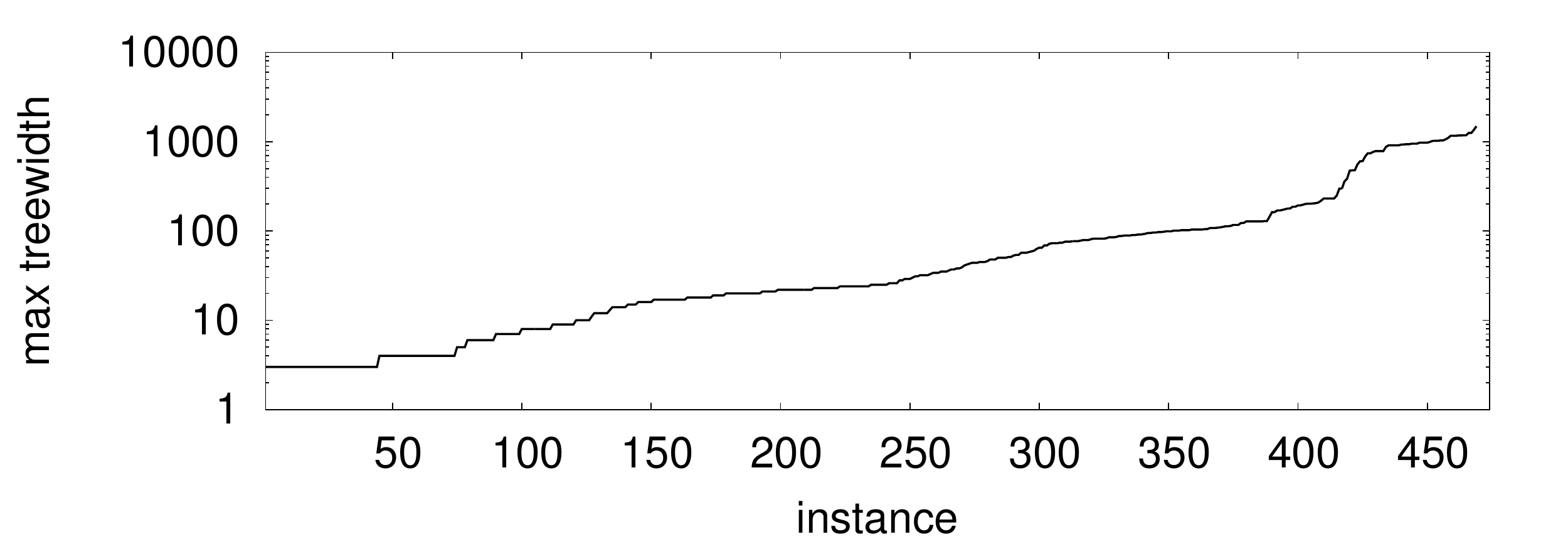}
\caption{Treewidth in SAT 05-11 instances}
\label{Fig:SatTreewidths}
\end{figure}

To find out to what extent unit propagation can simulate stronger parity
reasoning, we studied the 474 benchmark instances in ``crafted'' and
``industrial/application'' categories of the SAT Competitions 2005, 2007, 2009,
and 2011.
Computing the exact value of treewidth is an NP-complete problem~\cite{Arnborg:1987:CFE:37170.37183}, so we applied
the junction tree algorithm described in~\cite{DBLP:conf/aaai/Pearl82}
to get an upper bound for
treewidth.  
%
%On three instances, the algorithm ran out of memory. 
The found treewidths are shown in Fig.~\ref{Fig:SatTreewidths}.
%and the numbers
%of components with treewidth higher than three are shown in Fig.~\ref{Fig:SatNontreelike}. 
%If the treewidth of the primal graph is three, then the constraint graph is tree-like.
%
%\begin{figwindow}[2,r,%
%\makebox[.490\textwidth][c]{\includegraphics[width=.490\textwidth]{GeneratedPlots/sat-nontreelike}},%
%{Non-tree-like components in SAT 05-11 instances
%    \label{Fig:SatNontreelike}}]
%
There are some instances that have compact GE-simulation formulas, but
for the majority of the instances, full GE-simulation formula is likely
to be intractably large.
For these instances a powerful solution technique can be to choose a suitable
propagation method for each biconnected component separately, either through a
translation or an xor-reasoning module.
%\end{figwindow}

%% file: decomposition.tex
The new decomposition technique is a generalization of the
%decomposition
method  in~\cite{LJN:ICTAI2012},
which states that, 
in order to guarantee full
propagation, it is enough to (i) propagate only values through ``cut
variables'', and (ii) have full propagation for the ``biconnected components''.
Now we extend the technique to larger cuts.
Given an xor-constraint conjunction $\xorclauses$,
      a {\it cut variable set} is a set of variables $X \subseteq \VarsOf{\xorclauses}$ for which there is a partition $ (\CutA, \CutB) $
      of xor-constraints in $\xorclauses$ with $\VarsOf{\CutA} \cap
      \VarsOf{\CutB} = X$; such a partition $ (\CutA, \CutB)$ 
      is called an {\it $\Vars$-cut partition} of $\xorclauses$.  
If full propagation can be guaranteed for both sides of an $ \Vars$-cut
partition, then communicating the implied linear combinations involving cut
variables is enough to guarantee full propagation for the whole instance:
\newcommand{\xorclausesA}{\xorclauses^{\textup{a}}}
\newcommand{\xorclausesB}{\xorclauses^{\textup{b}}}
\begin{theorem}
\label{Thm:XCut}
Let $ (\CutA, \CutB) $ be an $\Vars$-cut partition of $\xorclauses$.
Let $ \xorclausesA = \bigwedge_{D \in \CutA} D$, $\xorclausesB = \bigwedge_{D \in \CutB} D$, and $\AL_1,\dots,\AL_k \in \LitsOf{\xorclauses}$.
 Then
it holds that:
\begin{itemize}
\item If $\xorclauses \wedge \AL_1 \wedge \dots \wedge \AL_k $ is unsatisfiable, then
\begin{enumerate}
\item $ \xorclausesA \wedge \AL_1 \wedge \dots \wedge \AL_k $ 
or $ \xorclausesB \wedge \AL_1 \wedge \dots \wedge \AL_k $ is unsatisfiable; or 
\item $\xorclausesA \wedge \AL_1 \wedge \dots \wedge \AL_k \Models 
(X' \equiv \parity{}')$ and $ \xorclausesB \wedge \AL_1 \wedge \dots
\AL_k \Models (X' \equiv \parity{}' \oplus \top) $ for some $ X' \subseteq X $
and $ \parity{}' \in \set{\top, \bot} $.
\end{enumerate}
\item 
If $\xorclauses \wedge \AL_1 \wedge \dots \wedge \AL_k $ is satisfiable
and $ \xorclauses \wedge \AL_1 \wedge \dots \wedge \AL_k \Models (Y \equiv \parity{}) $ for some $Y \subseteq \VarsOf{\xorclauses^\alpha}$, $Y \cap (\VarsOf{\xorclauses^\beta} \setminus \VarsOf{\xorclauses^\alpha}) = \emptyset$, and $\parity{} \in \Set{\top, \bot}$ where $\alpha \in \Set{\textup{a},\textup{b}}$ and $ \beta \in \Set{\textup{a}, \textup{b}} \setminus \Set{\alpha} $, then
\begin{enumerate}
\item $\xorclausesA \wedge \AL_1 \wedge \dots \wedge \AL_k \Models (Y \equiv \parity{}) $
or $ \xorclausesB \wedge \AL_1 \wedge \dots \wedge \AL_k \Models (Y \equiv \parity{}) $;
or
\item $\xorclauses^\alpha \wedge \AL_1 \wedge \dots \wedge \AL_k \Models 
( X' \equiv \parity{}') $ and $ \xorclauses^\beta \wedge \AL_1 \wedge \dots \wedge \AL_k \wedge (X' \equiv \parity{}') \Models (Y \equiv \parity{}) $ for some $X' \subseteq X$,
$\parity{}' \in \set{\top,\bot}$, $\alpha \in \Set{\textup{a},\textup{b}}$, and $\beta \in \Set{\textup{a},\textup{b}}\setminus\Set{\alpha}$.
%\item $\xorclauses^a \wedge \AL_1 \wedge \dots \wedge \AL_k \Models 
%( X' \equiv \parity{}) $ and $ \xorclauses^b \wedge \AL_1 \wedge \dots \wedge \AL_k \wedge (X' \equiv \parity{}) \Models \IL $ for some $ X' \subseteq X$ and $\parity{} \in \set{\top,\bot}$.
%\item $\xorclauses^b \wedge \AL_1 \wedge \dots \wedge \AL_k \Models 
%( X' \equiv \parity{}) $ and $ \xorclauses^a \wedge \AL_1 \wedge \dots \wedge \AL_k \wedge (X' \equiv \parity{}) \Models \IL $ for some $ X' \subseteq X$ and $\parity{} \in \set{\top,\bot}$.
\end{enumerate}
\end{itemize}
\end{theorem}

\begin{example}
Consider the constraint graph in Fig.~\ref{Fig:KgeExample}. The cut variable set $ \Set{x_2, x_3, x_6} $ partitions the xor-constraints into 
two conjunctions $ \xorclausesA = (x_1 \oplus x_6 \oplus x_7 \equiv \top) \wedge (x_2 \oplus x_3 \oplus x_7 \equiv \top) $ and $\xorclausesB = (x_2 \oplus x_5 \oplus x_8 \equiv \bot) \wedge (x_3 \oplus x_4 \oplus x_5 \equiv \top) \wedge (x_4 \oplus x_6 \oplus x_8 \equiv \bot) $. It holds that $\xorclausesB \Models (x_2 \oplus x_3 \oplus x_6 \equiv \top) $ and $\xorclausesA \wedge (x_2 \oplus x_3 \oplus x_6 \equiv \top) \Models (x_1 \equiv \top) $.
\end{example}

%% file: conclusions.tex
\section{Conclusions}
We have studied how stronger parity reasoning techniques in the DPLL(XOR) framework can be simulated by simpler systems.
We have shown that resolution simulates equivalence reasoning efficiently. 
We have proven that parity explanations on nondeterministic unit propagation
derivations can simulate Gauss-Jordan elimination on
a restricted yet practically relevant class of instances.
We have shown that Gauss-Jordan elimination can be simulated by unit
propagation by adding additional xor-constraints, and for instance families of
bounded treewidth, a polynomial number of additional xor-constraints
suffices.%, and
%
%the bounded version of the translation can lead to significant speedupds.
%
%These generalizable results can be useful in understanding how xor-constraints can be handled in next generation SAT solvers.

%Resolution, which is equivalent to the underlying proof system of modern
%

%% file: proofs.tex
\newpage
\section{Proofs}
\newenvironment{relemma}[1]{\renewcommand{\thelemma}{#1}\begin{lemma}}{\end{lemma}}
\newenvironment{retheorem}[1]{\renewcommand{\thetheorem}{#1}\begin{theorem}}{\end{theorem}}

\subsection{Fundamental Properties of Linear Combinations}

Some fundamental, easy to verify properties are
$\XC \LinComb \XC \LinComb \XCB = \XCB$,
${\XC \land \XCB} \Models {\XC \LinComb \XCB}$,
${\XC \land \XCB} \Models {\XC \land (\XC \LinComb \XCB)}$,
and
${\XC \land (\XC \LinComb \XCB)} \Models {\XC \land \XCB}$.

The logical consequence xor-constraints of an xor-constraint conjunction $\psi$
are exactly those that are linear combinations of the xor-constraints in $\psi$:
%\cite{LJN:ICTAI2012full}:
\begin{lemma}[from \cite{LJN:ICTAI2012full}]\label{Lemma:LinearCombs}
  Let $\psi$ be a conjunction of xor-constraints.
  Now $\psi$ is unsatisfiable if and only if
  there is a subset $S$ of xor-constraints in $\psi$ such that
  $\BigLinComb_{\XC \in S} \XC = (\F \Equal \T)$.
  If $\psi$ is satisfiable and $\XCB$ is an xor-constraint,
  then $\psi \Models \XCB$ %for some xor-constraint $\XCB$
  if and only if
  there is a subset $S$ of xor-constraints in $\psi$ such that
  $\BigLinComb_{\XC \in S} \XC = \XCB$.
\end{lemma}
%

%--------------------------------------------------------------------------
%
%--------------------------------------------------------------------------
%
%--------------------------------------------------------------------------
\subsection{Proof of Theorem~\ref{Thm:ResDeriv}}

\newcommand{\ResLabFunc}{\Labfunc_\ResDer}
\newcommand{\ResLab}[1]{\ResLabFunc(#1)}

\begin{retheorem}{\ref{Thm:ResDeriv}}
  Assume a \SUBST-derivation $\igraph = \Tuple{\vertices,\edges,\Labfunc}$ on
  a conjunction $\psi$ of xor-constraints.
  There is a resolution derivation $\ResDer$ on
  $\bigwedge_{\XC \in \psi}\cnf{\XC}$
  such that
  (i)
  if $\Vertex \in \vertices$ and $\Lab{\Vertex} \neq \T$,
  then the clauses $\cnf{\Lab{\Vertex}}$ occur in $\ResDer$,
  and
  (ii)
  $\ResDer$ has at most $\Card{\vertices}2^{m-1}$ clauses,
  where $m$ is the number of variables in the largest xor-constraint in $\psi$.
\end{retheorem}
\begin{proof}
We construct a resolution derivation $\ResDer$ with the desired properties
by first inductively associating each vertex $\Vertex$ in $\igraph$
with a set $\ResLab{\Vertex}$ of clauses such that
%
%We construct a resolution derivation $\ResDer$ with the desired properties
%by structural induction on the \SUBST-derivation $\igraph = \Tuple{\vertices,\edges,\Labfunc}$.
%
%Starting from the input vertices,
%we associate each vertex $\Vertex$ in $\igraph$
%a set $\ResLab{\Vertex}$ of clauses such that
\begin{enumerate}
\item
  if $\Lab{\Vertex} \neq \T$, then all the clauses in $\cnf{\Lab{\Vertex}}$ occur in $\ResLab{\Vertex}$,
\item
  all the clauses in $\ResLab{\Vertex}$ either occur in the CNF translation $\bigwedge_{\XC \in \psi}\cnf{\XC}$ or
  can be obtained with one resolution step
  (i) from the ones in $\ResLab{\Vertexp}$ and $\ResLab{\Vertexpp}$,
  where $\Vertexp$ and $\Vertexpp$ are the source vertices of the two edges incoming to $\Vertex$,
  or
  (ii) from the ones produced as in (i) above,
and
\item
  $\Card{\ResLab{\Vertex}} \le 2^{m-1}$.
\end{enumerate}
A resolution derivation can be obtained directly from this construction
by just listing the clauses in the $\ResLabFunc$-sets in appropriate order.

For each input vertex $\Vertex$
we have that $\Lab{\Vertex}$ occurs in $\psi$ and
thus we simply set $\ResLab{\Vertex} = \Setdef{C}{\text{$C$ occurs in $\cnf{\Lab{\Vertex}}$}}$.

%Base case.
%%
%Let $\vertices_0$ be the input vertices of $\igraph$.
%%
%We construct a resolution derivation $\ResDer_0$ for the sub-graph of $\igraph$
%induced by $\vertices_0$ by simply taking $\ResDer_0$ to consist of
%all the clauses in $\bigwedge_{\Vertex \in \vertices}\cnf{\Lab{\Vertex}}$.
%%
%All these clauses are in $\bigwedge_{\XC \in \psi}\cnf{\XC}$ and
%are thus allowed to occur in the resolution derivation.

%Induction hypothesis:
%the theorem holds for a sub-graph of $\igraph$ that is closed under
%the predecessor relation, i.e.~\

For non-input vertices, we apply the following construction.
\begin{enumerate}
\item
  $\Lab{\Vertex}$ is obtained from $\Lab{\Vertexp}$ and $\Lab{\Vertexpp}$
  by using $\unitruleP$.

  Suppose that $\Lab{\Vertexp} = (\Var \Equal \T)$ for some variable $\Var$.
  Thus $\ResLab{\Vertexp} \supseteq \Set{(\Var)}$.
  \begin{itemize}
  \item
    If $\Lab{\Vertexpp} = (\Var \Equal \T)$,
    then $\Lab{\Vertex} = \T$ and we set $\ResLab{\Vertex} = \emptyset$.
  \item
    If $\Lab{\Vertexpp} = (\Var \Equal \F)$,
    then
    $\ResLab{\Vertexpp} \supseteq \Set{(\neg \Var)}$,
    $\Lab{\Vertex} = \F$
    and
    we set $\ResLab{\Vertex}$ to be the resolvent of
    $(\Var) \in \ResLab{\Vertexp}$
    and
    $(\neg \Var) \in \ResLab{\Vertexp}$.
    Thus $\ResLab{\Vertex} = () = \cnf{\F \Equal \T}$.
  \item
    If $\Lab{\Vertexpp} = (x \X y \X ... \Equal \parity{})$,
    then
    $\ResLab{\Vertexpp} \supseteq \Setdef{C}{\text{$C$ occurs in $\cnf{\Lab{\Vertexpp}}$}}$,
    $\Lab{\Vertex} = \simplification{\Lab{\Vertexpp}}{\Var}{\T}$
    and
    we set $\ResLab{\Vertex}$ to be the set of all clauses obtained by
    resolving
    $(\Var) \in \ResLab{\Vertexp}$
    with the clauses of form $(\neg\Var \lor ...)$
    occurring in $\cnf{\Lab{\Vertexpp}}$.
    One can verify that indeed
    $\ResLab{\Vertex} = \Setdef{C}{\text{$C$ occurs in $\cnf{\Lab{\Vertex}}$}}$.
  \end{itemize}
\item
  $\Lab{\Vertex}$ is obtained from $\Lab{\Vertexp}$ and $\Lab{\Vertexpp}$
  by using $\unitruleN$.

  This case is similar to the previous one.
\item
  $\Lab{\Vertex}$ is obtained from $\Lab{\Vertexp}$ and $\Lab{\Vertexpp}$
  by using $\eqvruleP$.

  Suppose that $\Lab{\Vertexp} = (\Var \X \AnotherVar \Equal \F)$ for some variables $\Var$ and $\AnotherVar$.
  Thus $\ResLab{\Vertexp} \supseteq \Set{(\neg\Var \lor \AnotherVar),(\Var \lor \neg\AnotherVar)}$
  and
  $\Lab{\Vertex} = \simplification{\Lab{\Vertexpp}}{\Var}{\AnotherVar}$.

  \begin{itemize}
  \item
    If $\Lab{\Vertexpp} = (\Var \X ... \Equal \parity{})$
    such that $\AnotherVar$ does not occur in it,
    then we set
    $\ResLab{\Vertex}$ to consist of all the clauses obtained by
    (i) resolving $(\neg\Var \lor \AnotherVar)$ with each clause of form
    $(\Var \lor ...)$ occurring in $\cnf{\Lab{\Vertexpp}}$
    [and thus also in $\ResLab{\Vertexpp}$],
    and
    (ii) resolving $(\Var \lor \neg\AnotherVar)$ with each clause of form
    $(\neg\Var \lor ...)$ occurring in $\cnf{\Lab{\Vertexpp}}$
    [and thus also in $\ResLab{\Vertexpp}$].
    It is straightforward to verify that
    $\ResLab{\Vertex} = \Setdef{C}{\text{$C$ occurs in $\cnf{\Lab{\Vertex}}$}}$.
  \item
    If $\Lab{\Vertexpp} = (\Var \X \AnotherVar \Equal \F)$,
    then $\Lab{\Vertex} = \T$ and we set $\ResLab{\Vertex}=\emptyset$.
  \item
    If $\Lab{\Vertexpp} = (\Var \X \AnotherVar \Equal \T)$,
    then $\ResLab{\Vertexpp}\supseteq\Set{(\Var \lor \AnotherVar),(\neg\Var\lor\neg\AnotherVar)}$,
    $\Lab{\Vertex} = \F$ and
    we set $\ResLab{\Vertex}=\Set{(y),(\neg y),()}$ [all these clauses can be obtained with resolution from the ones in $\ResLab{\Vertexp}$ and $\ResLab{\Vertexpp}$].
  \item
    If $\Lab{\Vertexpp} = (\Var \X \AnotherVar \X \ThirdVar_1 \X ... \X \ThirdVar_k \Equal \parity{})$,
    then we first resolve
    (i)
    $(\neg x \lor y)$ with each of the $2^{k-1}$ clauses
    in $\Setdef{(x \lor y \lor C)}{\text{$C$ occurs in $\cnf{\ThirdVar_1 \X ... \X \ThirdVar_k \Equal \parity{}}$}}$
    and
    (ii)
    $(x \lor \neg y)$ with each of the $2^{k-1}$ clauses
    in $\Setdef{(\neg x \lor \neg y \lor C)}{\text{$C$ occurs in $\cnf{\ThirdVar_1 \X ... \X \ThirdVar_k \Equal \parity{}}$}}$.
    We then resolve, for each $C$ occurring in $\cnf{\ThirdVar_1 \X ... \X \ThirdVar_k \Equal \parity{}}$,
    the clauses $(y \lor C)$ and $(\neg y \lor C)$ obtained above;
    the result is the $2^{k-1}$ clauses in
    $\cnf{\ThirdVar_1 \X ... \X \ThirdVar_k \Equal \parity{}}$,
    as required to represent $\Lab{\Vertex} = (\ThirdVar_1 \X ... \X \ThirdVar_k \Equal \parity{})$.
  \end{itemize}
\item
  $\Lab{\Vertex}$ is obtained from $\Lab{\Vertexp}$ and $\Lab{\Vertexpp}$
  by using $\eqvruleN$.

  This case is similar to the previous one.
\end{enumerate}
\qed
\end{proof}

%--------------------------------------------------------------------------
%
%--------------------------------------------------------------------------
%
%--------------------------------------------------------------------------
\subsection{Proof of Theorem~\ref{Thm:ResExp}}

\newcommand{\Trigger}{\phi}
\newcommand{\Triggerp}{\phi'}
\newcommand{\CS}{S}
\newcommand{\CSp}{S'}
\newcommand{\CSpp}{S''}

We start by giving some auxiliary results and lemmas.

First, observe that
$(\Var\Equal\T)\LinComb\XC = \simplification{\XC}{\Var}{\T}$,
$(\Var\Equal\F)\LinComb\XC = \simplification{\XC}{\Var}{\F}$,
$(\Var\X\AnotherVar\Equal\F)\LinComb\XC = \simplification{\XC}{\Var}{\AnotherVar}$,
and
$(\Var\X\AnotherVar\Equal\T)\LinComb\XC = \simplification{\XC}{\Var}{\AnotherVar\X\T}$
when $\Var$ occurs in $\XC$ and thus the \SUBST-rules in Fig.~\ref{Fig:SUBST}
are special cases of a more general linear combination rule.

The next lemmas show that these special cases,
when conditioned with some conjunctions of literals,
can be derived with resolution with a linear number of steps.
\begin{lemma}\label{Lemma:UnitExplanation}
Let $\Trigger$ and $\Triggerp$ be conjunctions of literals
and take some xor-constraints
$\XC = (\Var \Equal \parity{})$ and
$\XCp = (\Var \X \ThirdVar_1 \X ... \X \ThirdVar_k \Equal \Parityp)$.
Given the sets
$\CS = \Setdef{\Trigger \Implies C}{C \in \cnf{\XC}}$
and
$\CSp = \Setdef{\Triggerp \Implies C}{C \in \cnf{\XCp}}$
of clauses,
the set
$\Setdef{(\Trigger \land \Triggerp)\Implies C}{C \in \cnf{\XC \LinComb \XCp}}$
has $2^{k-1}$ clauses and
we can derive them from those in $\CS$ and $\CSp$
with $2^{k-1}$ resolution steps.
\end{lemma}
\begin{proof}
Take the only clause $l_1 \land ... \land l_m \Implies (\Var \Equal \parity{})$ in $\CS$ and
resolve it with each clause
$l'_1 \land ... \land l'_n \Implies C$ with
$C = ((\Var \Equal \neg\parity{}) \lor ...) \in \cnf{\XCp}$
in $\CSp$ (there are $2^{k-1}$ of them).
Each resulting clause forces that
either
(i) one of the literals $l_1,...,l_m,l'_1,...,l'_n$ is false
or
(ii) that the parity of $\ThirdVar_1,...,\ThirdVar_k$ is not
one of the $2^{k-1}$ ones not allowed by 
$(\parity{} \X \ThirdVar_1 \X ... \X \ThirdVar_k \Equal \Parityp) = 
 \XC \LinComb \XCp$.
\qed
\end{proof}
\begin{lemma}\label{Lemma:EqExplanation1}
Let $\Trigger$ and $\Triggerp$ be conjunctions of literals
and take some xor-constraints
$\XC = (\Var \X \AnotherVar \Equal \parity{})$ and
$\XCp = (\Var \X \ThirdVar_1 \X ... \X \ThirdVar_k \Equal \Parityp)$.
Given the sets
$\CS = \Setdef{\Trigger \Implies C}{C \in \cnf{\XC}}$
and
$\CSp = \Setdef{\Triggerp \Implies C}{C \in \cnf{\XCp}}$
of clauses,
the clause set
$\Setdef{(\Trigger \land \Triggerp)\Implies C}{C \in \cnf{\XC \LinComb \XCp}}$
has $2^k$ clauses and
we can derive them from those in  $\CS$ and $\CSp$
with $2^k$ resolution steps.
\end{lemma}
\begin{proof}
Take the clause $l_1 \land ... \land l_m \Implies ((\Var \Equal \T) \lor (\AnotherVar \Equal \parity{}))$ in $\CS$ and
resolve it with each clause
$l'_1 \land ... \land l'_n \Implies C$ with
$C = ((\Var \Equal \F) \lor ...) \in \cnf{\XCp}$
in $\CSp$ (there are at most $2^{k-1}$ of them).
Each resulting clause
 $l_1 \land ... \land l_m \land l'_1 \land ... \land l'_n \Implies ((\AnotherVar \Equal \parity{}) \lor ...)$
forces that
either
(i) one of the literals $l_1,...,l_m,l'_1,...,l'_n$ is false,
or
(ii) $\AnotherVar \Equal \neg \parity{}$
implies that the parity of $\ThirdVar_1,...,\ThirdVar_k$ is not
one of the $2^{k-1}$ ones not allowed by 
$(\T \X \ThirdVar_1 \X ... \X \ThirdVar_k \Equal \Parityp)$

Similarly,
take the clause $l_1 \land ... \land l_m \Implies ((\Var \Equal \F) \lor (\AnotherVar \Equal \neg\parity{}))$ in $\CS$ and
resolve it with each clause
$l'_1 \land ... \land l'_n \Implies C$ with
$C = ((\Var \Equal \T) \lor ...) \in \cnf{\XCp}$
in $\CSp$ (there are at most $2^{k-1}$ of them).
Each resulting clause
 $l_1 \land ... \land l_m \land l'_1 \land ... \land l'_n \Implies ((\AnotherVar \Equal \neg\parity{}) \lor ...)$
forces that
either
(i) one of the literals $l_1,...,l_m,l'_1,...,l'_n$ is false,
or
(ii) $\AnotherVar \Equal \parity{}$ implies
that the parity of $\ThirdVar_1,...,\ThirdVar_k$ is not
one of the $2^{k-1}$ ones not allowed by 
$(\F \X \ThirdVar_1 \X ... \X \ThirdVar_k \Equal \Parityp)$.

As $\XC \LinComb \XCp = (\AnotherVar \X \ThirdVar_1 \X ... \X \ThirdVar_k \Equal \parity{} \X \Parityp)$,
the $2^k$ clauses above are the ones in
$\Setdef{(\Trigger \land \Triggerp)\Implies C}{C \in \cnf{\XC \LinComb \XCp}}$.
\qed
\end{proof}
\begin{lemma}\label{Lemma:EqExplanation2}
Let $\Trigger$ and $\Triggerp$ be conjunctions of literals
and take some xor-constraints
$\XC = (\Var \X \AnotherVar \Equal \parity{})$ and
$\XCp = (\Var \X \AnotherVar \X \ThirdVar_1 \X ... \X \ThirdVar_k \Equal \Parityp)$.
Given the sets
$\CS = \Setdef{\Trigger \Implies C}{C \in \cnf{\XC}}$
and
$\CSp = \Setdef{\Triggerp \Implies C}{C \in \cnf{\XCp}}$
of clauses,
the clause set
$\Setdef{(\Trigger \land \Triggerp)\Implies C}{C \in \cnf{\XC \LinComb \XCp}}$
has $2^{k-1}$ clauses and
we can derive them from the ones in $\CS$ and $\CSp$
with $2^k$ resolution steps.
\end{lemma}
\begin{proof}
Take the clause $l_1 \land ... \land l_m \Implies ((\Var \Equal \T) \lor (\AnotherVar \Equal \parity{}))$ in $\CS$ and
resolve it with each clause
$l'_1 \land ... \land l'_n \Implies C$ with
$C = ((\Var \Equal \F) \lor (\AnotherVar \Equal \parity{}) \lor ...) \in \cnf{\XCp}$
in $\CSp$ (there are at most $2^{k-1}$ of them).
Each resulting clause
$l_1 \land ... \land l_m \land l'_1 \land ... \land l'_n \Implies ((\AnotherVar \Equal \parity{}) \lor ...)$
forces that
either
(i) one of the literals $l_1,...,l_m,l'_1,...,l'_n$ is false,
or
(ii) $\AnotherVar \Equal \neg \parity{}$
implies that the parity of $\ThirdVar_1,...,\ThirdVar_k$ is not
one of the $2^{k-1}$ ones not allowed by 
$(\T \X \parity{} \X \T \X \ThirdVar_1 \X ... \X \ThirdVar_k \Equal \Parityp \X)$,
i.e.,
$(\ThirdVar_1 \X ... \X \ThirdVar_k \Equal \Parityp \X \parity{})$

Similarly,
take the clause $l_1 \land ... \land l_m \Implies ((\Var \Equal \F) \lor (\AnotherVar \Equal \neg\parity{}))$ in $\CS$ and
resolve it with each clause
$l'_1 \land ... \land l'_n \Implies C$ with
$C = ((\Var \Equal \T) \lor (\AnotherVar \Equal \neg\parity{}) \lor ...) \in \cnf{\XCp}$
in $\CSp$ (there are at most $2^{k-1}$ of them).
Each resulting clause
$l_1 \land ... \land l_m \land l'_1 \land ... \land l'_n \Implies ((\AnotherVar \Equal \neg\parity{}) \lor ...)$
forces that
either
(i) one of the literals $l_1,...,l_m,l'_1,...,l'_n$ is false,
or
(ii) $\AnotherVar \Equal \parity{}$ implies
that the parity of $\ThirdVar_1,...,\ThirdVar_k$ is not
one of the $2^{k-1}$ ones not allowed by 
$(\F \X \parity{} \X \ThirdVar_1 \X ... \X \ThirdVar_k \Equal \Parityp)$,
i.e.,
$(\ThirdVar_1 \X ... \X \ThirdVar_k \Equal \Parityp \X \parity{})$,

Finally, resolve each obtained clause
$l_1 \land ... \land l_m \land l'_1 \land ... \land l'_n \Implies ((\AnotherVar \Equal \parity{}) \lor \tilde{C})$,
$\tilde{C}$ being disjunction of literals,
with the corresponding clause
$l_1 \land ... \land l_m \land l'_1 \land ... \land l'_n \Implies ((\AnotherVar \Equal \neg\parity{}) \lor \tilde{C})$.
The resulting $2^{k-1}$ clauses together
force that
either
(i) one of the literals $l_1,...,l_m,l'_1,...,l'_n$ is false,
or
(ii) $(\ThirdVar_1 \X ... \X \ThirdVar_k \Equal \Parityp \X \parity{})$ holds.

As $\XC \LinComb \XCp = (\ThirdVar_1 \X ... \X \ThirdVar_k \Equal \parity{} \X \Parityp)$,
the $2^{k-1}$ clauses above are the ones in
$\Setdef{(\Trigger \land \Triggerp)\Implies C}{C \in \cnf{\XC \LinComb \XCp}}$.
\qed
\end{proof}

\newcommand{\DerCls}[1]{\operatorname{der}(#1)}

\begin{retheorem}{\ref{Thm:ResExp}}
  Assume a \SUBST-derivation $\igraph = \Tuple{\vertices,\edges,\Labfunc}$ on
  $\xorpart \land \AL_1 \land \dots \land \AL_k$
  and
  a cnf-compatible cut $\Cut = (\CutA, \CutB)$.
  %for a non-input vertex $\Vertex$ with $\Lab{\Vertex} \neq \T$.
  %
  There is a resolution derivation $\ResDer$ on
  $\bigwedge_{\XC \in \xorpart}\cnf{\XC}$
  such that
  (i)
  for each vertex $\Vertex \in \CutB$ with $\Lab{\Vertex} \neq \T$,
  $\ResDer$ includes all the clauses in $\Setdef{\CExpl{v,\Cut} \Implies C}{C \in \cnf{\Lab{v}}}$,
  and
  (ii)
  $\ResDer$ has at most $\Card{\vertices}2^{m-1}$ clauses,
  where $m$ is the number of variables in the largest xor-constraint in $\xorpart$.
\end{retheorem}
\begin{proof}
  Iteratively on the structure of $\igraph$,
  we show how to derive the clauses
  $\DerCls{\Vertex} = \Setdef{\CExpl{v,\Cut} \Implies C}{C \in \cnf{\Lab{v}}}$
  for each $\Vertex \in \CutB$ with $\Lab{\Vertex} \neq \T$.
  First, we case split by the rule type and have the following two cases.
  
  \paragraph{Case I: $\Lab{\Vertex}$ is obtained from $\Lab{\Vertexp}$ and $\Lab{\Vertexpp}$ by using $\unitruleP$ or $\unitruleN$.}
  Thus $\Lab{\Vertexp} = (\Var \Equal \parity{})$ for some variable $\Var$ and parity $\parity{}$.
  We have the following cases depending on the role of $\Vertexp$.
  \begin{enumerate}
  %
  % Case I
  %
  \item
    $\Vertexp$ is an input vertex with $\Lab{\Vertexp} \in \xorclauses$.
    
    Now $\tmpf{\Vertexp}{\Cut} = \T$
    and
    the clauses in $\CSp = \Setdef{\tmpf{\Vertexpp}{\Cut} \Implies C}{C \in \cnf{\Lab{\Vertexp}}} = \Set{(\Var \Equal \parity{})}$
    occur in $\bigwedge_{\XC \in \xorpart}\cnf{\XC}$.
        
    We then case split by the role of $\Vertexpp$.
    \begin{enumerate}
    \item
      $\Vertexpp$ is an input vertex with $\Lab{\Vertexp} \in \xorclauses$.
      
      Now $\tmpf{\Vertexpp}{\Cut} = \T$
      and
      $\CSpp=\Setdef{\tmpf{\Vertexpp}{\Cut} \Implies C}{C \in \cnf{\Lab{\Vertexpp}}}$, equalling to $\Setdef{C}{C \in \cnf{\Lab{\Vertexpp}}}$,
      consists of clauses already in $\bigwedge_{\XC \in \xorpart}\cnf{\XC}$.
      By Lemma~\ref{Lemma:UnitExplanation},
      the clauses $\DerCls{\Vertex} = \Setdef{\tmpf{\Vertexp}{\Cut} \land \tmpf{\Vertexpp}{\Cut} \Implies C}{C \in \cnf{\Lab{\Vertex}}}$
      can thus be derived from the ones in
      $\CSp$ and $\CSpp$.
      
    \item
      $\Vertexpp$ is an input vertex with  $\Lab{\Vertexpp} \in \Set{\AL_1,...,\AL_k}$.
      
      If $\Lab{\Vertexpp} = (\Var \Equal \parity{})$,
      then $\Lab{\Vertex} = \T$ and there is nothing to prove.
      
      If $\Lab{\Vertexpp} = (\Var \Equal \neg\parity{})$,
      then
      $\tmpf{\Vertexpp}{\Cut} = (\Var \Equal \neg\parity{})$,
      $\Lab{\Vertex} = \F$,
      and
      $\DerCls{\Vertex}=
       \Set{\T \land (\Var \Equal \neg\parity{}) \Implies \F} =
       \Set{(\Var \Equal \parity{})}$
      occurring in $\CSp$. %$\bigwedge_{\XC \in \xorpart}\cnf{\XC}$.

    \item
      $\Vertexpp$ is a non-input vertex in $\CutA$.

      As the cut is cnf-compatible,
      $\Lab{\Vertexpp}$ is either $(\Var \Equal \parity{})$ or $(\Var \Equal \neg\parity{})$.

      If $\Lab{\Vertexpp} = (\Var \Equal \parity{})$,
      then $\Lab{\Vertex} = \T$ and there is nothing to prove.

      If $\Lab{\Vertexpp} = (\Var \Equal \neg\parity{})$,
      then
      $\tmpf{\Vertexpp}{\Cut} = (\Var \Equal \neg\parity{})$,
      $\Lab{\Vertex} = \F$,
      and
      $\DerCls{\Vertex}=
       \Set{\T \land (\Var \Equal \neg\parity{}) \Implies \F} =
       \Set{(\Var \Equal \parity{})}$
      occurring in $\CSp$. %$\bigwedge_{\XC \in \xorpart}\cnf{\XC}$.

    \item
      $\Vertexpp$ is a non-input vertex in $\CutB$.

      Now $\tmpf{\Vertexpp}{\Cut}$
      is a conjunction of literals as the cut is cnf-compatible.
      We already have derived the clauses in
      $\CSpp = \Setdef{\tmpf{\Vertexpp}{\Cut} \Implies C}{C \in \cnf{\Lab{\Vertexpp}}}$.
      By Lemma~\ref{Lemma:UnitExplanation},
      the clauses $\DerCls{\Vertex} = \Setdef{\tmpf{\Vertexp}{\Cut} \land \tmpf{\Vertexpp}{\Cut} \Implies C}{C \in \cnf{\Lab{\Vertex}}}$
      can thus be derived from the ones in
      $\CSp$ and $\CSpp$.
    \end{enumerate}

  %
  % Case II
  %
  \item
    $\Vertexp$ is an input vertex with $\Lab{\Vertexp} \in \Set{\AL_1,...,\AL_k}$.

    Now $\tmpf{\Vertexp}{\Cut} = \Lab{\Vertexp} = (\Var \Equal \parity{})$.

    We next case split by the role of $\Vertexpp$.
    \begin{enumerate}
    \item
      $\Vertexpp$ is an input vertex with $\Lab{\Vertexp} \in \xorclauses$.
      
      Now $\tmpf{\Vertexpp}{\Cut} = \T$
      and
      $\DerCls{\Vertex} =
       \Setdef{\tmpf{\Vertexp}{\Cut}\land\tmpf{\Vertexpp}{\Cut}\Implies C}{C \in \cnf{\Lab{\Vertex}}} =
      %equals to
       \Setdef{(\Var \Equal \parity{}) \Implies C}{C \in \cnf{\simplification{\Lab{\Vertexpp}}{\Var}{\parity{}}}}$.
% =
%       \Setdef{(\Var \Equal \neg\parity{}) \lor C}{C \in \cnf{\simplification{\Lab{\Vertexpp}}{\Var}{\parity{}}}}$.
      The clauses in $\DerCls{\Vertex}$ are thus a subset of those occurring
      in $\cnf{\Lab{\Vertexpp}}$ and
      thus also in $\bigwedge_{\XC \in \xorpart}\cnf{\XC}$.

    \item
      $\Vertexpp$ is an input vertex with  $\Lab{\Vertexpp} \in \Set{\AL_1,...,\AL_k}$.

      Now $\Lab{\Vertexpp}$ is either $(\Var \Equal \parity{})$ or
      $(\Var \Equal \neg\parity{})$.

      If $\Lab{\Vertexpp} = (\Var \Equal \parity{})$,
      then $\Lab{\Vertex} = \T$ and there is nothing to prove.

      If $\Lab{\Vertexpp} = (\Var \Equal \neg\parity{})$,
      then
      $\tmpf{\Vertexpp}{\Cut} = (\Var \Equal \neg\parity{})$,
      $\Lab{\Vertex} = \F$,
      and
      $\DerCls{\Vertex}=
       \Set{(\Var \Equal \parity{}) \land (\Var \Equal \neg\parity{}) \Implies \F} =
       \emptyset$.
      This is fine because the clausal explanation is
      also the tautology stating that $\Var$ cannot be true and false at the same time.

    \item
      $\Vertexpp$ is a non-input vertex in $\CutA$.

      Now $\Lab{\Vertexpp}$ is either $(\Var \Equal \parity{})$ or
      $(\Var \Equal \neg\parity{})$ because the cut is cnf-compatible.

      The rest is thus similar to the previous sub-sub-case.

    \item
      $\Vertexpp$ is a non-input vertex in $\CutB$.

      Now $\tmpf{\Vertexpp}{\Cut}$
      is a conjunction of literals as the cut is cnf-compatible
      and
      we have already derived the clauses in
      $\CSpp = \Setdef{\tmpf{\Vertexpp}{\Cut} \Implies C}{C \in \cnf{\Lab{\Vertexpp}}}$.
      The clause set
      $\DerCls{\Vertex} =
       \Setdef{\tmpf{\Vertexp}{\Cut} \land \tmpf{\Vertexpp}{\Cut} \Implies C}{C \in \cnf{\Lab{\Vertex}}}$
      equals to
      %$\Setdef{(\Var\Equal\parity{}) \land \tmpf{\Vertexpp}{\Cut} \Implies C}{C \in \cnf{\simplification{\Lab{\Vertexpp}}{\Var}{\parity{}}}}$,
      %i.e.,
      $\Setdef{\tmpf{\Vertexpp}{\Cut} \Implies (\Var\Equal\neg\parity{}) \lor C}{C \in \cnf{\simplification{\Lab{\Vertexpp}}{\Var}{\parity{}}}}$.
      The clauses in $\DerCls{\Vertex}$ are thus a subset of those occurring
      in $\CSpp$.
    \end{enumerate}

  %
  % Case III
  %
  \item
    $\Vertexp$ is a non-input vertex in $\CutA$.

    Now $\tmpf{\Vertexp}{\Cut} = \Lab{\Vertexp} = (\Var \Equal \parity{})$
    as the cut is cnf-compatible.

    The rest of this sub-case is similar to the previous sub-case.

  %
  % Case IV
  %
  \item
    $\Vertexp$ is a non-input vertex in $\CutB$.

    Now $\tmpf{\Vertexp}{\Cut}$
    is a conjunction of literals as the cut is cnf-compatible
    and
    we have already derived
    $\CSp = \Setdef{\tmpf{\Vertexp}{\Cut} \Implies C}{C \in \cnf{\Lab{\Vertexp}}} = \Set{(\tmpf{\Vertexp}{\Cut} \Implies (\Var\Equal\parity{}))}$.

    %The rest of the sub-case is basically similar to the sub-case
    %``$\Vertexp$ is an input vertex with $\Lab{\Vertexp} \in \xorclauses$''
    %proven above.

    We then case split by the role of $\Vertexpp$.
    \begin{enumerate}
    \item
      $\Vertexpp$ is an input vertex with $\Lab{\Vertexp} \in \xorclauses$.
      
      Now $\tmpf{\Vertexpp}{\Cut} = \T$
      and
      $\CSpp=\Setdef{\tmpf{\Vertexpp}{\Cut} \Implies C}{C \in \cnf{\Lab{\Vertexpp}}}$, eqaulling to $\Setdef{C}{C \in \cnf{\Lab{\Vertexpp}}}$,
      consists of clauses already in $\bigwedge_{\XC \in \xorpart}\cnf{\XC}$.
      By Lemma~\ref{Lemma:UnitExplanation},
      the clauses $\DerCls{\Vertex} = \Setdef{\tmpf{\Vertexp}{\Cut} \land \tmpf{\Vertexpp}{\Cut} \Implies C}{C \in \cnf{\Lab{\Vertex}}}$
      can thus be derived from the ones in
      $\CSp$ and $\CSpp$.
      
    \item
      $\Vertexpp$ is an input vertex with  $\Lab{\Vertexpp} \in \Set{\AL_1,...,\AL_k}$.
      
      If $\Lab{\Vertexpp} = (\Var \Equal \parity{})$,
      then $\Lab{\Vertex} = \T$ and there is nothing to prove.
      
      If $\Lab{\Vertexpp} = (\Var \Equal \neg\parity{})$,
      then
      $\tmpf{\Vertexpp}{\Cut} = (\Var \Equal \neg\parity{})$,
      $\Lab{\Vertex} = \F$,
      and
      $\DerCls{\Vertex}=
       \Set{(\tmpf{\Vertexp}{\Cut} \land (\Var \Equal \neg\parity{}) \Implies \F)} =
       \Set{(\tmpf{\Vertexp}{\Cut} \Implies (\Var \Equal \parity{}))}$
      occurring in $\CSp$. %$\bigwedge_{\XC \in \xorpart}\cnf{\XC}$.

    \item
      $\Vertexpp$ is a non-input vertex in $\CutA$.

      As the cut is cnf-compatible,
      $\Lab{\Vertexpp}$ is either $(\Var \Equal \parity{})$ or $(\Var \Equal \neg\parity{})$.

      The rest is thus similar to the previous sub-sub-case.

    \item
      $\Vertexpp$ is a non-input vertex in $\CutB$.

      Now $\tmpf{\Vertexpp}{\Cut}$
      is a conjunction of literals as the cut is cnf-compatible.
      We already have derived the clauses in
      $\CSpp = \Setdef{\tmpf{\Vertexpp}{\Cut} \Implies C}{C \in \cnf{\Lab{\Vertexpp}}}$.
      By Lemma~\ref{Lemma:UnitExplanation},
      the clauses $\DerCls{\Vertex} = \Setdef{\tmpf{\Vertexp}{\Cut} \land \tmpf{\Vertexpp}{\Cut} \Implies C}{C \in \cnf{\Lab{\Vertex}}}$
      can thus be derived from the ones in
      $\CSp$ and $\CSpp$.
    \end{enumerate}

  \end{enumerate}

\paragraph{Case II: $\Lab{\Vertex}$ is obtained from $\Lab{\Vertexp}$ and $\Lab{\Vertexpp}$ by using $\eqvruleP$ or $\eqvruleN$.}
Thus $\Lab{\Vertexp} = (\Var \X \AnotherVar \Equal \parity{})$
for some variables $\Var,\AnotherVar$ and parity $\parity{}$.
  We have the following cases depending on the role of $\Vertexp$.
  \begin{enumerate}
  %
  % Case I
  %
  \item
    $\Vertexp$ is an input vertex with $\Lab{\Vertexp} \in \xorclauses$.
    
    Now $\tmpf{\Vertexp}{\Cut} = \T$
    and
    the clauses in $\CSp = \Setdef{\tmpf{\Vertexpp}{\Cut} \Implies C}{C \in \cnf{\Lab{\Vertexp}}}$
    occur in $\bigwedge_{\XC \in \xorpart}\cnf{\XC}$.
    
    We then case split by the role of $\Vertexpp$.
    \begin{enumerate}
    \item
      $\Vertexpp$ is an input vertex with $\Lab{\Vertexp} \in \xorclauses$.
      
      Now $\tmpf{\Vertexpp}{\Cut} = \T$
      and
      $\CSpp=\Setdef{\tmpf{\Vertexpp}{\Cut} \Implies C}{C \in \cnf{\Lab{\Vertexpp}}}$
      consists of clauses already occurring in $\bigwedge_{\XC \in \xorpart}\cnf{\XC}$.
      By Lemmas \ref{Lemma:EqExplanation1} and \ref{Lemma:EqExplanation2},
      we can thus derive the clauses in
      $\DerCls{\Vertex} =
       \Setdef{\tmpf{\Vertexp}{\Cut} \land \tmpf{\Vertexpp}{\Cut} \Implies C}{C \in \cnf{\Lab{\Vertex}}}$.

    \item
      $\Vertexpp$ is an input vertex with  $\Lab{\Vertexpp} \in \Set{\AL_1,...,\AL_k}$.

      Now $\tmpf{\Vertexpp}{\Cut} = \Lab{\Vertexpp}$.
      As $\Var$ must occur in $\Lab{\Vertexpp}$,
      $\Lab{\Vertexpp} = (\Var \Equal \Paritypp)$ for a $\Paritypp \in \Set{\F,\T}$.

      Thus
      $\DerCls{\Vertex} =
       \Setdef{\tmpf{\Vertexp}{\Cut} \land (\Var \Equal \Paritypp) \Implies C}{C \in \cnf{\Lab{\Vertexp}\LinComb(\Var\Equal\Paritypp)}}$.
      This equals to 
      $\Setdef{\tmpf{\Vertexp}{\Cut} \Implies (\Var \Equal \neg\Paritypp) \lor C}{C \in \cnf{\Lab{\Vertexp}\LinComb(\Var\Equal\Paritypp)}}$
      and
      thus all the clauses in it are already in
      $\CSp = \Setdef{\tmpf{\Vertexp}{\Cut} \Implies C}{C \in \cnf{\Lab{\Vertexp}}}$.

    \item
      $\Vertexpp$ is a non-input vertex in $\CutA$.

      Now $\tmpf{\Vertexpp}{\Cut} = \Lab{\Vertexpp}$.
      As $\Var$ must occur in $\Lab{\Vertexpp}$ and the cut is cnf-compatible,
      $\Lab{\Vertexpp} = (\Var \Equal \Paritypp)$ for a $\Paritypp \in \Set{\F,\T}$.

      The rest of the case is similar to the previous one.

    \item
      $\Vertexpp$ is a non-input vertex in $\CutB$.

      Now $\tmpf{\Vertexpp}{\Cut}$
      is a conjunction of literals as the cut is cnf-compatible
      and
      we have already derived the clauses in
      $\CSpp = \Setdef{\tmpf{\Vertexpp}{\Cut} \Implies C}{C \in \cnf{\Lab{\Vertexp}}}$.
      The clauses in
      $\DerCls{\Vertex} =
       \Setdef{\tmpf{\Vertexp}{\Cut} \land \tmpf{\Vertexpp}{\Cut} \Implies C}{C \in \cnf{\Lab{\Vertex}}}$
      can thus be derived by Lemmas \ref{Lemma:EqExplanation1} and \ref{Lemma:EqExplanation2}.

    \end{enumerate}

  %
  % Case II
  %
  \item
    $\Vertexp$ is an input vertex with $\Lab{\Vertexp} \in \Set{\AL_1,...,\AL_k}$.

    This case is not possible because $\Lab{\Vertex}$ should be
    of form $(\Var \X \AnotherVar \Equal \parity{})$.

  %
  % Case III
  %
  \item
    $\Vertexp$ is a non-input vertex in $\CutA$.

    As the cut is cnf-compatible,
    $\Lab{\Vertexp}$ must be of form $(\Var \Equal \parity{})$,
    not of $(\Var \X \AnotherVar \Equal \parity{})$ as required.
    Therefore, this case is impossible.

  %
  % Case IV
  %
  \item
    $\Vertexp$ is a non-input vertex in $\CutB$.
    
    Now $\tmpf{\Vertexp}{\Cut}$
    is a conjunction of literals as the cut is cnf-compatible
    and
    we have already derived the clauses in
    $\CSp = \Setdef{\tmpf{\Vertexp}{\Cut} \Implies C}{C \in \cnf{\Lab{\Vertexp}}}$.

    The rest of the sub-case is similar to the sub-case
    ``$\Vertexp$ is an input vertex with $\Lab{\Vertexp} \in \xorclauses$''
    proven above.

  \end{enumerate}%
  \qed
\end{proof}

%--------------------------------------------------------------------------
%
%--------------------------------------------------------------------------
%
%--------------------------------------------------------------------------
\subsection{Proof of Theorem~\ref{Thm:PexpSim}}

The constructs in the proof are illustrated in Figures \ref{Fig:pexp-ex-unsat} and \ref{Fig:pexp-ex-sat}.

\begin{figure}[ht]
\begin{center}
\begin{tabular}{c@{\qquad}c}
\includegraphics[scale=0.7]{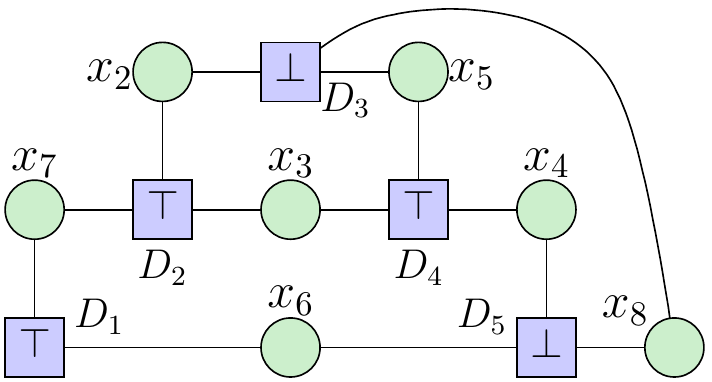}
&
\includegraphics[scale=0.7]{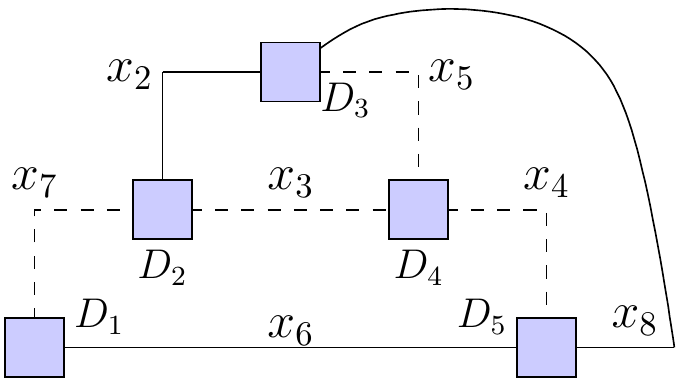}
\\
(a) constraint graph
&
(b) dual graph and spanning tree (dashed edges)
\\
\\
\multicolumn{2}{c}{\includegraphics[width=.98\textwidth]{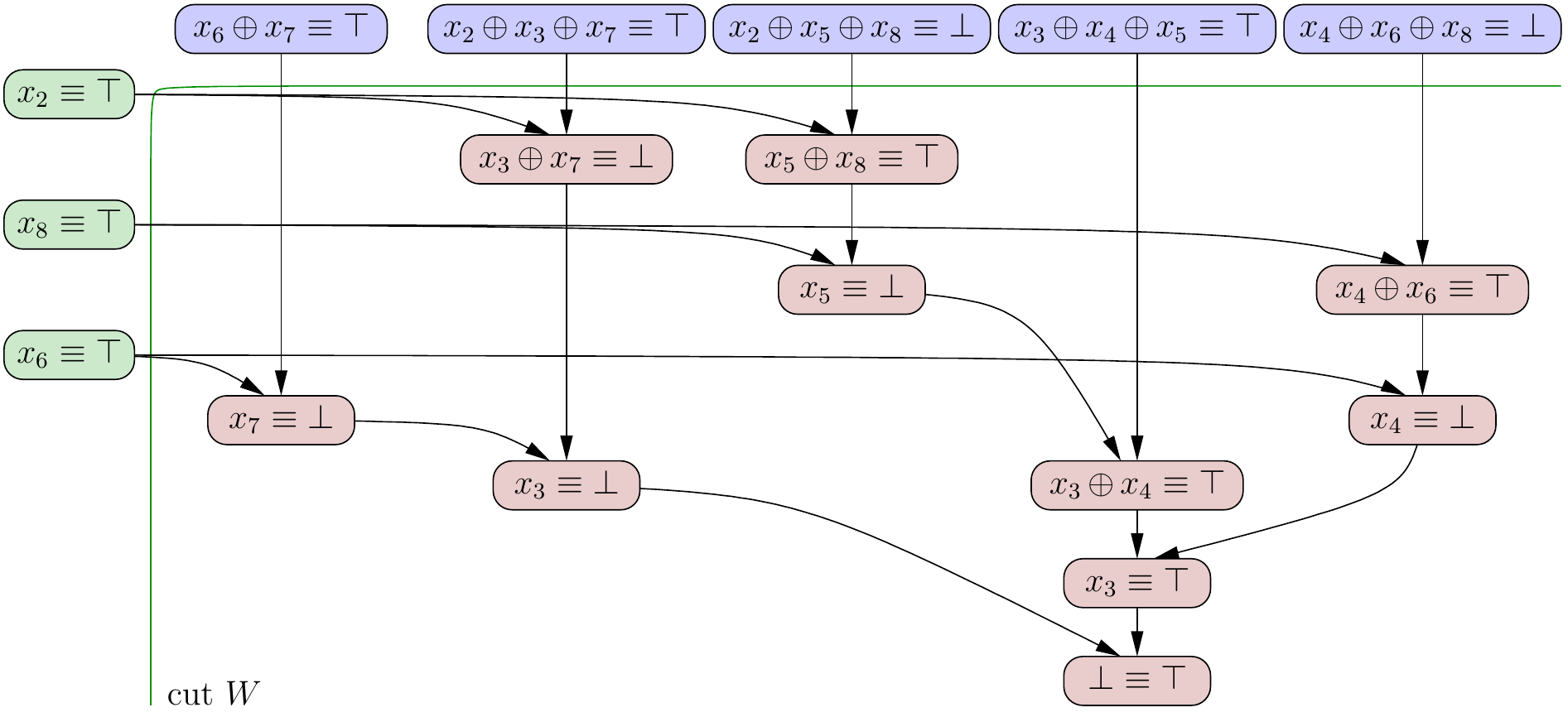}}
\\
\multicolumn{2}{c}{(c) a derivation and a cut with parity explanation $\F \Equal \F$ for the vertex $\F \Equal \T$}
\end{tabular}
\end{center}
\caption{The constraint (a) and dual (b) graphs of an unsatisfiable xor-constraint conjunction and an \UP-derivation giving the parity explanation $\F \Equal \F$ and clausal explanation $\T \Implies \F$.}
\label{Fig:pexp-ex-unsat}
\end{figure}

\begin{figure}[ht]
\begin{center}
\includegraphics[width=.98\textwidth]{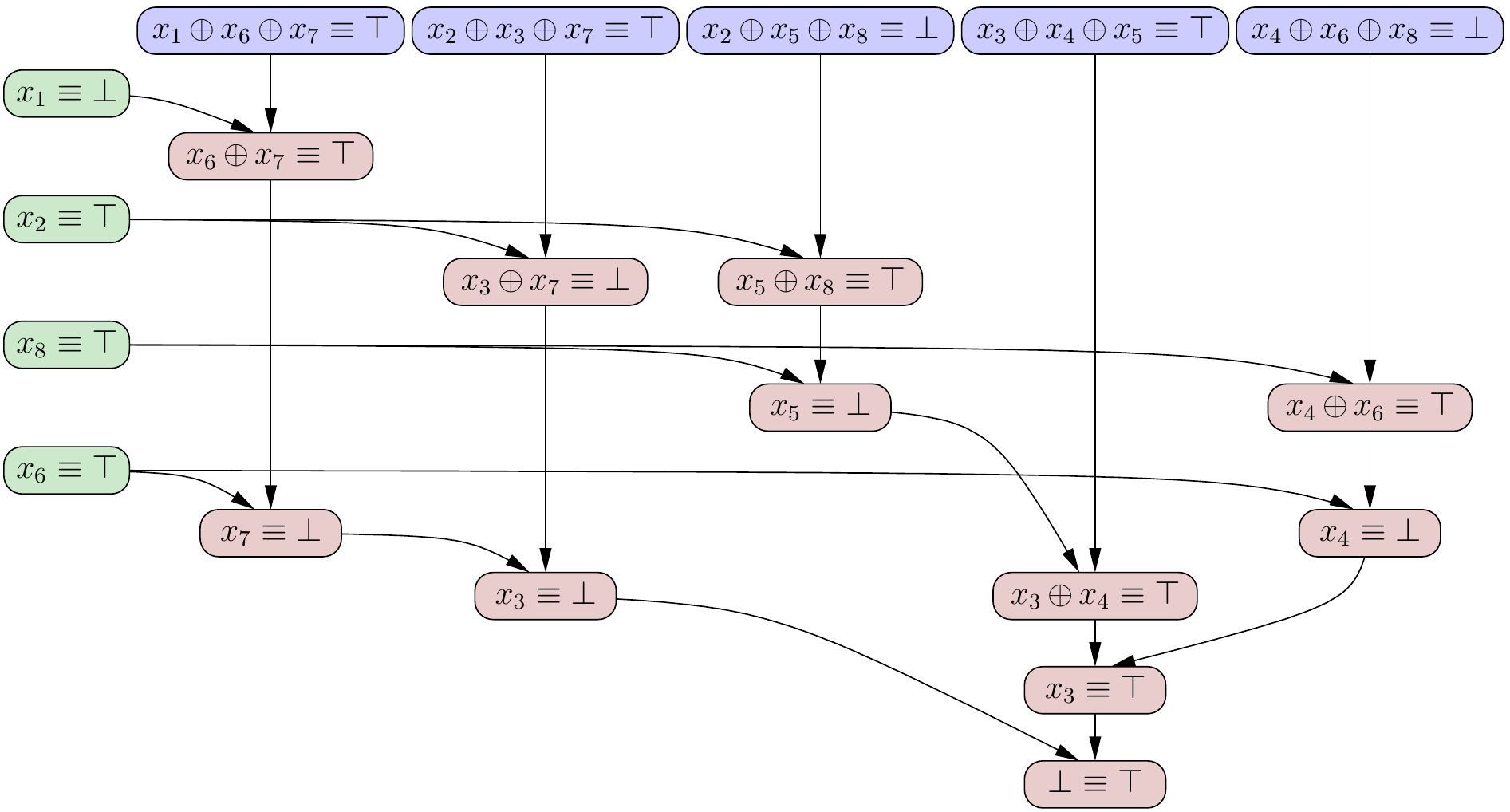}
\end{center}
\caption{An \UP-derivation for the instance in Ex.~\ref{Ex:kGE} and Fig.~\ref{Fig:KgeExample}, giving the parity explanation $x_1 \Equal \F$ and the clausal explanation $(x_1 \Equal \F) \Implies \F$ i.e.\ $(x_1 \Equal \T)$.}
\label{Fig:pexp-ex-sat}
\end{figure}

\begin{retheorem}{\ref{Thm:PexpSim}}
  Let $\xorclauses$ be a conjunction of xor-constraints
  such that each variable occurs in at most three xor-constraints.
  
  If $\xorclauses$ is unsatisfiable,
  then there is a
  \UP{}-derivation on $\xorclauses \land y_1 \land ... \land y_m$
  with some $y_1,...,y_m \in \VarsOf{\xorclauses}$,
  a vertex $\Vertex$ with $\Lab{\Vertex} = (\F \Equal \T)$ in it,
  and
  a cut $\Cut$ for $\Vertex$ such that
  $\PExpl{\Vertex, \Cut} = (\F \Equal \F)$
  and thus
  $\PExpl{\Vertex, \Cut} \LinComb \Lab{\Vertex} = (\F \Equal \T)$.

  If $\xorclauses$ is satisfiable and
  $\xorclauses \Models (\Var_1 \X ... \X \Var_k \Equal \parity{})$,
  then there is a
  \UP{}-derivation on $\xorclauses \land (\Var_1\Equal\parity{1}) \land ... \land (\Var_k\Equal\parity{k}) \land y_1 \land ... \land y_m$
  with some $y_1,...,y_m \in \VarsOf{\xorclauses}\setminus\Set{\Var_1,...,\Var_k}$,
  a vertex $\Vertex$ with $\Lab{\Vertex} = (\F \Equal \T)$ in it,
  and
  a cut $\Cut$ for $\Vertex$ such that
  $\PExpl{\Vertex, \Cut} \LinComb \Lab{\Vertex} = (\Var'_1 \X ... \X \Var'_l \Equal \Parityp)$ for some $\Set{\Var'_1,...,\Var'_l} \subseteq \Set{\Var_1,...,\Var_k}$ and $\Parityp \in \Set{\F,\T}$
  such that
  $\xorclauses \Models (\Var'_1 \X ... \X \Var'_l \Equal \Parityp)$.
\end{retheorem}
\begin{proof}
  Let $\xorclauses = \XC_1 \land ... \land \XC_n$ be a conjunction
  of xor-constraints such that each variable occurs in at most three xor-constraints.

  We construct the required \UP-derivations by starting from the one consisting of $n$ input vertices (one for each xor-constraint in $\xorclauses$) and then transforming each ``current vertex for the xor-constraint $\XC_i$'' into a new one by applying unit propagation to it.

  \noindent\textbf{Case I: $\xorclauses$ is unsatisfiable.}
  First,
  as long as the current xor-constraint vertices contain unary xor-constraints whose variable is occurring in other current xor-constraint vertices,
  apply the unit propagation rule to eliminate the other occurrences.
  If the false vertex $\F \Equal \T$ is derived,
  then the parity explanation for it will be $\F \Equal \F$
  under the furthest cut
  (i.e., the cut $\Tuple{\CutA,\CutB}$ with the smallest ``reason side'' $\CutA$), as required.

  Otherwise,
  the set $S'$ of current xor-constraint vertices with binary or longer xor-constraint labels
  induces an unsatisfiable conjunction of xor-constraints.
  By Lemma~\ref{Lemma:LinearCombs}
  there is a subset $S''$ of $S'$ such that
  $\BigLinComb_{\Vertex \in S''}\Lab{\Vertex} = (\F\Equal\T)$.
  We can, and will, assume that $S''$ is minimal,
  i.e.\ that there is no subset of $S''$ whose labels' linear combination is $(\F\Equal\T)$.
  Each variable occurring in the labels of $S''$ occurs there exactly two times:
  it occurs an even number of times because the linear combination of the labels is empty and it cannot occur more than three times due to the assumption we have made in the theorem.

  We next consider the ``dual graph'' for $S''$,
  meaning the edge-labeled multi-graph
  $\Tuple{\Setdef{\Lab{\Vertex}}{\Vertex \in S''},\Setdef{\Tuple{\Set{\XC,\XCp},\Var}}{\Var\in{\VarsOf{\XC}\cap\VarsOf{\XCp}}}}$
  and
  take any spanning tree of it.
  As each variable occurs at exactly two times in the labels of $S''$,
  it occurs in exactly one edge in the dual graph.

  To complete the \UP-derivation, we proceed in two phases.
  In phase one, we make an xor-assumption $(\Var\Equal\T)$ for each variable
  occurring in an edge of the dual graph \emph{not} belonging to the spanning tree.
  We apply unit propagation so that the variable is removed from the two xor-constraint labels it occurs in the current version of xor-constraints of $S''$.
  Thus the out-degree of the xor-assumption vertex is thus two.
  In phase two, we unit propagate the remaining variables in $S''$, starting from the leafs of the spanning tree, and obtain conflict on some variable occurring in an edge of the spanning tree.
  Take the furthest cut of the constructed \UP-derivation for the conflict vertex.
  As all the xor-constraints in $S''$ were required to obtain the conflict, all the occurrences of the variables in the edges not in the spanning tree (i.e.\ xor-assumptions made) were required, too.
  The out-degrees of the other vertices in the last two phases are one.
  Thus each xor-assumption occurs twice in when computing the parity explanation and these occurrences cancel each other out,
  resulting in the empty parity explanation as required.

  \noindent\textbf{Case II: $\xorclauses$ is satisfiable and $\xorclauses \Models (\Var_1 \X ... \X \Var_k \Equal \parity{})$.}
  First, choose some values $\parity{1},....,\parity{k}$ for the variables
  $\Var_1,...,\Var_k$ so that
  $\parity{1}\XX...\XX\parity{k} \neq \parity{}$.
  Now clearly $\xorclauses \land (\Var_1\Equal\parity{1}) \land ... \land (\Var_k \Equal\parity{k})$ is unsatisfiable.

  Next,
  make the xor-assumption $(\Var_i\Equal\parity{i})$ for each $\Var_i$
  and
  apply unit propagation as long as possible.
  If the falsity vertex $\F \Equal \T$ is derived,
  then the parity explanation of $\F \Equal \T$ under the furthest cut
  will be $(\Var'_1 \XX ... \XX \Var'_l \Equal \Paritypp)$
  for some $\Set{\Var'_1,...,\Var'_l} \subseteq \Set{\Var_1,...,\Var_k}$ and $\Paritypp \in \Set{\F,\T}$.
  As $\xorclauses \Models (\Var'_1 \X ... \X \Var'_l \Equal \Paritypp)+(\F \Equal\T)$,
  we have the desired result.

  Otherwise, 
  the set $S'$ of current xor-constraint vertices with binary or longer xor-constraint labels induces an unsatisfiable conjunction of xor-constraints.
  We can thus proceed as in Case I after the initial unit propagation.
  The conflict vertex obtained eventually may 
  depend on the xor-assumptions $(\Var_i\Equal\parity{i})$ we made above
  and these may also occcur in the parity explanation under the furthest cut.
  Thus the parity explanation 
  will be $(\Var'_1 \XX ... \XX \Var'_l \Equal \Paritypp)$
  for some $\Set{\Var'_1,...,\Var'_l} \subseteq \Set{\Var_1,...,\Var_k}$ and $\Paritypp \in \Set{\F,\T}$.
  As $\xorclauses \Models (\Var'_1 \X ... \X \Var'_l \Equal \Paritypp)+(\F \Equal\T)$,
  we have the desired result.

\qed
\end{proof}

%--------------------------------------------------------------------------
%
%--------------------------------------------------------------------------
%
%--------------------------------------------------------------------------

%% file: proofs2.tex
\subsection{Proof of Theorem~\ref{Thm:GESimulation}}

If an xor-constraint conjunction $\psi$ has a {\it \UP{}-propagation table} for
the set of variables $Y \subseteq \VarsOf{\psi}$, we denote this by
$\HasPropTable{Y}{\psi}$. 
%
%We denote the set of all possible ``alias'' variables for a set of variables $Y$ by $\AliasesOf{Y} = \Setdef{a}{a\mbox{ is the ``alias'' variable for }Y' \subseteq Y}$.

\begin{lemma}
\label{Lem:PTableProp}
Let $ \phi $ be a satisfiable conjunction of xor-constraints such that $
\HasPropTable{Y}{\phi} $ for some $ Y \subseteq
\VarsOf{\phi} $, and $ a, a_1, \dots, a_n
\in \VarsOf{\phi} $ ``alias'' variables for the subsets $Y', Y_1, \dots,
    Y_n \subseteq Y$, respectively, and $ Y' = Y_1 \oplus \dots \oplus Y_n $.
It holds that $ \phi \wedge (a_1 \equiv \parity{1}) \wedge \dots \wedge
(a_n \equiv \parity{n}) \UPderiv (a \equiv \parity{1} \oplus \dots \oplus
    \parity{n})$.
\end{lemma}

\begin{proof}
We prove the lemma by induction on the sequence $ a_1, \dots, a_n $.
The induction hypothesis is that Lemma~\ref{Lem:PTableProp} holds for the
case $a_1, \dots, a_{n-1} $.

Base case: $n = 1$.
The claim holds trivially, because $a = a_1$.

Induction step for $ n > 1 $. 
By the property PT1, the ``alias'' variable $a'$ for the set of variables $(Y_1 \oplus \dots \oplus Y_{n-1})$ is present in $\VarsOf{\phi}$ and the xor-constraint $ (a' \oplus Y_1 \oplus
        \dots \oplus Y_{n-1} \equiv \bot) $ is in $\phi$.
By the induction hypothesis, it holds that $ \phi \wedge (a_1 \equiv
        \parity{1}) \wedge \dots \wedge (a_{n-1} \equiv \parity{n-1}) \UPderiv
(a' \equiv \parity{1} \oplus \dots \oplus \parity{n-1}) $.
By the property PT2, it holds that the xor-constraint $ (a \oplus a_{n-1}
\oplus a' \equiv \bot) $ is in $ \phi$.
It follows that $ \phi \wedge (a_1 \equiv \parity{1}) \wedge \dots \wedge (a_n \equiv \parity{n}) \UPderiv (a \equiv \parity{1} \oplus \dots \oplus \parity{n}) $.

\end{proof}

\begin{lemma}
\label{Lem:PTableProp2}
Let $ \phi $ be a conjunction of xor-constraints such that $
\HasPropTable{Y}{\phi} $ for some $ Y \subseteq \VarsOf{\phi} $
of variables in $\phi$, and $ \phi' $ be a satisfiable conjunction of
xor-constraints in $\phi$ such that $ \VarsOf{\phi'} \subseteq Y$.
If $ \phi' \Models (Y' \equiv \parity{}) $ for some $Y' \subseteq Y$, then
it holds for the ``alias'' variable $ a \in \VarsOf{\phi} $ for the subset $Y'$ that $ \phi \UPderiv (a \equiv \parity{}) $.
\end{lemma}

\begin{proof}
By Lemma~\ref{Lemma:LinearCombs}, there is a subset $S = (Y_1 \equiv \parity{1}) \wedge \dots \wedge (Y_n \equiv \parity{n}) $ of xor-constraints in $\phi'$ such that $ \SetLC{S} = (Y' \equiv \parity{}) $.
By the property PT1, it holds that the the ``alias'' variable $a$ for the set of variables $Y'$ is present in $ \VarsOf{\phi} $ and the xor-constraint $ (a \oplus Y' \equiv \bot) $ is in $\phi$.
Also by the property PT1, it holds for each xor-constraint $ (Y_i \equiv
\parity{i}) $ in $S$ that the ``alias'' variable $ a_i $ for the set of
variables $Y_i$ is present in $\VarsOf{\phi}$ and the xor-constraint
$ (a_i \oplus Y_i \equiv
\bot) $ and by the property PT3 the xor-constraint $ (a_i \equiv
    \parity{i}) $ is in $ \phi$.
It holds by Lemma~\ref{Lem:PTableProp}, that $ \phi \wedge (a_1 \wedge
\parity{1}) \wedge \dots \wedge (a_n \equiv \parity{n}) \UPderiv
(a \equiv \parity{}) $.
\end{proof}

\begin{lemma}
\label{Lem:PropTableImpl}
If $ \xorclauses$ is an xor-constraint conjunction and $ Y \subseteq
\VarsOf{\xorclauses}$, then $ \HasPropTable{Y}{\xorclauses \wedge \PropTable{Y,
    \xorclauses, |Y|}}$ .
\end{lemma}

\begin{proof}
Consider the pseudo code for the algorithm $\PropTableName$ in Fig.~\ref{Fig:PropTable}.
The variable $ Y' $ takes the value of each subset of $ Y $ in the loop in 
lines 1-5, and as the result $ \xorclauses'$ has a variable $a$ for each non-empty
subset $Y'$ of $ Y $ such that $ (a \oplus Y' \equiv \parity{}) $ is in $ \xorclauses'$. The property PT1 is satisfied by the lines 2-3 and the property PT3 by the lines 4-5.

In the loop in lines 6-11 is iterated for every pair of subsets $ Y_1, Y_2
\subset Y $ such that $ Y_1 \not = Y_2$. It holds for the smallest-indexed
variables $a_1,a_2,a_3 \in \VarsOf{\xorclauses}$ such that 
 the xor-constraints $ (a_1 \oplus Y_1 \equiv  \bot) $,
$ (a_2 \oplus Y_2 \equiv \bot)$, and $ (a_3 \oplus (Y_1 \oplus Y_2) \equiv
    \bot)$ that the xor-constraint $ (a_1 \oplus a_2 \oplus a_3 \equiv \bot)$
is in $ \xorclauses' $. This satisfies PT2.
\end{proof}

\begin{lemma}
\label{Lem:GEConjunctions}
Given an xor-constraint conjunction $\phi_0$ 
and an elimination order
$\Tuple{x_1,\dots,x_n}$ for the variables of $\phi_0$ for the algorithm
$\getransname$ where $k=|\VarsOf{\phi_0}|$, it holds that there is a sequence of xor-constraint conjunctions $
\Tuple{\phi_1,\dots,\phi_n}$ in $ \psi = \phi_0 \wedge \getrans{\phi_0} $ and a sequence of sets of variables
$\Tuple{Y_1,\dots,Y_n}$ such that
it holds for each triple $ \Tuple{x_i,Y_i,\phi_i}$:
\begin{itemize}
\item $ Y_i = \VarsOf{\ClausesOf{x_i}{\phi_{i-1}}} \cap \Set{x_i,\dots,x_n} $
\item $ \HasPropTable{Y_i}{\psi}$,
\item $\phi_i = \phi_{i-1} \wedge \PropTable{Y_i,\phi_{i-1}, k}$
\item $\phi_n = \phi_0 \wedge \getrans{\phi_0}$
\end{itemize}
\end{lemma}
\begin{proof}
Assume an xor-constraint conjunction $ \phi_0 $ and an elimination
order $ \Tuple{x_1,\dots,x_n}$ for the variables $ \phi_0$
for the algorithm $ \getransname$ where $k=|\VarsOf{\phi_0}|$.
The translation $ \getrans{\phi_0} $ in Figure~\ref{Fig:GETrans} is initialized with $ \xorclauses' = \phi_0
$ and $ V = \VarsOf{\phi_0} $. The loop in lines 1-5 is run $n$ times and $V$
takes the values $ V_1,\dots,V_n$. In the first iteration of the loop, all
variables of $ \phi_0 $ are in the set $ V_1 = V$. Then for each successive
iteration $ i $ it holds that $ V_i = V_{i-1} \setminus \Set{x_{i-1}} $ because
$ x_i $ is removed from the set $ V$ in the line 4. 
We now argue that the xor-constraints in the conjunction  $ \phi_1 \wedge \dots \wedge \phi_n$ are in $ \psi = \phi_0 \wedge \getrans{\phi_0} $. 
After choosing to ``eliminate'' the variable $ x_i$ in the line 3, the
xor-constraint conjunction $ \xorclauses' $ is augmented with $
\PropTable{\VarsOf{\ClausesOf{x_i}{\xorclauses'}} \cap V_i, \xorclauses', k}$,
    so $\phi_i = \phi_{i-1} \wedge \PropTable{Y_i,\phi_{i-1}, k}$.
It is clear that $ V_i = \Set{x_i, \dots, x_n}$, so $ \phi_i $ is identical to
the xor-constraint conjunction $ \xorclauses'$ after the $i$th iteration of the
loop.
Upon $i$th iteration of the loop in the lines 1-5,       
the translation $ \PropTableName $ in Figure~\ref{Fig:PropTable} is initialized
with $Y = Y_i $ and $ \xorclauses' = \phi_{i-1}$. 
After all the $n$ iterations are done it is clear that $ \phi_n = \phi_0
\wedge \getrans{\phi_0}$.
By Lemma~\ref{Lem:PropTableImpl} 
it holds that $ \HasPropTable{Y_i}{\xorclauses'}$
and also $ \HasPropTable{Y_i}{\psi}$.
because adding xor-constraints cannot break any conditions
of the UP-propragation table.
\end{proof}

\begin{lemma}
\label{Lem:GEPropagation}
Given a satisfiable xor-constraint conjunction $ \phi_0' $ in an
xor-constraint conjunction $ \phi_0$ and an elimination
order
$\Tuple{x_1,\dots,x_n}$ for the variables of $\phi_0$ for the algorithm
$\getransname$ where $k=|\VarsOf{\phi_0}|$, it holds that there is a sequence of xor-constraint conjunctions $
\Tuple{\phi_1',\dots,\phi_n'}$ in $ \psi = \phi_0 \wedge \getrans{\phi_0} $
such that for each $ \phi_i'$ in $\Tuple{\phi_0',\dots,\phi_n'}$ it holds that
\begin{itemize}
\item given literals $ \AL_1, \dots, \AL_k, \IL$ such that $ (\SetLC{\phi_i'}) \wedge \AL_1 \wedge \dots \wedge \AL_k \Models \IL $, it holds that
    $ \psi \wedge \AL_1 \wedge \dots \wedge \AL_k \UPderiv \IL $.
\end{itemize}
\end{lemma}

\begin{proof}
Assume a satisfiable xor-constraint clause conjunction $ \phi_0' $ in an
xor-constraint conjunction $ \phi_0 $ and an elimination
order $ \Tuple{x_1, \dots, x_n}$ for the variables of $ \phi_0$
for the algorithm $ \getransname $.

By Lemma~\ref{Lem:GEConjunctions}, it holds that
there is a sequence of xor-constraint conjunctions $
\Tuple{\phi_1,\dots,\phi_n}$ in $ \psi = \phi_0 \wedge \getrans{\phi_0} $ and a sequence of sets of variables
$\Tuple{Y_1,\dots,Y_n}$ such that
it holds for each triple $ \Tuple{x_i,Y_i,\phi_i}$:
\begin{itemize}
\item $ Y_i = \VarsOf{\ClausesOf{x_i}{\phi_{i-1}}} \cap \Set{x_i,\dots,x_n} $
\item $ \HasPropTable{Y_i}{\psi}$,
\item $\phi_i = \phi_{i-1} \wedge \PropTable{Y_i,\phi_{i-1},k}$
\end{itemize}

Let $ \LCVarsGone{\phi_{i-1}'} = \VarsOf{\phi_{i-1}'} \setminus
\VarsOf{\SetLC{\phi_{i-1}'}}$ be the set of variables the ``disappear'' in the
normal form of the linear combination of the xor-constraints in $ \phi_{i-1}'$.

We define a corresponding sequence of $n$ tuples $\Tuple{Y_i', X_i, V_i, a_i,
    \phi_i'} $ as follows:
\begin{itemize}
\item $Y_i' = \VarsOf{\ClausesOf{x_i}{\phi_{i-1}'}} \cap \Set{x_i,\dots,x_n} $
, and
\item $X_i = \VarsOf{\SetLC{\phi_{i-1}'}} \cap \VarsOf{\ClausesOf{x_{i}}{\phi_i'}}$
be the set of variables have occurrences in the xor-constraints of the variable $x_i$
and also remain in the normal form of the linear combination of $ \phi_{i-1}'$, and
\item $V_i = \VarsOf{\SetLC{\ClausesOf{x_i}{ \phi_{i-1}'}}} \cap \LCVarsGone{\phi_{i-1}'} $
be the set of variables remain in the normal form of the linear combination of the xor-constraints of the variable $x_i$ that also disappear in the normal form of the linear combination of $ \phi_{i-1}'$, and
\item $a_i$ is a variable such that the xor-constraint $(a_i \oplus V_i \equiv \parity{i})$ is in $ \phi_i \wedge (\bot \equiv \bot) $ (it exists because $V_i \subseteq Y_i$ and $ \HasPropTable{Y_i}{\phi_i}$), and
\item if $ x_i \not \in \VarsOf{\phi_i'} $ or $ V_i = \emptyset $, then $ \phi_i' = \phi_{i-1}' $, otherwise
\begin{itemize}
\item if $ (a_i \oplus V_i \equiv \parity{i} \oplus \parity{i}') $ is in $ \phi_{i-1}'$,
    then $ \phi_i' = \phi_{i-1}' \setminus \ClausesOf{x_i}{\phi_{i-1}'}$, otherwise
\item 
$ \phi_i' = \phi_{i-1}' \setminus \ClausesOf{x_i}{\phi_{i-1}'} \wedge (a_i \oplus V_i \equiv \parity{i} \oplus \parity{i}')$.
\end{itemize}
\end{itemize}

We prove the lemma by induction on the structure of the xor-constraint conjunction
sequence $ \Tuple{\phi_0', \dots, \phi_n'} $.

The induction hypothesis is that the lemma holds for the xor-constraint conjunction
sequence $ \Tuple{\phi_{i}',\dots,\phi_n'} $.

Base case: $i = n$. 
Assume any literals $ \AL_1,\dots,\AL_k $ such that
$\SetLC{\phi_i'} \wedge \AL_1 \wedge \dots \wedge \AL_k \Models \IL $. 
It holds that $ \VarsOf{\phi_i'} = \emptyset $, so $ \VarsOf{\IL} \in \VarsOf{\AL_1, \dots, \AL_k}$. 
It clearly holds that $ \phi_0 \wedge \AL_1 \wedge \dots \wedge
\AL_k \UPderiv \IL $.

Induction step: $ 0 \leq i-1 < n $.
Assume any literals $ \AL_1,\dots,\AL_k $ such that $
\VarsOf{\AL_1,\dots,\AL_k,\IL} \subseteq \VarsOf{\phi_0} $ and $
\SetLC{\phi_{i-1}'} \wedge \AL_1 \wedge \dots \wedge \AL_k \Models \IL $. 
If $ \phi_{i-1}' = \phi_{i}'$, then it holds by the induction hypothesis that $ \phi_0 \wedge \AL_1 \wedge \dots \wedge \AL_k \UPderiv \IL $.

We have two cases to consider:
\begin{itemize}
\item Case I: $\VarsOf{\IL} \in X_i$. 
It holds that $ \VarsOf{\SetLC{\phi_{i}'}} 
  \subseteq \VarsOf{\AL_1,\dots,\AL_k} \cup \Set{a_{i}}$
  and $ \SetLC{\phi_i'} \wedge \AL_1 \dots \wedge \AL_k \Models (V_i \equiv \parity{i}) $,
  so by induction hypothesis it holds that $ \phi_0 \wedge \AL_1 \wedge \dots \wedge \AL_k \UPderiv (a_i \equiv \parity{i}') $. 
It holds that $ \SetLC{\ClausesOf{x_i}{\phi_{i-1}'}} \wedge \AL_1 \wedge \dots \wedge \AL_k \wedge (a_i \equiv \parity{i}') \Models \IL $. 
By Lemma~\ref{Lem:PTableProp2}, it holds that $ \phi_i \wedge \AL_1 \wedge \dots
\AL_k \wedge (a_i \equiv \parity{i}') \UPderiv \IL $.

\item Case II: $\VarsOf{\IL} \not \in X_i$. 
It holds that $ \SetLC{\ClausesOf{x_i}{\phi_{i-1}'}} \wedge \AL_1 \wedge \dots
\wedge \AL_k \Models (V_i \equiv \parity{i}) $.
By Lemma~\ref{Lem:PTableProp2}, it holds that $ \phi_i \wedge \AL_1 \wedge \dots \wedge \AL_k \UPderiv (a_i \equiv \parity{i}') $. 
It holds that $ \VarsOf{\SetLC{\phi_{i}'}} \subseteq \VarsOf{\AL_1, \dots, \AL_k, \IL} \cup \Set{a_i}$, so $\SetLC{\phi_{i}'} \wedge \AL_1 \wedge \dots \wedge \AL_k \wedge (a_i \equiv \parity{i}') \Models \IL$.
It holds by induction hypothesis that $ \phi_0 \wedge \AL_1 \wedge \dots \wedge \AL_k \wedge (a_i \equiv \parity{i}') \UPderiv \IL$.
\end{itemize}
\end{proof}

The following lemma states that $\getransname$ translation refutes any
unsatisfiable xor-constraint conjunctions.
\begin{lemma}
\label{Lem:GEUnsat}
If $\xorclauses$ is an unsatisfiable xor-constraint conjunction, then
$\xorclauses \wedge \getrans{\xorclauses} \UPderiv (\bot \equiv \top) $ where
$k=|\VarsOf{\xorclauses}|$.
\end{lemma}

\begin{proof}
Assume an unsatisfiable xor-constraint conjunction $ \xorclauses$.
By Lemma~\ref{Lemma:LinearCombs}, there is a subset $S = C_1 \wedge \dots \wedge C_m $ of xor-constraints in $\xorclauses$ 
such that $\SetLC{S} = (\bot \equiv \top) $.
Let $ (X_m \equiv \parity{m}) = C_m$.  It holds that $ \SetLC{C_1 \wedge \dots
\wedge C_{m-1}} = (X_m \equiv \parity{m} \oplus \top) $. Thus, $ \psi = C_1
\wedge \dots \wedge C_{m-1} $ is satisfiable. 
Let $\psi = \xorclauses \wedge \getrans{\xorclauses}$.
By
Lemma~\ref{Lem:GEConjunctions}, it holds that:
\begin{itemize}

\item there is a set of variables $ Y \subseteq
\VarsOf{\psi} $ such that $ \VarsOf{C_m}
\subseteq Y $ and $ \HasPropTable{Y}{\psi} $, and

\item there is a variable $ y \in \VarsOf{\psi}$
such that the xor-constraint $ (y \oplus X_m \equiv
    \parity{m} \oplus \parity{m}') $ is in $\psi$, and

\item the xor-constraint $ (y \equiv \parity{m}') $ is in $ \psi$.
\end{itemize}

Because $ \psi \Models (y \equiv \parity{m}' \oplus \top) $, it holds 
by Lemma~\ref{Lem:GEPropagation} that $\psi \UPderiv (y \equiv \parity{m}' \oplus \top) $.
Since $ \psi \UPderiv (y \equiv \parity{m}) $ and $ \psi \wedge
\getrans{\xorclauses} \UPderiv (y \equiv \parity{m} \oplus \top) $, it follows
that $ \psi \UPderiv (\bot \equiv \top) $.
\end{proof}

\begin{lemma} 
\label{Lem:PropTableModels}
The satisfying truth assignments of $\xorclauses$ are exactly the
ones of $ \xorclauses \wedge \PropTable{Y, \xorclauses, k}$ when projected to $
\VarsOf{\xorclauses}$ where $ Y \subseteq \VarsOf{\xorclauses}$.
\end{lemma}

\begin{proof}
It holds by definition that $ \xorclauses \wedge \PropTable{Y, \xorclauses,k} \Models
\xorclauses$, so it suffices to show that if $\TA$ is a satisfying truth
assignment for $\xorclauses$, it can be extended to a satisfying truth
assignment $\TA'$ for $\PropTable{Y, \xorclauses, k}$. 
Assume that $\TA$ is a truth assignment such that $\TA \Models \xorclauses$.
Let $\TA'$ be a truth assignment identical to $\TA$ except for the following additions. The translation $\PropTable{Y,\phi,k}$ in Figure~\ref{Fig:PropTable} adds four kinds of xor-constraints.
\begin{enumerate}
\item $ (y \oplus Y' \equiv \bot) $ where $ Y' $ is a non-empty subset of $Y$
and $y$ is a new variable. If $ \tau \Models (Y' \equiv \top) $, add $y$ to $\tau'$, otherwise add $\neg y$ to $\tau'$. It is clear that $ \tau' \Models (y \oplus Y' \equiv \bot) $.

\item $ (a_1 \oplus a_2 \oplus a_3 \equiv \parity{1} \oplus \parity{2} \oplus
\parity{3}) $ if the xor-constraints $ (a_1 \oplus Y_1 \equiv \parity{1})
$, $ (a_2 \oplus Y_2 \equiv \parity{2})$, and $ (a_3 \oplus (Y_1 \oplus Y_2)
\equiv \parity{3}) $ are in $\phi$ augmented with xor-constraints from the
previous step.
From the previous step it is clear that $ \tau' \Models (a_1 \oplus Y_1 \equiv
\parity{1}) $, $\tau' \Models (a_2 \oplus Y_2 \equiv \parity{2}) $, and
$\tau' \Models (a_3 \oplus Y_3 \equiv \parity{3}) $. It follows that $\tau'
\Models (a_1 \oplus a_2 \oplus a_3 \equiv \parity{1} \oplus \parity{2} \oplus
\parity{3})$.

\item $ (y \equiv \parity{}') $ where $ y \in \VarsOf{\phi} $ such that
the xor-constraints $ (y \oplus Y' \equiv \parity{} \oplus \parity{}') $ and $ (Y' \equiv \parity{}) $ are in $\xorclauses$ augmented with xor-constraints from the previous step. Since $ (Y' \equiv \parity{}) $ is an original xor-constraint in $\phi$,
    it holds that $\tau' \Models (Y' \equiv \parity{}) $. It follows that $\tau' \Models (y \equiv \parity{}') $.

\item $ (y \oplus y' \equiv \parity{} \oplus \parity{}') $ where $ y, y' \in \VarsOf{\phi} $ and $ \parity{}, \parity{}' \in \Set{\top,\bot}$ such that
the xor-constraints $ (y \oplus Y' \equiv \parity{}) $ and $ (y' \oplus Y' \equiv \parity{}') $ where $Y'$ is a non-empty subset of $Y$
are in $\xorclauses$ augmented with xor-constraints from the previous step. If $ \tau' \Models (Y' \equiv \parity{} \oplus \top) $, then add $ y $ to $ \tau'$, otherwise add $ \neg y $ to $\tau'$. If $ \tau' \Models (Y' \equiv \parity'{} \oplus \top)$, then add $ y' $ to $ \tau'$, otherwise add $ \neg y' $ to $\tau'$.
It follows that $ \tau' \Models (y \oplus y' \equiv \parity{} \oplus \parity{}')$.
\end{enumerate}
\end{proof}

\begin{retheorem}{\ref{Thm:GESimulation}} 
If $\xorclauses$ is an xor-constraint
conjunction, then $\getrans{\xorclauses}$ is a GE-simulation formula for
$\xorclauses$ provided that $k=|\VarsOf{\xorclauses}|$.
\end{retheorem}

\begin{proof}
We first prove that the satisfying truth assignments of $\xorclauses$ are
exactly the ones of $\psi = \xorclauses \wedge \getrans{\xorclauses}$ when
projected to $\VarsOf{\xorclauses}$. 
The translation \getransname{} in Figure~\ref{Fig:GETrans} only adds xor-constraint
conjunctions of the type $\PropTable{Y, \phi, k} $ for some set of variables $
Y \subseteq \VarsOf{\xorclauses}$ and some xor-constraint conjunction $\phi$
and by Lemma~\ref{Lem:PropTableModels} the satisfying assignments
of $ \phi$ are exactly the ones of $ \phi \wedge \PropTable{Y, \phi, k}$ 
when projected to $ \VarsOf{\phi} $.
It follows by induction that the satisfying truth assignment for $\xorclauses$
are exactly to the ones of $ \xorclauses \wedge 
\getrans{\xorclauses}$ when projected to $ \VarsOf{\xorclauses}$.

Next we show that if $\xorclauses$ is satisfiable 
and 
 $\xorclauses \wedge \AL_1 \wedge \dots \wedge \AL_k \Models \IL$, then $\IL$ is $\UP{}$-derivable from $ \xorclauses \wedge \getrans{\xorclauses} \wedge \AL_1 \wedge \dots \wedge \AL_k$. 
By Lemma~\ref{Lemma:LinearCombs}, there is a subset $S$ of xor-constraints in $\phi \wedge \AL_1 \wedge \dots \wedge \AL_k $ such that $ \SetLC{S} = \IL $,
By Lemma~\ref{Lem:GEPropagation}, it holds that 
$\SetLC{S} \UPderiv \IL $, so $ \psi \wedge \AL_1 \wedge \dots \wedge \AL_k \UPderiv \IL $.

It remains to show that if $\xorclauses$ is unsatisfiable,
   then $\psi \UPderiv (\bot \equiv \top)$. Assume that 
   $\xorclauses$ is unsatisfiable. By Lemma~\ref{Lem:GEUnsat}, it holds
   that $\psi \UPderiv (\bot \equiv \top)$.
All the requirements for GE-simulation formula are satisfied, so $\getrans{\xorclauses}$ is a GE-simulation formula for $\xorclauses$.
\end{proof}

\subsection{Proof of Theorem~\ref{Thm:XorSimp}}
\begin{retheorem}{\ref{Thm:XorSimp}}
If $ \Tuple{\phi_a', \phi_b'} $ is the result of applying
one of the simplification rules to $ \Tuple{\phi_a, \phi_b} $ 
and 
$\phi_a \wedge \phi_b \wedge \AL_1 \wedge \dots \wedge
\AL_k \UPderiv \IL $, then $\phi_a' \wedge \phi_b' \wedge
\AL_1 \wedge \dots \wedge \AL_k \UPderiv \IL $.  
\end{retheorem}

\begin{proof}
Let $ \Tuple{\phi_a', \phi_b'} $ be the result of applying one of the
simplification rules to $ \Tuple{\phi_a, \phi_b}$, $ \phi_a \wedge \phi_b
\wedge \AL_1 \wedge \dots \wedge \AL_k \UPderiv \IL $, and $ \IL = (x \equiv \parity{})$.
If S1 was the simplification rule used, then it clearly
holds that $ \phi_a' \wedge \phi_b' \wedge \AL_1 \wedge \dots \wedge \AL_k \UPderiv \IL $.

Otherwise, S2 was used to simplify an xor-constraint $ \XC$ in $\phi_a $ with
an xor-constraint $ \XC' $ in $ \phi_b $ such that $ |\VarsOf{\XC} \cap \VarsOf{\XC'}|
\geq |\VarsOf{\XC'}| - 1$.
It holds that $ \phi_a' = \phi_a \setminus \Set{\XC} \cup \Set{\XC \LC \XC'} $
and $ \phi_b' = \phi_b $.
It must hold that there is an xor-clause $ C = (x \oplus y_1 \oplus \dots
\oplus y_n \equiv \parity{} \oplus \parity{1} \oplus \dots \oplus
\parity{n}) $ in $ \psi $ such that for each $y_i \in \Set{y_1, \dots,
y_n} $ it holds that $ \psi \UPderiv (y_i \equiv \parity{i}) $.
We prove by induction $ \IL $ is $\UP$-derivable from $ \psi' = \phi_a' \wedge
\phi_b' \wedge \AL_1 \wedge \dots \wedge \AL_k $.
The induction hypothesis is that for each $y_i \in \Set{y_1, \dots, y_n} $ it
holds that $ \psi \UPderiv (y_i \equiv \parity{i}') $.

Base case: $C = \IL = (x \equiv \parity{})$. If $ C \not = \XC $, then $ (x \equiv \parity{}) $ is in $ \psi' $ and $ \psi' \UPderiv \IL $. Otherwise, $ C = \XC $.
Since $ |\VarsOf{\XC'} \cap \VarsOf{\XC}| \geq |\VarsOf{\XC'}| - 1$, it holds that
$\VarsOf{\XC \LC \XC'} = \Set{x'} $ for some $x' \in \VarsOf{\psi} $
$ \VarsOf{\XC'} = \Set{x, x'} $.
The xor-constraint $ \XC \LC \XC'$ is in $ \phi_a' $, and the
xor-constraint $\XC'$ $\phi_b' $.
It clearly holds that $ \psi' \UPderiv \IL $.

Induction step: $C \not = \IL$. If $ C \not = \XC $, then $ C $ is in $ \psi'$
and $ \psi' \UPderiv \IL$. Otherwise $ C = \XC $.
We have two cases to consider:
\begin{itemize}
\item Case 1: $ x \in \VarsOf{\XC'} $.
By induction hypothesis it holds for each $ y_i \in \Set{y_1, \dots, y_n} $
that $ \psi' \UPderiv (y_i \equiv \parity{i}) $.
If there is a variable $z \in \VarsOf{\XC'}$ such that $ z \not \in
\VarsOf{\XC} $, then $ (\XC \LC \XC') \wedge (y_1 \equiv \parity{1}) \wedge \dots \wedge (y_n \equiv \parity{n}) \UPderiv (z \equiv \parity{}') $ and then
$ \XC \wedge (y_1 \equiv \parity{1}) \wedge \dots \wedge (y_n \equiv \parity{n}) \wedge (z \equiv \parity{}') \UPderiv (x \equiv \parity{}) $, so $ \psi' \UPderiv \IL $.
Otherwise, it holds $ \VarsOf{\XC'} \subseteq \VarsOf{\XC} $, and it clearly holds that $ \psi' \UPderiv \IL $.

\item Case 2: $ x \not \in \VarsOf{\XC'} $.
By induction hypothesis it holds for each $ y_i \in \Set{y_1, \dots, y_n} $
that $ \psi' \UPderiv (y_i \equiv \parity{i}) $.
If there is a variable $z \in \VarsOf{\XC'}$ such that $ z \not \in
\VarsOf{\XC} $, then $ \XC' \wedge (y_1 \equiv \parity{1}) \wedge \dots \wedge (y_n \equiv \parity{n}) \UPderiv (z \equiv \parity{}') $ and then
$ (\XC \LC \XC') \wedge (y_1 \equiv \parity{1}) \wedge \dots \wedge (y_n \equiv \parity{n}) \wedge (z \equiv \parity{}') \UPderiv (x \equiv \parity{}) $, so $ \psi' \UPderiv \IL $.
Otherwise, it holds $ \VarsOf{\XC'} \subseteq \VarsOf{\XC} $, and it clearly holds that $ \psi' \UPderiv \IL $.
\end{itemize}
\end{proof}

\subsection{Proof of Theorem~\ref{Thm:XCut}}
\newcommand{\CutAp}{V'_\textup{a}}
\newcommand{\CutBp}{V'_\textup{b}}
\newcommand{\xorpartA}{\xorpart^\textup{a}}
\newcommand{\xorpartB}{\xorpart^\textup{b}}
\newcommand{\Iface}{X'}
\newcommand{\Ifaceparity}{\parity{}'}

\begin{retheorem}{\ref{Thm:XCut}}
Let $ (\CutA, \CutB) $ be an $\Vars$-cut partition of $\xorclauses$.
Let $ \xorclausesA = \bigwedge_{D \in \CutA} D$, $\xorclausesB = \bigwedge_{D \in \CutB} D$, and $\AL_1,\dots,\AL_k \in \LitsOf{\xorclauses}$.
 Then
it holds that:
\begin{itemize}
\item If $\xorclauses \wedge \AL_1 \wedge \dots \wedge \AL_k $ is unsatisfiable, then
\begin{enumerate}
\item $ \xorclausesA \wedge \AL_1 \wedge \dots \wedge \AL_k $ 
or $ \xorclausesB \wedge \AL_1 \wedge \dots \wedge \AL_k $ is unsatisfiable; or 
\item $\xorclausesA \wedge \AL_1 \wedge \dots \wedge \AL_k \Models 
(X' \equiv \parity{}')$ and $ \xorclausesB \wedge \AL_1 \wedge \dots
\AL_k \Models (X' \equiv \parity{}' \oplus \top) $ for some $ X' \subseteq X $
and $ \parity{}' \in \set{\top, \bot} $.
\end{enumerate}
\item 
If $\xorclauses \wedge \AL_1 \wedge \dots \wedge \AL_k $ is satisfiable
and $ \xorclauses \wedge \AL_1 \wedge \dots \wedge \AL_k \Models (Y \equiv \parity{}) $ for some $Y \subseteq \VarsOf{\xorclauses^\alpha}$, $Y \cap (\VarsOf{\xorclauses^\beta} \setminus \VarsOf{\xorclauses^\alpha}) = \emptyset$, and $\parity{} \in \Set{\top, \bot}$ where $\alpha \in \Set{\textup{a},\textup{b}}$ and $ \beta \in \Set{\textup{a}, \textup{b}} \setminus \Set{\alpha} $, then
\begin{enumerate}
\item $\xorclausesA \wedge \AL_1 \wedge \dots \wedge \AL_k \Models (Y \equiv \parity{}) $
or $ \xorclausesB \wedge \AL_1 \wedge \dots \wedge \AL_k \Models (Y \equiv \parity{}) $;
or
\item $\xorclauses^\alpha \wedge \AL_1 \wedge \dots \wedge \AL_k \Models 
( X' \equiv \parity{}') $ and $ \xorclauses^\beta \wedge \AL_1 \wedge \dots \wedge \AL_k \wedge (X' \equiv \parity{}') \Models (Y \equiv \parity{}) $ for some $X' \subseteq X$,
$\parity{}' \in \set{\top,\bot}$, $\alpha \in \Set{\textup{a},\textup{b}}$, and $\beta \in \Set{\textup{a},\textup{b}}\setminus\Set{\alpha}$.
%\item $\xorclauses^a \wedge \AL_1 \wedge \dots \wedge \AL_k \Models 
%( X' \equiv \parity{}) $ and $ \xorclauses^b \wedge \AL_1 \wedge \dots \wedge \AL_k \wedge (X' \equiv \parity{}) \Models \IL $ for some $ X' \subseteq X$ and $\parity{} \in \set{\top,\bot}$.
%\item $\xorclauses^b \wedge \AL_1 \wedge \dots \wedge \AL_k \Models 
%( X' \equiv \parity{}) $ and $ \xorclauses^a \wedge \AL_1 \wedge \dots \wedge \AL_k \wedge (X' \equiv \parity{}) \Models \IL $ for some $ X' \subseteq X$ and $\parity{} \in \set{\top,\bot}$.
\end{enumerate}
\end{itemize}
\end{retheorem}
\begin{proof}
  Let $(\CutAp,\CutBp)$ be an $\Vars$-cut partition of
  $\xorpart \land (\AL_1) \land ... \land (\AL_k)$
  with
  $\VarsOf{\CutAp} = \VarsOf{\CutA}$,
  $\VarsOf{\CutBp} = \VarsOf{\CutB}$,
  $\CutA \subseteq \CutAp$, and
  $\CutB\subseteq \CutBp$.
  Such partition exists because the xor-assumption literals $\AL_i$
  are unit xor-constraints.

  \begin{itemize}
\item  Case I: $\xorpart \land {\AL_1 \land ... \land \AL_k}$ is unsatisfiable.
  By Lemma~\ref{Lemma:LinearCombs},
  there is a subset $S$ of xor-constraints
  in $\xorpart \land (\AL_1) \land ... \land (\AL_k)$
  such that $\BigLinComb_{\XC \in S} \XC = (\F \equiv \T)$.
  Observe that
  $\BigLinComb_{\XC \in S} \XC =
   (\BigLinComb_{\XC \in {\CutAp \cap S}} \XC) \LinComb
   (\BigLinComb_{\XC \in {\CutBp \cap S}} \XC)$.
  If $\BigLinComb_{\XC \in {\CutAp \cap S}} \XC = (\F \equiv \T)$,
  then
  $\xorpartA \land {\AL_1 \land ... \land \AL_k}$ is also unsatisfiable.
  Similarly,
  if $\BigLinComb_{\XC \in {\CutBp \cap S}} \XC = (\F \equiv \T)$,
  then $\xorpartB \land {\AL_1 \land ... \land \AL_k}$ is unsatisfiable.
  Otherwise,
  it must be that
  $\BigLinComb_{\XC \in {\CutAp \cap S}} \XC = (\Iface \equiv \Ifaceparity)$
  and
  $\BigLinComb_{\XC \in {\CutBp \cap S}} \XC = (\Iface \equiv \Ifaceparity \X \T)$
  with $\Ifaceparity \in \Set{\F,\T}$
  because
  $\CutAp \cap \CutBp = \emptyset$,
  ${\VarsOf{\CutAp} \cap \VarsOf{\CutBp}} = \Iface$ and
  $(\BigLinComb_{\XC \in {\CutAp \cap S}} \XC) \LinComb
   (\BigLinComb_{\XC \in {\CutBp \cap S}} \XC) = (\F \equiv \T)$.
  Thus
  $\xorpartA \land {\AL_1 \land ... \land \AL_k} \Models (\Iface \equiv \Ifaceparity)$ and
  $\xorpartB \land {\AL_1 \land ... \land \AL_k} \Models (\Iface \equiv \Ifaceparity \X \T)$.
  
 \item Case II:
  $\xorpart \land \AL_1 \land ... \land \AL_k$ is satisfiable
  and
  $\xorpart \land \AL_1 \land ... \land \AL_k \Models (Y \equiv \parity{})$.
 By Lemma~\ref{Lemma:LinearCombs},
 there is a subset $S$ of xor-constraints
  in $\xorpart \land (\AL_1) \land ... \land (\AL_k)$
  such that ${\BigLinComb_{\XC \in S} \XC} = (Y \equiv \parity{})$.
  Again,
  observe that
  $(\BigLinComb_{\XC \in S} \XC) = 
   (\BigLinComb_{\XC \in {\CutAp \cap S}} \XC) \LinComb
   (\BigLinComb_{\XC \in {\CutBp \cap S}} \XC)$.
Assume that $Y \subseteq \VarsOf{\xorclausesB}$ and $ Y \cap (\VarsOf{\xorclausesA} \setminus \VarsOf{\xorclausesB}) = \emptyset $; the other case is symmetric.
Then we simplify the equation
$(\BigLinComb_{\XC \in S} \XC) = 
(\BigLinComb_{\XC \in {\CutAp \cap S}} \XC) \LinComb
(\BigLinComb_{\XC \in {\CutBp \cap S}} \XC)$ 
by (i) substituting $(\BigLinComb_{\XC \in S} \XC)$ with $ (Y \equiv \parity{})$ and (ii) evaluating $ (\BigLinComb_{\XC \in {\CutAp \cap S}} \XC) $.
This gives two cases:
\begin{enumerate}
\item evaluating $ (\BigLinComb_{\XC \in {\CutAp \cap S}} \XC) $ gives
an empty expression and the simplified equation is then $ (Y \equiv \parity{}) = (\BigLinComb_{\XC \in {\CutBp \cap S}} \XC)$, so it follows that
  $\xorpartB \land \AL_1 \land ... \land \AL_k \Models (Y \equiv \parity{})$.

\item evaluating $ (\BigLinComb_{\XC \in {\CutAp \cap S}} \XC) $ gives
an xor-constraint $ (\Iface \equiv \Ifaceparity{}) $ for some $ \Iface \subseteq X $ and $\Ifaceparity{} \in \Set{\top, \bot}$ because 
  $\CutAp \cap \CutBp = \emptyset$,
  ${\VarsOf{\CutAp} \cap \VarsOf{\CutBp}} = \Iface$ and
  $(\BigLinComb_{\XC \in {\CutAp \cap S}} \XC) \LinComb
   (\BigLinComb_{\XC \in {\CutBp \cap S}} \XC) = (Y \equiv \parity{})$.
The simplified equation is then $ (Y \equiv \parity{}) = (\Iface \equiv \Ifaceparity{}) \LinComb (\BigLinComb_{\XC \in {\CutBp \cap S}} \XC)$,
so it follows that
  $\xorpartB \land {\AL_1 \land ... \land \AL_k \land (\Iface \equiv \Ifaceparity)} \Models (Y \equiv \parity{})$.
\end{enumerate}
\end{itemize}
\end{proof}

\subsection{Proof of Theorem~\ref{Thm:TreeDecomposition}}

\begin{lemma}
\label{Lem:PCProp}
If $\phi$ is a satisfiable conjunction in $\xorclauses \wedge \psi$ such that
$\VarsOf{\phi} \subseteq Y$, $Y \subseteq \VarsOf{\xorclauses}$, $\HasPropTable{Y}{\xorclauses \wedge \psi}$,
and $ \phi \wedge (Y_1 \equiv \parity{1}) \wedge \dots \wedge (Y_n \equiv \parity{n}) \Models (Y' \equiv \parity{}')$
where $Y_1,\dots,Y_n,Y' \subseteq Y $
and $\parity{1},\dots,\parity{n},\parity{}' \in \Set{\top,\bot}$, 
then $\xorclauses \wedge \psi \wedge a_1 \equiv \parity{1} \wedge \dots \wedge a_n \equiv
\parity{n} \UPderiv a' \equiv \parity{}'$
where $ a_1,\dots,a_n,a' $ are the ``alias'' variables for the sets $Y_1,\dots,Y_n,Y'$, respectively.
\end{lemma}

\begin{proof}
By Lemma~\ref{Lemma:LinearCombs}, there is a subset $\phi'$ of xor-constraints in $\phi \wedge (Y_1 \equiv \parity{1}) \wedge \dots \wedge (Y_n \equiv \parity{n}) $ such that $ \SetLC{\phi'} = (Y' \equiv \parity{}') $.
By the property PT1, it holds for each xor-constraint $ (Y'' \equiv \parity{}'') $ in $\phi'$ 
that that the corresponding ``alias'' variable $ a'' $ for the set of
variables $Y''$ is present in $\VarsOf{\phi}$
and by the property PT3 the xor-constraint $ (a'' \equiv
    \parity{}'') $ is in $ \xorclauses \wedge \psi$.
It holds by Lemma~\ref{Lem:PTableProp}, that $ \xorclauses \wedge \psi \wedge (a_1 \wedge \parity{1}) \wedge \dots \wedge (a_n \equiv \parity{n}) \UPderiv
(a' \equiv \parity{}') $.
\end{proof}

\begin{retheorem}{\ref{Thm:TreeDecomposition}}
If $\Set{X_1, \dots, X_n} $ is the family of variable sets in the tree
decomposition of the primal graph of an xor-constraint conjunction $
\xorclauses$ and 
$\phi_0, \dots, \phi_n$ is a sequence of xor-constraint conjunctions
such that $ \phi_0 = \xorclauses $ and $ \phi_i = \phi_{i-1} \wedge \PropTable{X_i, \phi_{i-1}, |X_i|} $ for $i \in \Set{1,\dots,n}$, then $ \phi_n \setminus \xorclauses $ is a GE-simulation formula for
$\xorclauses$ with $O(n {2^{2k}}) + |\xorclauses|$ xor-constraints, where
$k = \max(|X_1|, \dots, |X_n|) $.  
\end{retheorem}

\begin{figure}[ht]
\centering
\begin{tabular}{c@{\qquad}c}
\includegraphics[scale=0.4]{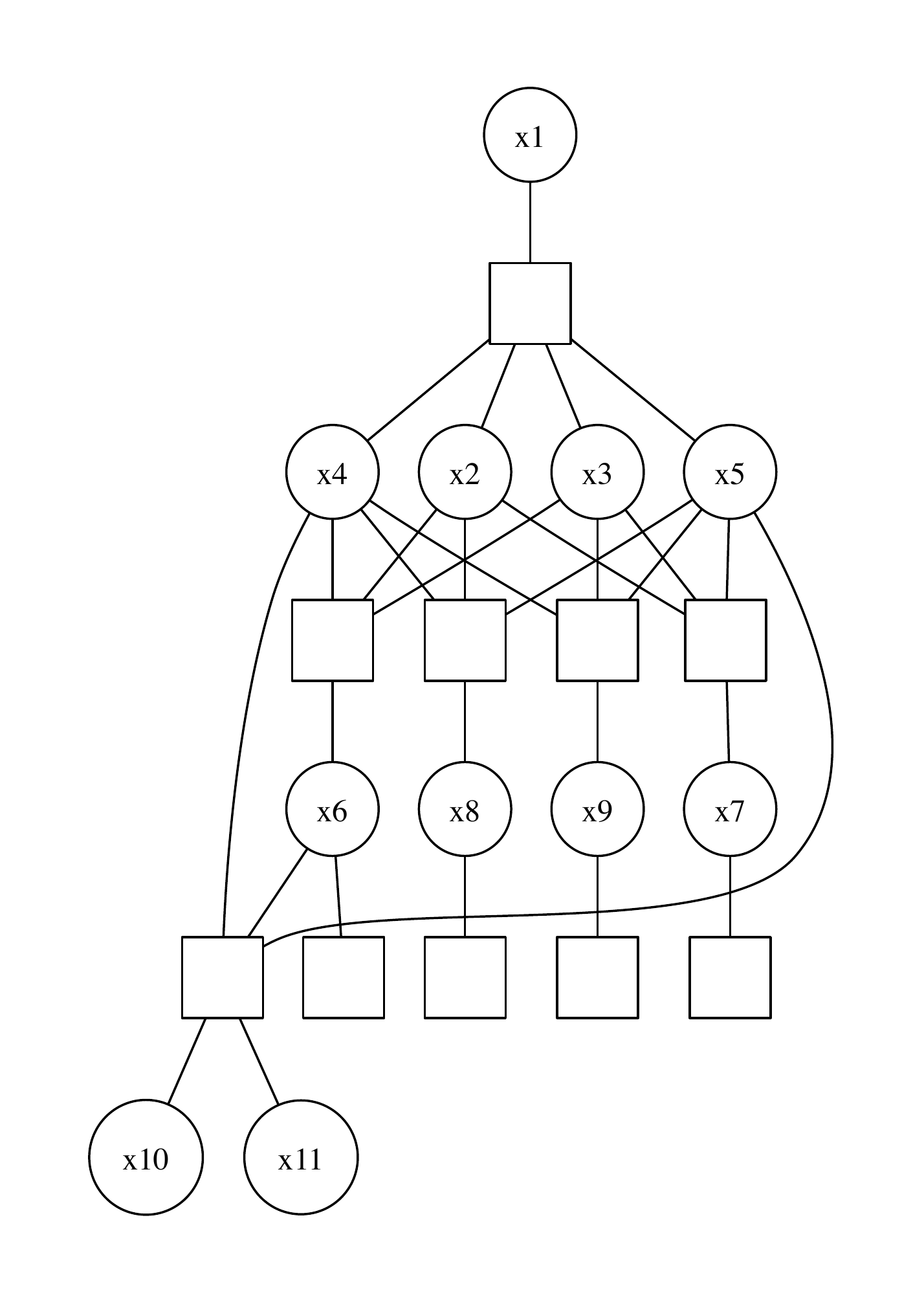}
&
\includegraphics[scale=0.4]{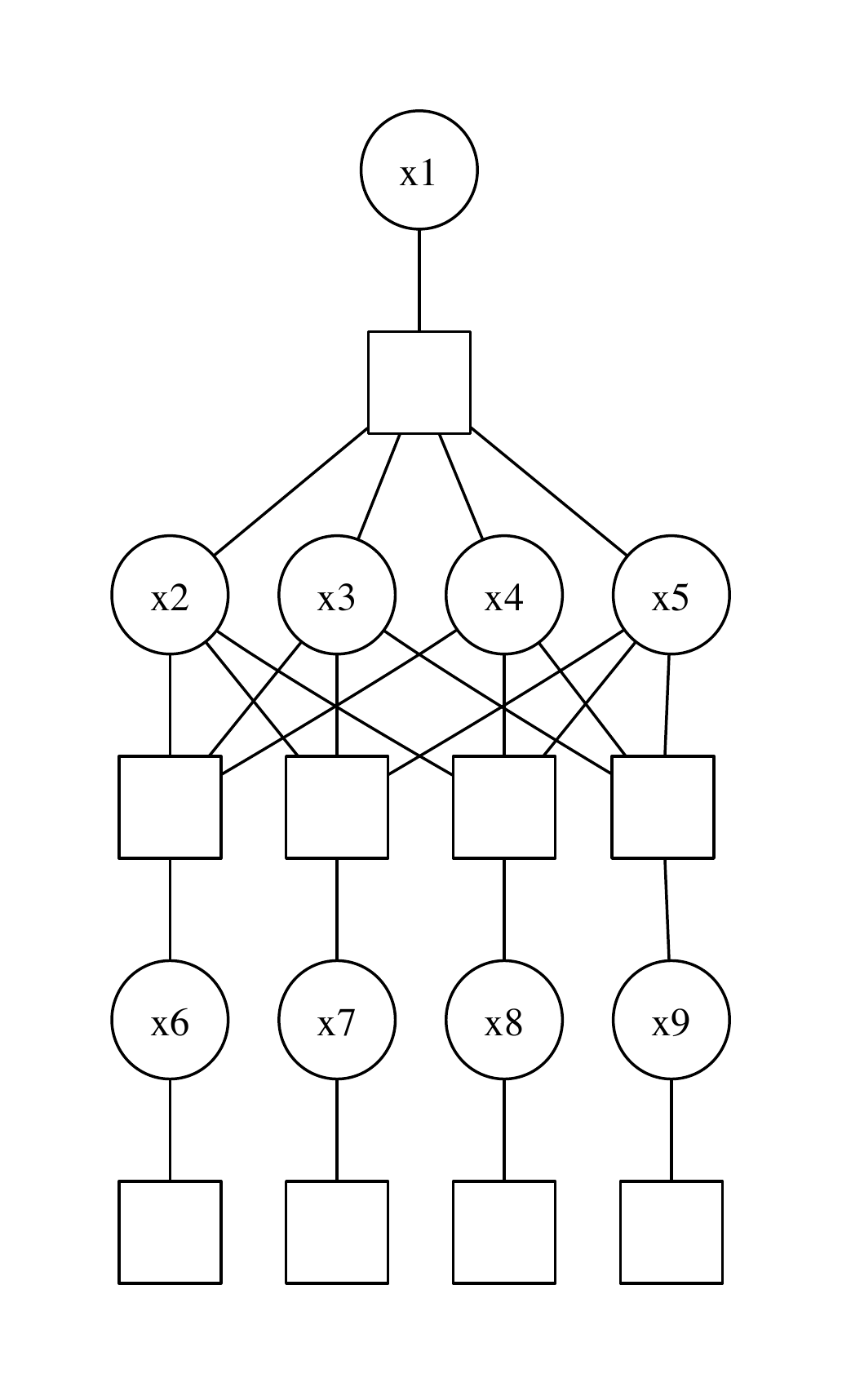}
\\
(a) a constraint graph
&
(b) subgraph of the constraint graph
\\
\\
\end{tabular}
\caption{(a) A constraint graph for an instance $\xorclauses$, (b) subgraph of the constraint graph illustrating that $\xorclauses \Models x_1 \equiv \top$}
\label{Fig:TWEx}
\end{figure}

\begin{figure}[ht]
\centering
\includegraphics[width=0.5\textwidth]{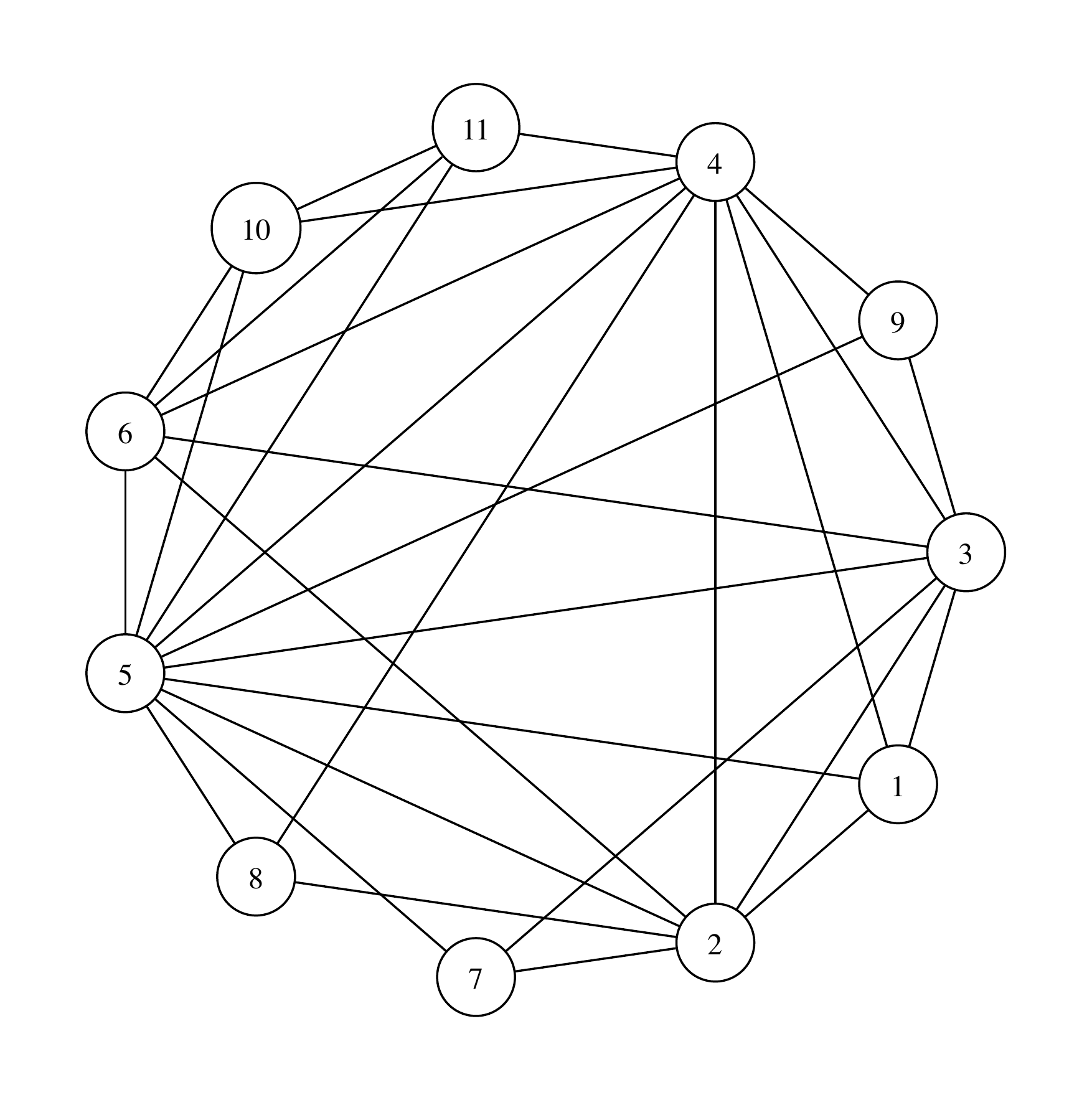}
\caption{Primal graph for the instance whose constraint graph is shown in Fig.~\ref{Fig:TWEx}(a)}
\label{Fig:TWExp}
\end{figure}

\begin{figure}[ht]
\centering
\includegraphics[width=0.5\textwidth]{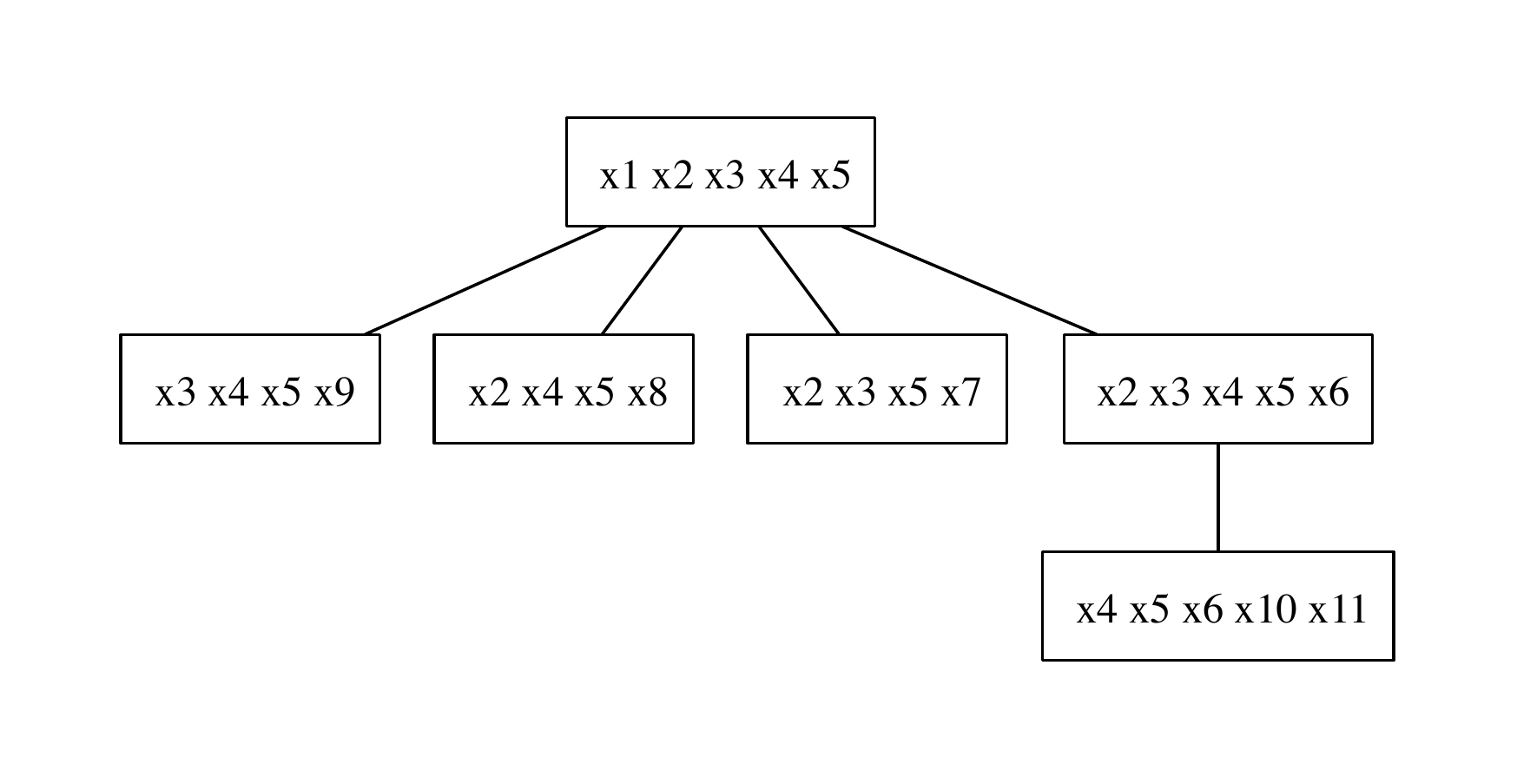}
\caption{Tree decomposition of the primal graph in Fig.~\ref{Fig:TWExp}. Assume
that $X_1 = \Set{x_1,x_2,x_3,x_4,x_5},
     X_2 = \Set{x_3,x_4,x_5,x_9},
     X_3 = \Set{x_2,x_4,x_5,x_8},
     X_4 = \Set{x_2,x_3,x_5,x_7},
     X_5 = \Set{x_2,x_3,x_4,x_5,x_6},
     X_6 = \Set{x_4,x_5,x_6,x_{10},x_{11}},  \phi_0=\xorclauses$, $\phi_1 = \phi_0 \wedge
\PropTable{X_1, \phi_0, |X_1|}$, $ \phi_2 = \phi_1
\wedge \PropTable{X_2, \phi_1, |X_2|}$, $\phi_3 = \phi_2
\wedge \PropTable{X_3, \phi_2, |X_3| }$, $ \phi_4 =
\phi_3 \wedge \PropTable{X_4, \phi_3, |X_4|}$, $\phi_5
= \phi_4 \wedge \PropTable{X_5, \phi_4, |X_5|}$, and
$\phi_6 = \phi_5 \wedge \PropTable{X_6, 
    \phi_5, |X_6|} $, and $ \psi = \phi_6 \setminus \xorclauses$. It holds that
        $\HasPropTable{X_1}{\xorclauses \wedge \psi}, \dots,
    \HasPropTable{X_6}{\xorclauses \wedge \psi}$. The \UP{} system can deduce $
        (x_1 \equiv \top)$, i.e. $ \xorclauses \wedge \psi \UPderiv (x_1 \equiv
                \top)$ by ``propagating'' intermediate linear combinations
        starting from the leaves of the tree decomposition towards the root
        node (the node with the set of variables $X_1$). Since $x_6 \equiv \top$ is in
        $\xorclauses$ it holds that $ \xorclauses \wedge \psi \UPderiv
        a_{2,3,4} \equiv \bot$. And in a similar way because $ x_7 \equiv
        \top$, $x_8\equiv\top$,and $x_9 \equiv\top$ are in $\xorclauses$, then
        $ \xorclauses \wedge \psi \UPderiv a_{2,3,5} \equiv \bot $,
    $\xorclauses \wedge \psi \UPderiv a_{2,4,5} \equiv \bot$, and $\xorclauses
        \wedge \psi \UPderiv a_{3,4,5} \equiv \bot$. By combining these
        intermediate results, it holds that $ \xorclauses \wedge \psi \UPderiv
        a_{2,3,4,5} \equiv \bot$ and finally $ \xorclauses \wedge \psi \UPderiv
        x_1 \equiv \top$}
        \label{Fig:TWTree}
\end{figure}

\begin{proof}
The construction is illustrated in Figures~\ref{Fig:TWEx}, \ref{Fig:TWExp}, \ref{Fig:TWTree}.

Let $ \psi = \phi_n \setminus \xorclauses $.
We first prove that the satisfying truth assignments of $\xorclauses$ are
exactly the ones of $\xorclauses \wedge \psi$ when projected to
$\VarsOf{\xorclauses}$. 
By Lemma~\ref{Lem:PropTableModels}, the satisfying truth assignments of $ \phi$ are exactly the ones of $ \phi \wedge \PropTable{Y, \phi, k} $ when projected
to $\VarsOf{\phi}$, so by induction the satisfying truth assignments of $
\xorclauses$ are exactly the ones of $ \xorclauses \wedge \psi$ when projected
to $ \VarsOf{\xorclauses}$.
The number of xor-constraints in $ \PropTable{Y, \phi, k} $ is $O(2^{2k}) +
|\phi|$, so the number of xor-constraints in $ \psi $ is $ O(n 2^{2k}) +
|\xorclauses| $.

It holds for each $ X_i \in \Set{X_1, \dots, X_n} $ by
Lemma~\ref{Lem:PropTableImpl} that $ \HasPropTable{X_i}{\xorclauses \wedge
\psi}$.
Next we show that if $\xorclauses$ is satisfiable and $\xorclauses \wedge \AL_1
\wedge \dots \wedge \AL_k \Models \IL$, then $\IL$ is $\UP{}$-derivable from $
\xorclauses \wedge \psi \wedge \AL_1 \wedge \dots \wedge
\AL_k$. 
Assume that $ \xorclauses$ is satisfiable and $ \xorclauses \wedge \AL_1 \wedge \dots \wedge \AL_k \Models \IL $.
\newcommand{\subtree}{\ensuremath{T'}}
We prove by induction on the structure of the tree decomposition that the
following property holds for each subtree $\subtree$ of the tree decomposition having the set of variables $X_{T'}$ and the root node of
$\subtree$ with the set of variables $X_r$:
\begin{itemize}
\item If $\phi$ is a satisfiable conjunction in $\xorclauses \wedge \psi \wedge \AL_1 \wedge \dots \wedge \AL_k $ such that
$\VarsOf{\phi} \subseteq X_{T'}$,
and $ \phi \wedge (Y_1 \equiv \parity{1}) \wedge
\dots \wedge (Y_m \equiv \parity{m}) \Models (Y' \equiv \parity{}') $ 
where $Y' \subseteq X_r $ and for each $Y_j \in \Set{Y_1,\dots,Y_n}$ there is a
$k \in \Set{1,\dots,n}$ for which it holds that $Y_j \subseteq X_k $
and
$\parity{1},\dots,\parity{m},\parity{}' \in \Set{\top,\bot}$, 
then $\xorclauses \wedge \psi \wedge (a_1 \equiv \parity{1}) \wedge \dots \wedge (a_n \equiv
        \parity{m}) \UPderiv (a' \equiv \parity{}')$
where $a_1,\dots,a_n,a'$ are the ``alias'' variables for the variable sets
$Y_1,\dots,Y_n,Y'$, respectively.
\end{itemize}
The induction hypothesis is that the property holds for each proper subtree of
$T'$.

Base case: $T'$ has only one node. The property holds by Lemma~\ref{Lem:PCProp}.

Induction step: $T'$ has more than one node.
Let $\phi' = \phi \wedge (Y_1 \equiv \parity{1}) \wedge \dots \wedge (Y_m \equiv \parity{m}) $. 
The idea is to remove xor-constraints involving variables other than in
$X_r$ from $\phi'$ and add additional xor-constraints of the type involving variables only in $X_r$.
This is done by considering each direct child node of the root node of $T'$ at
a time possibly rewriting $\phi'$ by substituting a sub-conjunction of $\phi'$ with at most one xor-constraint having only variables in $X_r$.
Let $T''$ be the subtree induced by one direct child node of the root node
having the set of variables $X_{T''}$. The per-child substitution operation of $\phi'$ is defined as follows.
\newcommand{\phiA}{\phi^{\textup{a}}}
\newcommand{\phiB}{\phi^{\textup{b}}}
Let $\phiA $ be the maximal conjunction of xor-constraints in $\phi'$ such that
$\VarsOf{\phiA} \subseteq X_{T''}$, and $\phiB$ be the conjunction of
xor-constraints in $\phi'$ but not in $\phiA$.
If $\phiA$ is empty, then nothing needs to be removed from $\phi'$.
Otherwise, $\phiA$ is non-empty and there is an $\Vars$-cut partition
$(\CutA, \CutB)$ of $\phi'$ such that $ \xorclausesA = \bigwedge_{D \in \CutA} D$, $\xorclausesB = \bigwedge_{D \in \CutB} D$ and $\VarsOf{\phi^a} \cap \VarsOf{\phi^b} = X \subseteq X_r \cap X_c $.
By Theorem~\ref{Thm:XCut}, it holds that
\begin{enumerate}
\item $\phiA \Models (Y' \equiv \parity{}') $
or $ \phiB \Models (Y' \equiv \parity{}') $;
or
\item $\phi^\alpha \Models 
( X'' \equiv \parity{}'') $ and $ \phi^\beta \wedge (X'' \equiv \parity{}'') \Models (Y' \equiv \parity{}') $ for some $X'' \subseteq X$,
$\parity{}' \in \set{\top,\bot}$, $\alpha \in \Set{\textup{a},\textup{b}}$, and $\beta \in \Set{\textup{a},\textup{b}}\setminus\Set{\alpha}$.
%\item $\xorclauses^a \wedge \AL_1 \wedge \dots \wedge \AL_k \Models 
%( X' \equiv \parity{}) $ and $ \xorclauses^b \wedge \AL_1 \wedge \dots \wedge \AL_k \wedge (X' \equiv \parity{}) \Models \IL $ for some $ X' \subseteq X$ and $\parity{} \in \set{\top,\bot}$.
%\item $\xorclauses^b \wedge \AL_1 \wedge \dots \wedge \AL_k \Models 
%( X' \equiv \parity{}) $ and $ \xorclauses^a \wedge \AL_1 \wedge \dots \wedge \AL_k \wedge (X' \equiv \parity{}) \Models \IL $ for some $ X' \subseteq X$ and $\parity{} \in \set{\top,\bot}$.
\end{enumerate}

We analyze the cases:
\begin{itemize}
\item[] Case 1: $\phiA \Models (Y' \equiv \parity{}') $
or $ \phiB \Models (Y' \equiv \parity{}') $. Since $Y' \subseteq X_r$, it must be that $\phiB \Models (Y' \equiv \parity{}') $. In this case, set $\phi' \leftarrow \phiB $.
\item[] Case 2: $\phi^\alpha \Models 
( X'' \equiv \parity{}'') $ and $ \phi^\beta \wedge (X'' \equiv \parity{}'') \Models (Y' \equiv \parity{}') $ for some $X'' \subseteq X$,
$\parity{}' \in \set{\top,\bot}$, $\alpha \in \Set{\textup{a},\textup{b}}$, and $\beta \in \Set{\textup{a},\textup{b}}\setminus\Set{\alpha}$. Again since $Y' \subseteq X_r$, it must be that $\alpha = \textup{a}$. In this case, set $ \phi' \leftarrow \phiB \wedge (X'' \equiv \parity{}'') $.
\end{itemize} 

After each child node has been processed in this way, it holds that $\phi'
\Models (Y \equiv \parity{}) $ and $\VarsOf{\phi'} \subseteq X_r $.
It also holds by the induction hypothesis for each xor-constraint $ (X_i'' \equiv \parity{i}'') $ in the
sequence $ (X_1'' \equiv \parity{1}''), \dots, (X_q'' \equiv \parity{q}'') $ of
added xor-constraints that the corresponding ``alias'' variables $a_1'', \dots,
      a_q''$ for $X_1'', \dots, X_q''$, respectively, that,
      since $ \phi \wedge (Y_1 \equiv \parity{1}) \wedge \dots \wedge (Y_m \equiv \parity{m}) \wedge (X_{1}'' \equiv \parity{1}'') \wedge \dots \wedge (X_{i-1}'' \equiv \parity{i-1}'') \Models (X_{i}'' \equiv \parity{i}'')$, then it holds
      that $\xorclauses \wedge \psi \wedge \AL_1 \wedge \dots \wedge \AL_k \wedge (a_1 \equiv \parity{1}) \wedge \dots \wedge (a_m \equiv \parity{m}) \wedge (a_1'' \equiv \parity{1}'') \wedge \dots \wedge (a_{i-1}'' \equiv \parity{i-1}'') \UPderiv a_i \equiv \parity{i}''$. 
Now $\phi' \Models (Y \equiv \parity{}) $, $\VarsOf{\phi'} \subseteq X_r$,
and each xor-constraint $ (X'' \equiv \parity{}'') $ in $\phi'$ has its corresponding ``alias'' variable $ a'' $ implied by unit propagation, that is, $\xorclauses \wedge \psi \wedge \AL_1 \wedge \dots \wedge \AL_k \UPderiv (a' \equiv \parity{}') $.
By induction it follows that $\xorclauses \wedge \psi \wedge \AL_1 \wedge \dots \wedge
\AL_k \UPderiv \IL $.

It remains to show that if $\xorclauses$ is unsatisfiable,
then $\psi \UPderiv (\bot \equiv \top)$. Assume that 
$\xorclauses$ is unsatisfiable. 
By Lemma~\ref{Lemma:LinearCombs}, there is a minimal subset $S$ of
xor-constraints in $ \xorclauses \wedge \AL_1 \wedge \dots \wedge \AL_k $
such that $ \SetLC{S} = (\bot \equiv \top) $. 
Now, let $ S' = S \setminus \XC $
be a subset of $S$ identical to $S$ except that one xor-constraint $\XC = (X' \equiv \parity{}) $ in $S$ is removed.
It clearly holds that $S'$ is satisfiable and $ \SetLC{S'} = (X' \equiv \parity{} \oplus \top) $. 
There is a node in $T'$ that has the variables $X_i$ such that $ X' \subseteq X_i$.
It holds that $ \HasPropTable{X_i}{\xorclauses \wedge \psi}$, so by Lemma~\ref{Lem:PTableProp2} it holds for the ``alias'' variable $a' $ for $X'$ that 
$\xorclauses \wedge \psi \wedge \AL_1 \wedge \dots \wedge \AL_k \UPderiv (a' \equiv \parity{} \oplus \top) $.
Repeat the proof as above for the satisfiable case and for the subset $S'$ showing
that $ \xorclauses \wedge \psi \wedge \AL_1 \wedge \dots \wedge \AL_k \UPderiv 
(a' \equiv \parity{}) $.
It follows that $ 
\xorclauses \wedge \psi \wedge \AL_1 \wedge \dots \AL_k \UPderiv (\bot \equiv \top) $.
All the requirements for GE-simulation formula are satisfied, so 
$\psi$ is a GE-simulation formula for $\xorclauses$.
\end{proof}

%% file: paper.bbl
\begin{thebibliography}{10}

\bibitem{Handbook:CDCL}
Marques-Silva, J., Lynce, I., Malik, S.:
\newblock Conflict-driven clause learning {SAT} solvers.
\newblock In: Handbook of Satisfiability.
\newblock IOS Press (2009)

\bibitem{tseitin}
Tseitin, G.S.:
\newblock On the complexity of derivations in the propositional calculus.
\newblock Studies in Mathematics and Mathematical Logic \textbf{Part II} (1968)
   115--125

\bibitem{Urquhart:JACM1987}
Urquhart, A.:
\newblock Hard examples for resolution.
\newblock Journal of the ACM \textbf{34}(1) (1987)  209--219

\bibitem{PipatsrisawatDarwiche:AI2011}
Pipatsrisawat, K., Darwiche, A.:
\newblock On the power of clause-learning {SAT} solvers as resolution engines.
\newblock Artificial Intelligence \textbf{175}(2) (2011)  512--525

\bibitem{Li:AAAI2000}
Li, C.M.:
\newblock Integrating equivalency reasoning into {Davis-Putnam} procedure.
\newblock In: Proc.~AAAI/IAAI 2000, AAAI Press (2000)  291--296

\bibitem{Li:IPL2000}
Li, C.M.:
\newblock Equivalency reasoning to solve a class of hard {SAT} problems.
\newblock Information Processing Letters \textbf{76}(1--2) (2000)  75--81

\bibitem{BaumgartnerMassacci:CL2000}
Baumgartner, P., Massacci, F.:
\newblock The taming of the {(X)OR}.
\newblock In: Proc.~CL~2000. Volume 1861 of LNCS., Springer (2000)  508--522

\bibitem{Li:DAM2003}
Li, C.M.:
\newblock Equivalent literal propagation in the {DLL} procedure.
\newblock Discrete Applied Mathematics \textbf{130}(2) (2003)  251--276

\bibitem{HeuleMaaren:SAT2004}
Heule, M., van Maaren, H.:
\newblock Aligning {CNF}- and equivalence-reasoning.
\newblock In: Proc.\ SAT 2004. Volume 3542 of LNCS., Springer (2004)  145--156

\bibitem{HeuleEtAl:SAT2004}
Heule, M., Dufour, M., van Zwieten, J., van Maaren, H.:
\newblock March\_eq: Implementing additional reasoning into an efficient
  look-ahead {SAT} solver.
\newblock In: Proc.\ SAT 2004. Volume 3542 of LNCS., Springer (2004)  345--359

\bibitem{Chen:SAT2009}
Chen, J.:
\newblock Building a hybrid {SAT} solver via conflict-driven, look-ahead and
  {XOR} reasoning techniques.
\newblock In: Proc.\ SAT 2009. Volume 5584 of LNCS., Springer (2009)  298--311

\bibitem{SoosEtAl:SAT2009}
Soos, M., Nohl, K., Castelluccia, C.:
\newblock Extending {SAT} solvers to cryptographic problems.
\newblock In: Proc.\ SAT 2009. Volume 5584 of LNCS., Springer (2009)  244--257

\bibitem{LJN:ECAI2010}
Laitinen, T., Junttila, T., Niemel{\"a}, I.:
\newblock Extending clause learning {DPLL} with parity reasoning.
\newblock In: Proc.\ ECAI 2010, IOS Press (2010)  21--26

\bibitem{Soos}
Soos, M.:
\newblock Enhanced gaussian elimination in {DPLL}-based {SAT} solvers.
\newblock In: Pragmatics of SAT, Edinburgh, Scotland, GB (July 2010)  1--1

\bibitem{LJN:ICTAI2011}
Laitinen, T., Junttila, T., Niemel{\"a}, I.:
\newblock Equivalence class based parity reasoning with {DPLL(XOR)}.
\newblock In: Proc.\ ICTAI 2011, IEEE (2011)  649--658

\bibitem{LJN:SAT2012}
Laitinen, T., Junttila, T., Niemel{\"a}, I.:
\newblock Conflict-driven {XOR}-clause learning.
\newblock In: Proc.\ SAT 2012. Volume 7317 of LNCS., Springer (2012)  383--396

\bibitem{LJN:CP2012}
Laitinen, T., Junttila, T., Niemel{\"a}, I.:
\newblock Classifying and propagating parity constraints.
\newblock In: Proc.\ CP 2012. Volume 7514 of LNCS., Springer (2012)  357--372

\bibitem{LJN:ICTAI2012}
Laitinen, T., Junttila, T., Niemel\"a, I.:
\newblock Extending clause learning {SAT} solvers with complete parity
  reasoning.
\newblock In: Proc.\ ICTAI 2012, IEEE (2012)

\bibitem{Weaver:2012:SAE:2520447}
Weaver, S.A.:
\newblock Satisfiability advancements enabled by state machines.
\newblock PhD thesis, Cincinnati, OH, USA (2012) AAI3554401.

\bibitem{NieuwenhuisEtAl:JACM06}
Nieuwenhuis, R., Oliveras, A., Tinelli, C.:
\newblock Solving {SAT} and {SAT} modulo theories: From an abstract
  {D}avis-{P}utnam-{L}ogemann-{L}oveland procedure to {DPLL(T)}.
\newblock Journal of the ACM \textbf{53}(6) (2006)  937--977

\bibitem{Kullmann:Sep2013}
Gwynne, M., Kullmann, O.:
\newblock On {SAT} representations of {XOR} constraints.
\newblock arXiv document arXiv:1309.3060 [cs.CC] (2013)

\bibitem{Handbook:SMT}
Barrett, C., Sebastiani, R., Seshia, S.A., Tinelli, C.:
\newblock Satisfiability modulo theories.
\newblock In: Handbook of Satisfiability.
\newblock IOS Press (2009)

\bibitem{HanJiang:CAV2012}
Han, C.S., Jiang, J.H.R.:
\newblock When boolean satisfiability meets gaussian elimination in a simplex
  way.
\newblock In: Proc.\ CAV 2012. Volume 7358 of LNCS., Springer (2012)  410--426

\bibitem{ZhangMalik:DATE2003}
Zhang, L., Malik, S.:
\newblock Validating {SAT} solvers using an independent resolution-based
  checker: Practical implementations and other applications.
\newblock In: Proc.\ DATE 2003, IEEE (2003)  880--885

\bibitem{BeameKautzSabharwal:JAIR2004}
Beame, P., Kautz, H., Sabharwal, A.:
\newblock Towards understanding and harnessing the potential of clause
  learning.
\newblock Journal of Artificial Intelligence Research \textbf{22} (2004)
  319--351

\bibitem{EenSorensson:2004}
E\'en, N., S\"orensson, N.:
\newblock An extensible {SAT} solver.
\newblock In: Proc.~SAT 2003. Volume 2919 of LNCS., Springer (2004)  502--518

\bibitem{Freuder:1990:CKS:1865499.1865500}
Freuder, E.C.:
\newblock Complexity of k-tree structured constraint satisfaction problems.
\newblock In: Proc.\ AAAI 1990, AAAI Press (1990)  4--9

\bibitem{Arnborg:1987:CFE:37170.37183}
Arnborg, S., Corneil, D.G., Proskurowski, A.:
\newblock Complexity of finding embeddings in a k-tree.
\newblock SIAM J. Algebraic Discrete Methods \textbf{8}(2) (April 1987)
  277--284

\bibitem{DBLP:conf/aaai/Pearl82}
Pearl, J.:
\newblock Reverend {Bayes} on inference engines: A distributed hierarchical
  approach.
\newblock In: Proc.\ AAAI 1982, AAAI Press (1982)  133--136

\bibitem{LJN:ICTAI2012full}
Laitinen, T., Junttila, T., Niemel\"a, I.:
\newblock Extending clause learning {SAT} solvers with complete parity
  reasoning (extended version).
\newblock arXiv document arXiv:1207.0988 [cs.LO] (2012)

\end{thebibliography}
